\documentclass[12pt, a4paper]{article}

\usepackage{amsmath,amsfonts,amsthm,amssymb}
\usepackage[retainorgcmds]{IEEEtrantools}
\usepackage{tikz}

\usepackage[authoryear]{natbib}

\usepackage{fourier}

\theoremstyle{plain}
\newtheorem{theorem}{Theorem}[section]

\newtheorem{application}[theorem]{Application}
\newtheorem{lemma}[theorem]{Lemma}

\newtheorem{claim}[theorem]{Claim}
\newtheorem{corollary}[theorem]{Corollary}

\theoremstyle{definition}
\newtheorem{example}{Example}
\newtheorem{definition}[theorem]{Definition}
\newtheorem{remark}[theorem]{Remark}

\DeclareMathOperator*{\argmax}{arg\,max}

\newcommand{\mysetminusD}{\hbox{\tikz{\draw[line width=0.6pt,line cap=round] (3pt,0) -- (0,6pt);}}}
\newcommand{\mysetminusT}{\mysetminusD}
\newcommand{\mysetminusS}{\hbox{\tikz{\draw[line width=0.45pt,line cap=round] (2pt,0) -- (0,4pt);}}}
\newcommand{\mysetminusSS}{\hbox{\tikz{\draw[line width=0.4pt,line cap=round] (1.5pt,0) -- (0,3pt);}}}
\newcommand{\mysetminus}{\mathbin{\mathchoice{\mysetminusD}{\mysetminusT}{\mysetminusS}{\mysetminusSS}}}

\newcommand{\mms}{\boldsymbol{\mu}}

\allowdisplaybreaks

\title{Truthful Allocation Mechanisms Without Payments: Characterization and Implications on Fairness\thanks{A conference version to appear in the 18th ACM conference on Economics and Computation (ACM EC '17).}}  
\author{Georgios Amanatidis\thanks{Athens University of Economics and Business. Emails: \texttt{\{gamana, gebirbas, markakis\}@aueb.gr}} 
\and
Georgios Birmpas\footnotemark[2] 
\and
George Christodoulou\thanks{University of Liverpool. Email: \texttt{G.Christodoulou@liverpool.ac.uk}} 
\and
Evangelos Markakis\footnotemark[2]
}

\begin{document}
\maketitle
\begin{abstract}
We study the mechanism design problem of allocating a set of indivisible items without monetary transfers. Despite the vast literature on this very standard model, it still remains unclear how do truthful mechanisms look like. We focus on the case of two players with additive valuation functions and our purpose is twofold. 
First, our main result provides a complete characterization of truthful mechanisms that allocate all the items to the players. Our characterization reveals an interesting structure underlying all truthful mechanisms, showing that they can be decomposed into two components: a \textit{selection part} where players pick their best subset among prespecified choices determined by the mechanism, and an \textit{exchange part} where players are offered the chance to exchange certain subsets if it is favorable to do so. 
In the remaining paper, we apply our main result and derive several consequences on the design of mechanisms with fairness guarantees. We consider various notions of fairness, (indicatively, maximin share guarantees and envy-freeness up to one item) and provide tight bounds for their approximability. Our work settles some of the open problems in this agenda, and we conclude by discussing possible extensions to more players. 
\end{abstract}

%
%

%




%
\section{Introduction} \label{sec:intro}
We study a very elementary and fundamental model for allocating
indivisible goods from a mechanism design viewpoint.  Namely, we
consider a set of indivisible items that need to be
allocated to a set of players. An outcome of the problem is an
allocation of all the items to the players, i.e., a partition into
bundles, and each player evaluates an allocation by his own additive valuation
function.  Our primary motivation originates from the fair division
literature, where such models have been considered
extensively. However, the same setting also appears in several domains,
including job scheduling, load balancing and many other resource
allocation problems.

The focus of our work is on understanding the interplay between
truthfulness and fairness in this setting. Hence, we want to identify
the effects on fairness guarantees, imposed by eliminating any
incentives for the players to misreport their valuation functions.
This type of questions has been posed already in previous works and
for various notions of fairness, such as envy-freeness, or for the
concept of maximin shares \protect\citep[see, among others,][]{LMMS04,CKKK09,ABM16}. 
However, the results so far have been
rather scarce in the sense that a) in most cases, they concern
impossibility results which are far from being tight
and b) the proof
techniques are based on constructing specific families of instances
that do not enhance our understanding on the structure of truthful
mechanisms, with the exception of \protect\citet{CKKK09} which, however, is only
for two players and two items.
   
In order to comprehend the trade-offs that are inherent between
incentives and fairness, we first take a step back and focus solely on
truthfulness itself. 
As is quite common in fair division models, we will not allow any
monetary transfers, so that a mechanism simply outputs an allocation
of the items. Hence, the question we want to begin with is: {\em what
  is the structure of truthful allocation mechanisms?}

There has been already a significant volume of works on characterizing
truthful allocation mechanisms for indivisible items,
yet there are
some important differences
from our approach.  First, a typical line of work studies this
question under the additional assumption of Pareto efficiency or
related notions \protect\citep{Papai00,KM02,EK03}. The characterization results
that have been obtained show that the combination of truthfulness
together with Pareto efficiency tends to make the class of available
deterministic mechanisms very poor; only some types of dictatorship survive when
imposing both criteria.  Second, in some cases the analysis is carried
out without any restrictions on the class of valuation functions,
i.e., not even monotonicity, 
which again often results in a very
limited class of mechanisms \protect\citep[see, e.g.,][]{Papai01}. When moving to
a specific class, such as the class of additive functions which is
usualy assumed in fair division, it is conceivable that we can have a
much richer class of truthful mechanisms.   
The results above indicate that the known characterizations of
truthful mechanisms are also dependent on further assumptions, which
may be well justified in various scenarios, but they  are not 
aligned with the goal of fair division.

\subsection{Contribution}
Our main result is a characterization of deterministic truthful
mechanisms that allocate all the items to two players with additive
valuations. In doing so, we identify some important allocation
properties that every truthful mechanism should satisfy.
One such crucial property is the notion of \emph{controlling} items
(Definition \ref{def:control}); we say that a
player controls an item, whenever it is possible to report values that will guarantee him this item, regardless of the other player's valuation function.  We
show that truthfulness implies that every item is controlled by some
player.
Exploiting this property
further, greatly helps us in understanding how a mechanism operates.
Consequently, our analysis and the characterization we eventually
obtain reveals an interesting structure underlying all truthful
mechanisms; they can all be essentially decomposed into two
components: (i) a {\em selection part} where players pick their best
subset among prespecified choices determined by the mechanism, and
(ii) an {\em  exchange part} where players are offered the chance to
exchange certain subsets if it is favorable to do so. Hence, we call
them {\em picking-exchange mechanisms}.

Next, we apply our main result and derive several consequences on the
design of mechanisms with (approximate) fairness guarantees. 
We consider various notions of fairness in Section \ref{SEC:FAIRNESS},
starting our discussion with the more standard ones such as
proportionality and envy-freeness, and explaining why such concepts
cannot be attained---even approximately---by truthful mechanisms. We then
focus on more recently studied relaxations of either envy-freeness or
proportionality where positive algorithmic results have been obtained
(e.g., finding allocations that are envy-free up to one item, or
achieve approximate maximin share guarantees). For these notions, we
provide tight bounds on the approximation guarantees of truthful
mechanisms, settling some of the open problems in this
area \protect\citep{CKKK09,ABM16}. Interestingly, our results also reveal that the best truthful
approximation algorithms for fair division are achieved by {\it
  ordinal} mechanisms, i.e., mechanisms that exploit only the relative
ranking of the items and not the cardinal information of the valuation
functions.

The heart of our approach for obtaining lower bounds on the approximability of fairness criteria, is a necessary condition for fairness 
in view of our notion of control, which we call {\em no control of pairs}. It
 states that no player should control more than one
item. We show how this condition summarizes minimum 
requirements for various fairness concepts  previously studied
in the literature. Although this condition does not offer an alternative fairness criterion, it is a useful tool for
showing lower bounds.
 
Finally, in Section~\ref{sec:further_dir} we provide a general class
of truthful mechanisms for the case of multiple players. This class
generalizes picking-exchange mechanisms in a non-trivial way.
As indicated by our mechanisms, there is a much richer structure in the case of multiple players. In particular, the notion of control does not convey enough information anymore. Instead, there seem to exist several different levels of control.


\subsection{Related Work}\label{subsec:relwork}
The only work we are aware of, in which a full characterization is
given for truthful mechanisms with indivisible items, additive
valuations, and no further assumptions is by \protect\citet{CKKK09}.
However, this is only a characterization  for two players and two
items.  Apart from characterizations, there have been several works
that try to quantify the effects of truthfulness on several
concepts of fairness. For the performance of truthful mechanisms
with respect to envy-freeness, see \protect\citet{CKKK09} and \protect\citet{LMMS04}, whereas for max-min
fairness see \protect\citet{BD05}. Coming to more recent results and along the
same spirit, \protect\citet{ABM16} and \protect\citet{MP11} study the notion of maximin
share allocations, and a related notion of worst-case guarantees
respectively. They obtain separation results, showing that the
approximation factors achievable by truthful mechanisms are strictly
worse than the known algorithmic (nontruthful) results.  Obtaining a
better understanding for the structure of truthful mechanisms and how
they affect fairness has been an open problem underlying all the above
 works. For a better and more complete elaboration on
fairness and the numerous fairness concepts that have been suggested,
we refer the reader to the books \protect\citep{BT96,RW98,M03} and the recent
surveys \protect\cite{BCM16-survey,Procaccia16-survey}.

There has been a long series of works on characterizing mechanisms
with indivisible items beyond the context of fair divison.  Many works
characterize the allocation mechanisms that arise when we combine truthfulness
with Pareto efficiency \protect\citep[see, e.g.,][]{Papai00,KM02,EK03}. Typically,
such mechanisms tend to be dictatorial, and it is also well known that
economic efficiency is mostly incompatible with fairness \protect\citep[see,
e.g.,][]{BCM16-survey}.  Another assumption that has been used is
nonbossiness, which means that one cannot change the outcome without
affecting his own bundle. For instance, \protect\citet{Svensson99}  assumes nonbossiness in a setting where each
player is interested in acquiring only one item. For general valuations, this also leads to
dictatorial algorithms \protect\citep{Papai01}.   
In most of these works ties are
ignored by considering strict preference orders over all subsets of
the items, while in some cases it is
also allowed for the mechanism not to allocate all the items.

There have also been relevant works for the setting of divisible goods \protect\citep[see, among others,][]{CLPP13,ColeGG13}. We note that for additive valuation functions, a mechanism for divisible items can be interpreted as a randomized mechanism for indivisible items. This connection is already discussed and explored in~\protect\citet{GuoC10,AzizFCMM16}. In our work, we do not study randomized mechanisms, however it is an interesting question to have characterization results for such settings as well. Along this direction, see \protect\citet{MennleS14} where a relaxed notion of truthfulness is studied.

Related to our work is also the literature on exchange markets. These
are models where players are equipped with an initial endowment,
e.g., a house or a set of items. 
For the case where players can
have multiple indivisible items as their initial endowment, see
\protect\citet{Papai03,Papai07}. Exchange markets provide an example where the existing characterizations go well beyond dictatorships and are closely related to the exchange component of our mechanisms.

Finally, for settings with payments, the work of \protect\citet{DS08}, and independently of \protect\citet{CKV08}, provided a characterization of truthful mechanisms with two players and additive valuations when all items are allocated. However, their characterization does not
apply to our setting because they make an additional assumption,
namely {\em decisiveness}. It roughly requires that each player should
be able to receive any possible bundle of items, by making an
appropriate bid. Their motivation is the characterization of truthful
mechanisms with bounded makespan (maximum finishing time) for the
scheduling problem, and in their case decisiveness is necessary in
order to achieve bounded guarantees. In our case, our motivation is
fairness, and decisiveness is a very strong assumption which has the
opposite effects of what we need; e.g., assigning the full-bundle to a
player is unacceptable in terms of fairness. Finally,
\protect\citet{ChristodoulouK11} give a global characterization of envy-free and truthful
mechanisms for settings with payments, when there are multiple players
but only two items.

\section{Preliminaries and Notation} \label{sec:notation}

With the exception of Section \ref{sec:further_dir}, we consider a setting with two players and a set of $m$ indivisible items, $M = \{1, \ldots, m\} = [m]$, to be allocated to the players. We assume that each player $i$ has an additive valuation function $v_i$ over the items, so that for every $S\subseteq M$, $v_i(S) = \sum_{j\in S} v_i(\{j\})$. For $j\in M$, we write $v_{ij}$ instead of $v_i(\{j\})$.

We say that $(S_1, S_2, \ldots, S_k)$ is a \emph{partition} of a set $S$, if $\bigcup_{i\in [k]} S_i = S$, and $S_i \cap S_j = \emptyset$ for any $i, j\in [k]$ with $i\neq j$. Note that
we do not require that $S_i\neq \emptyset$ for all $i\in [n]$. 
An allocation of $M$ to the players is a partition in the form $S=(S_1, S_2)$. By $\mathcal{M}$ we denote the set of all allocations of $M$.

The set $\mathcal{V}_{\! m}$ of all possible profiles is $\mathbb{R}^m_+ \times \mathbb{R}^m_+$, i.e., we assume that $v_{ij}>0$ for every $i\in\{1, 2\}$ and $j\in M$. 
For some statements we  need the assumption that the players' valuation functions are such that no two sets have the same value. So, let $\mathcal{V}^{\neq}_{\! m}$ denote the set of such profiles, i.e.,
\[ \mathcal{V}^{\neq}_{\! m} = \bigg\lbrace (v_1, v_2)\in \mathcal{V}_{\! m} \ \Big| \ \forall S, T\subseteq [m] \text{ with } S\neq T, \text{ and } \forall i\in\{1, 2\},\,\sum_{j\in S}v_{ij} \neq \sum_{j\in T}v_{ij} \bigg\rbrace  \,. \]

\begin{definition}
A \emph{deterministic allocation mechanism with no monetary transfers}, or simply a \emph{mechanism}, for allocating all the items in $M = [m]$,  is a mapping $\mathcal{X}$ from $\mathcal{V}_{\! m}$ to $\mathcal{M}$. That is, for any profile $\mathbf{v}$, the outcome of the mechanism is  $\mathcal{X}(\mathbf{v})=(X_1(\mathbf{v}), X_2(\mathbf{v}))\in\mathcal{M}$, 
and $X_i(\mathbf{v})$ denotes the set of items player $i$ receives.  

A mechanism $\mathcal{X}$ is {\it truthful} if for any instance $\mathbf{v} = (v_1, v_2)$, any player $i\in\{1, 2\}$, and any $v_i'$:
\[ v_i(X_i(\mathbf{v})) \geq v_i(X_i(v_{i}', v_{-i})).\]
\end{definition}

Since we will repeatedly argue about intersections of $X_i(\mathbf{v})$ with various subsets of $M$, we use  $X_i^S(\mathbf{v})$ as a shortcut for $X_i(\mathbf{v}) \cap S$, where $S\subseteq M$.

\subsection{Fairness concepts\protect\footnote{The material of this subsection is needed in the sequel only within Section \ref{SEC:FAIRNESS}.}}
Several notions have emerged throughout the years as to what can be considered a fair allocation. We define below the concepts that we will examine in Section \ref{SEC:FAIRNESS}. Although all concepts can be clearly defined for any number of players, we provide the definitions for two players, since this is the focus of the paper. 

We start with two of the most dominant solution concepts in fair division, namely proportionality and envy-freeness. 
\begin{definition}
	An allocation $S = (S_1,S_2)$ is
		\item[{\qquad \it 1.}]  \emph{proportional}, if $v_i(S_i) \geq \frac{1}{2} v_i(M)$, for $i\in \{1, 2\}$.
		\item[{\qquad \it 2.}]  \emph{envy-free}, if $v_1(S_1) \geq v_1(S_2)$, and $v_2(S_2) \geq v_2(S_1)$ .
\end{definition}


Proportionality was considered in the very first work on fair division by \citet{Steinhaus48}. Envy-freeness was suggested later by \citet{GS58}, and with a more formal argumentation by \citet{Foley67} and \citet{Varian74}.

Envy-freeness is a stricter notion than proportionality, 
but even for the latter existence cannot be guaranteed under indivisible goods.
One can also consider approximation versions of these problems as follows: Given an instance $I$, let $E(I)$ be the minimum possible envy that can be achieved at $I$, among all possible allocations. We say that a mechanism achieves a $\rho$-approximation, if for every instance $I$, it produces an allocation where the envy between any pair of players is at most $\rho E(I)$. Similarly for proportionality, suppose that an instance $I$ admits an allocation where every player receives a value of at least $\frac{c(I)}{2} v_i(M)$ for some $c(I)\leq 1$. Then a $\rho$-approximation would mean that each player is guaranteed a bundle with value at least $\frac{\rho c(I)}{2}v_i(M)$. 

Apart from the approximation versions, the fact that we cannot always have proportional or envy-free allocations gives rise to relaxations of these definitions, with the hope of obtaining more positive results. We describe below three such relaxations, all of which admit either exact or constant-factor approximation algorithms (not necessarily truthful) in polynomial time.  

The first such relaxation is the concept of envy-freeness up to one item, where each person may envy another player by an amount which does not exceed the value of a single item in the other player's bundle. Formally:  
\begin{definition}
	\label{def:uptoone}
	An allocation $S = (S_1,S_2)$ is \emph{envy-free up to one item}, if there exists an item $a_1\in S_1$, and an item $a_2\in S_2$, such that  $v_i(S_i) \geq v_i(S_j\setminus \{a_j\})$, for  $i, j\in\{1, 2\}$. 
\end{definition}

It is quite easy to achieve envy-freeness up to one item, e.g., a round-robin algorithm that alternates between the players and gives them in each step their best remaining item suffices. Other algorithms are also known to satisfy this criterion \citep[see][]{LMMS04}.

A more interesting relaxation from an algorithmic point of view, comes from the notion of maximin share guarantees, recently proposed by~\citet{Budish11}. For two players, the maximin share of a player $i$ is the value that he could achieve by being the cutter in a discretized form of the cut and choose protocol. This is a guarantee for player $i$, if he would partition the items into two bundles so as to maximize the value of the least valued bundle. We define below the approximate version of this notion.  
Recall that $\mathcal{M}$ is the set of all allocations of $M$.
\begin{definition}\label{def:mms}
	Given a set of items $[m]$, the \emph{maximin share} of a player $i\in\{1, 2\}$, is
\[ \mms_i = \displaystyle\max_{S\in\mathcal{M}} \min \{v_i(S_1), v_i(S_2)\} \, .\]
For $\rho\leq 1$, an allocation $S = (S_1,S_2)$ is called a $\rho$-approximate maximin share allocation if $v_i(S_i)\geq \rho\cdot\mms_i\,$, for $i\in \{1, 2\}$.	
\end{definition}

For two players maximin share allocations always exist and even though they are NP-hard to compute, we  have a PTAS by reducing this to standard job scheduling problems. Hence each player can receive a value of at least $(1-\epsilon)\mms_i$. For a higher number of players, constant factor approximation algorithms also exist \citep[see][]{PW14,AMNS15}.

Finally, a related approach was undertaken by \citet{hill87}. This work examined what is the worst case guarantee that a player can have as a function of the total number of players and the maximum value of an item across all players. 
For two players, the following function was identified precisely as the guarantee that can be given to each player. Note that the total value of the items is normalized to 1 in this case.

\begin{definition}\label{def:Vn}
	Let $V_2 : \left[ 0,1 \right] \to \left[ 0 , 1/2\right]$
	be the unique nonincreasing function satisfying $V_2(\alpha) = 1/2$ for $\alpha=0$, whereas for $\alpha >0$: 
	\[ 
	V_2(\alpha) = 
	\left\{
	\begin{array}{ll}
	1 - k\alpha  & \mbox{if } \alpha \in I(2,k) \\
	1 - \frac{(k+1)}{2(k+1)-1} & \mbox{if } \alpha \in NI(2,k)
	\end{array}
	\right.
	\]
	where for any integer $k\geq1$, $I(2,k) = \left[ \frac{k+1}{k(2(k+1)-1)} , \frac{1}{2k-1} \right]$ and $NI(2,k) = \left(  \frac{1}{2(k+1)-1} ,\frac{k+1}{k(2(k+1)-1)} \right)$.
\end{definition}


\citet{MP11} proved that for two players, there always exists an allocation such that each player $i$ receives at least $V_2(\alpha_i)$, where $\alpha _i = \max_{j\in [m]}v_{ij}$. The approximation version of this notion would be to construct allocations where each player receives a value of at least $\rho V_2(\alpha_i)$. 
Recently, a stricter variant of this guarantee has been provided by \citet{GMT15} (also see Remark \ref{rem:Wn}).

\newcommand*\circled[1]{\tikz[baseline=(char.base)]{
            \node[shape=circle,draw,inner sep=2pt] (char) {#1};}}

\section{Characterization of Truthful Mechanisms} \label{SEC:CHARACTERIZATION}

We present our main characterization result in this section. We start
in subsection \ref{subsec:p/e_mech_defn} with the main definitions and
illustrating examples, and then we state our result in subsection
\ref{subsec:characterization} along with a road map of the proof. To avoid repetition, when referring to a truthful mechanism $\mathcal{X}$, we mean a truthful mechanism  for allocating all the items in $M$ to two players  with additive valuation functions.

\subsection{A Non-Dictatorial Class of Mechanisms} \label{subsec:p/e_mech_defn}

The main result of this section is that every truthful mechanism  is a
picking-exchange mechanism (Theorem
\ref{thm:truthful-->p/e_mech}). Before we make a precise statement, we
formally define the types of mechanisms involved and provide illustrating examples. 

\paragraph{Picking Mechanisms.}
We start with a family of mechanisms where players make a selection out of choices that the mechanism offers to them.
Given a subset $S$ of items, we define a \emph{set of offers}
$\mathcal{O}$ on $S$, as a nonempty collection of proper subsets of $S$ that exactly covers
$S$ (i.e., $\bigcup_{T\in \mathcal{O}}T = S$), and in which there is no
common element that appears in all subsets 
(i.e., $\bigcap_{T\in \mathcal{O}}T = \emptyset$).

\begin{definition}
  A mechanism $\mathcal{X}$ is a \emph{picking
    mechanism}\footnote{Picking mechanisms are a generalization of
    \emph{truthful} picking sequences for two players 
    \citep[see][]{BoLa14}.} if there exists a partition $(N_1, N_2)$ of $M$,
  and sets of offers $\mathcal{O}_1$ and $\mathcal{O}_2$ on $N_1$ and
  $N_2$ respectively, such that for every profile
  $\mathbf{v}$, $$X_i(\mathbf{v})\cap N_i\in \argmax_{S \in
    \mathcal{O}_i}
  v_i(S).$$ 
\end{definition}

Technical nuances aside, such a mechanism can be implemented by first
letting player 1 choose his best offer from $\mathcal{O}_1$ and giving
what remains from $N_1$ to player 2. Then it lets player 2 choose his
best offer from $\mathcal{O}_2$ and gives what remains from $N_2$ to
player 1. The following example illustrates a picking mechanism.

\begin{example}\label{ex:ex-pick}
  Consider the following mechanism $\mathcal{X}$ on a set $M= \{1,
  \ldots 6\}$, which first partitions $M$ into $N_1=\{1,2,3,4\},
  N_2=\{5,6\}$ and then constructs the offer sets $
  \mathcal{O}_1=\{\{1,2\},\{2,3\},\{4\}\},
  \mathcal{O}_2=\{\{5\},\{6\}\}$. On input $\mathbf{v}$, $\mathcal{X}$
  first gives to player 1 his best set---with respect  to $v_1$---among $\{1,
  2\}$, $\{2, 3\}$ and $\{4\}$, and then gives what remains from $N_1$
  to player 2. Next, $\mathcal{X}$ gives to player 2 his best
  set---according to $v_2$---among $\{5\}$ and $\{6\}$, and then gives
  what remains from $N_2$ to player 1. $\mathcal{X}$ resolves ties
  lexicographically, e.g., in case of a tie, $\{1, 2\}$ is preferred
  to $\{4\}$.

  It is not hard to see that $\mathcal{X}$ is truthful. 
  For the following input $v$, the circles denote the allocation.
\[v=\left(
\begin{array}{c c c c c c}
	 3 & \circled{5} & \circled{5} & 10 & 4 & \circled{2} \\   
	 \circled{2} & 3 & 6 & \circled{1} & \circled{5} & 3 \\   
\end{array}\right).
\]
\end{example}

\paragraph{Exchange Mechanisms.} We now move to a quite different class of mechanisms. 
Let $X,Y$ be two disjoint subsets of
$M$. We call the ordered pair $(X,Y)$ an exchange deal. Moreover, we
say that an exchange deal $(X, Y)$ is \emph{favorable with respect to
  $\mathbf{v}$} if $v_1(Y) > v_1(X)$ and $v_2(Y) < v_2(X)$, while it
is \emph{unfavorable with respect to $\mathbf{v}$} if $v_1(Y) < v_1(X)$ 
or $v_2(Y) > v_2(X)$. Let $S$ and $T$ be two disjoint subsets
of items and let $S_1, S_2, \ldots, S_k$ and $T_1, \ldots, T_k$ be two
collections of nonempty and pairwise disjoint subsets of $S$ and $T$
respectively. We say then that the set of exchange deals 
$D=\{(S_1, T_1), (S_2, T_2), \ldots,\allowbreak (S_k, T_k)\}$ on 
$(S, T)$ is \emph{valid}.

\begin{definition}
A mechanism $\mathcal{X}$ is an \emph{exchange
mechanism}\footnote{If we think about $E_1, E_2$ as fixed a
    priori, then exchange mechanisms are a generalization of fixed
    deal exchange rules in general exchange markets for two players \citep[see][]{Papai07}.} if there exists a partition $(E_1, E_2)$ of
  $M$, and a valid set of exchange deals $D=\{(S_1, T_1), \ldots, (S_k, T_k)\}$ 
on $(E_1, E_2)$, such that for every profile $\mathbf{v}$,
  there exists a set of indices $I=I(\mathbf{v})\subseteq [k]$ for which 
\[X_1(\mathbf{v})= \Bigg( E_1\mysetminus \bigcup_{i\in I}
    S_i\Bigg)\cup \bigcup_{i\in I} T_i\,,\quad X_2(\mathbf{v}) =
  M\setminus X_1 \,.\] 
Moreover, $I$ contains the indices of every
  favorable exchange deal with respect to $\mathbf{v}$, but no indices
  of unfavorable exchange deals.
\end{definition}

On a high level, an exchange mechanism initially partitions the items
into endowments for the players, and then examines a list of possible
exchange deals. Every exchange that improves both players is
performed, while every exchange that reduces the value of even one
player is avoided. The mechanism may also perform other exchanges
where one player is indifferent and the other player can be either indifferent or
improved. Whether such exchange deals are materialized or not is up to
the tie-breaking rule employed by the mechanism. The following example
illustrates an exchange mechanism.

\begin{example}\label{ex:ex-exchange} Let $M=
  \{1,\ldots 5\}$, and consider the following mechanism $\mathcal{Y}$,
  with $E_1=\{1,2,3\}$, $E_2=\{4,5\}$, and
a valid set of exchange deals $D=\{(\{2,3\},\{4\})\}$
on $(E_1, E_2)$:
    One can think of such a mechanism as if $\mathcal{Y}$ initially
  reserves the set $E_1$ for player 1 and the set $E_2$
  for player 2. Then it examines whether exchanging $\{2, 3\}$ with
  $\{4\}$ strictly improves both players, and performs the exchange
  only if the answer is yes. Mechanism $\mathcal{Y}$ is an example of an \emph{exchange mechanism}
  with only one possible \emph{exchange deal}. Again, one can see that
  no player has an incentive to lie.

  For the following input $v$, the circles denote the allocation produced.
\[v=\left(
\begin{array}{c c c c c}
	 \circled{6} & 2 & 3 & \circled{7} & 1 \\   
	 1 & \circled{6} & \circled{1} & 4 & \circled{7}    
\end{array}\right).
\]
\end{example}

\paragraph{Picking-Exchange Mechanisms}

Finally, we define the class of picking-exchange mechanisms which is a
generalization of both picking and exchange mechanisms.

\begin{definition}\label{def:p/e_mech}
  A mechanism $\mathcal{X}$ is a \emph{picking-exchange mechanism} if
  there exists a partition $(N_1, \allowbreak N_2, E_1, E_2)$ of $M$, sets of
  offers $\mathcal{O}_1$ and $\mathcal{O}_2$ on $N_1$ and $N_2$
  respectively, and a valid set of exchange deals $D=\{(S_1, T_1),
  \ldots, (S_k, T_k)\}$ on $(E_1, E_2)$, such that for every profile
  $\mathbf{v}$, 
  \[X_i(\mathbf{v})\cap{N_i}\in \argmax_{S \in
    \mathcal{O}_i} v_i(S) \text{\quad and \quad} X_1(\mathbf{v})\cap(E_1\cup E_2)=
     \Bigg( E_1\mysetminus \bigcup_{i\in I} S_i\Bigg)\cup \bigcup_{i\in I} T_i \,,\]
  where $I= I(\mathbf{v})\subseteq [k]$ contains the
  indices of all favorable exchange deals, but no indices of
  unfavorable exchange deals.
\end{definition}

It is helpful to think that a picking-exchange mechanism runs
independently a picking mechanism on $N_1\cup N_2$ and an exchange
mechanism on $E_1\cup E_2$, like in Example
\ref{ex:p/e_mech}. Although this is true under the assumption that the
players' valuation functions are such that no two sets have the same
value, it is not true for general additive valuations. The reason is
that the tie-breaking for choosing the offers from $\mathcal{O}_1$ and
$\mathcal{O}_2$ may not be independent from the decision of whether to
perform each exchange that is neither favorable nor
unfavorable. 

The following example illustrates a picking exchange mechanism.

\begin{example}\label{ex:p/e_mech}
  Let $M=\{1,\ldots, 11\}$, and consider the mechanism $\mathcal{Z}$ that
  partitions $M$ into $N_1=\{1, 2, \allowbreak 3, 4\}$, $N_2=\{5,6\}$, $E_1=\{7,8,9\}$
  and $E_2=\{10,11\}$, and is the combination of $\mathcal{X}$ and
  $\mathcal{Y}$ from the previous two examples: On input $\mathbf{v}$, $\mathcal{Z}$ runs
  $\mathcal{X}$ on $N_1\cup N_2$ and $\mathcal{Y}$ on $E_1 \cup
  E_2$ (where $\mathcal{Y}$ should of course be adjusted to run on $\{7,  \ldots, 11\}$ instead of $\{1, \ldots, 5\}$). It outputs the union of the outputs of $\mathcal{X}$ and
  $\mathcal{Y}$.

For the following input $v$, the circles denote the final allocation.
\[v=\left(
\begin{array}{c c c c c c c c c c c}
	 3 & \circled{5} & \circled{5} & 10 & 4 & \circled{2} & \circled{6} & 2 & 3 & \circled{7} & 1 \\   
	 \circled{2} & 3 & 6 & \circled{1} & \circled{5} & 3 & 	 1 & \circled{6} & \circled{1} & 4 & \circled{7}     \\   
\end{array}\right).
\]
\end{example}

\subsection{Truthfulness and Picking-Exchange Mechanisms}\label{subsec:characterization}

Essentially, we show that a mechanism is truthful if and only if it is a picking-exchange mechanism. 
We begin with the easier part of our characterization, namely that under the assumption that each valuation function induces a strict preference relation over all possible subsets, every picking-exchange mechanism is truthful. Recall 
that the set of such profiles is denoted by $\mathcal{V}^{\neq}_{\! m}$. 

\begin{theorem}\label{thm:p/e_mech-->truthful}
	When restricted to $\mathcal{V}^{\neq}_{\! m}$, every picking-exchange mechanism $\mathcal{X}$ for allocating $m$ items is truthful.
\end{theorem}


\begin{remark}\label{rem:tie-breaking_1}
	For simplicity, Theorem \ref{thm:p/e_mech-->truthful} is stated for a subclass of additive valuation functions. However, it  holds for general additive valuations as long as the mechanism uses a sensible tie-breaking rule (e.g., label-based or welfare-based, see Remark \ref{rem:ties} in Appendix \ref{app:characterization}). 
	The proof is very similar.\footnote{Describing all such tie-breaking rules seems to be an interesting, nontrivial question for future work, but not our main focus here. It is not hard to see, though, that there exist tie-breaking rules that make a picking-exchange mechanism nontruthful, e.g.,  break ties on offers of player 1 so that the value that player 2 gets from $N_1$ is minimized.}
\end{remark}

We are now ready to state the main result of this work. 

\begin{theorem}\label{thm:truthful-->p/e_mech}
Every truthful mechanism $\mathcal{X}$  
can be implemented as a picking-exchange mechanism.
\end{theorem}

The rest of this subsection is a road map to the proof of Theorem \ref{thm:truthful-->p/e_mech}.
The proof is long and technical, so for the sake of presentation, it is broken down to several lemmata. In order to illustrate the high-level ideas, the proofs  of those lemmata are deferred to the 
 Appendix \ref{app:characterization}. 

For the rest of this subsection we assume a truthful mechanism $\mathcal{X}$ for allocating all the items in $M=[m]$  to two players  with additive valuation functions. Every statement is going to be with respect to this $\mathcal{X}$.

\subsubsection{The Crucial Notion of Control}
We begin by introducing the notions of \textit{strong desire} and of \textit{control}, which are of key importance for our characterization. We say that  player $i$ \textit{strongly desires} a set $S$ if each item in $S$ has more value for him than all the items of $M\mysetminus S$ combined, i.e., if for every $x \in S$ we have $v_{ix}>\sum_{y \in M\mysetminus S}v_{iy}$. 

\begin{definition}\label{def:control}
	We say that player $i$ \textit{controls} a set $S$ with respect to  $\mathcal{X}$, if every time he strongly desires $S$ he gets it whole, i.e., for every $\mathbf{v}=(v_1,v_2)$ in which player $i$ strongly desires $S$, then we have that $S\subseteq X_i(\mathbf{v})$ . 
\end{definition}
	
Clearly, given $\mathcal{X}$, any set $S$ can be controlled by at most one player.

The following is a key lemma for understanding how truthful mechanisms operate. The lemma together with Corollary \ref{lem:controlsin} below show that every item is controlled by some player under any truthful mechanism.

\begin{lemma}[Control Lemma]\label{lem:control}
Let 
$S \subseteq M$. If there exists a profile $\mathbf{v}=(v_1,v_2)$ such that both players strongly desire $S$, and $S\subseteq X_i(\mathbf{v})$ for some $i \in \{1,2\}$, then player $i$ controls every $T\subseteq S$ with respect to $\mathcal{X}$.
\end{lemma}

\begin{proof} 
	Let $\mathbf{v}=(v_1,v_2)$ be a profile such that both players strongly desire $S$ and $S\subseteq X_1(\mathbf{v})$ (the case where $S\subseteq X_2(\mathbf{v})$ is symmetric). We first prove the statement for $T=S$. Let $\mathbf{v}'=(v'_1,v'_2)$ be any profile in which player $1$ strongly desires $S$, i.e., $v'_{1x}>\sum_{y \in M\mysetminus S}v'_{1y}, \forall x \in S$. Initially, consider the intermediate profile $\mathbf{v}^*=(v_1,v'_2)$. If $S \cap X_2(\mathbf{v}^*)\neq \emptyset$ then player 2 would deviate from profile $\mathbf{v}$ to $\mathbf{v}^*$ in order to strictly improve his total utility. So by truthfulness we derive that $S \subseteq X_1(\mathbf{v}^*)$. Similarly, in the profile $\mathbf{v}'$, if $S \cap X_2(\mathbf{v}')\neq \emptyset$ then player 1 would deviate from  $\mathbf{v}'$ to $\mathbf{v}^*$ in order to strictly improve. Thus by truthfulness we have $S\subseteq X_1(\mathbf{v}')$. We conclude that player 1 controls $S$.
	
	Now, suppose that $\mathbf{v}''=(v''_1,v''_2)$ is any profile in which player $1$ strongly desires $T\subsetneq S$. If $T \nsubseteq X_1(\mathbf{v}'')$ then player 1 could strictly improve his utility by playing $v'_1$ from before (i.e., he declares that he strongly desires $S$) and getting $S\supsetneq T$. Thus, by truthfulness, $T \subseteq X_1(\mathbf{v}'')$, and we conclude that player 1 controls $T$.
\end{proof}

Notice here that the existence of sets that are controlled by some player is always guaranteed. Specifically, each singleton $\{x\}$ is always controlled (only) by one of the players. Indeed, when both players strongly desire $\{x\}$, it is always the case that $\{x\}\subseteq X_i(\mathbf{v})$ for some $i \in \{1,2\}$.  This is summarized in the following corollary. 

\begin{corollary}\label{lem:controlsin}
Let $\mathcal{X}$ be a truthful mechanism for allocating the items in $M$ to two players with additive valuations. 
For every $x\in M$ there exists $i\in\{1, 2\}$ such that only player $i$ controls $\{x\}$ with respect to $\mathcal{X}$.
\end{corollary}

Aside from its use in the current proof, the corollary has implications on fairness, that will be explored in Section \ref{SEC:FAIRNESS}.



\subsubsection{Identifying the Components of a Mechanism}
Our goal now is to determine the ``exchange component'' and the ``picking component'' of mechanism $\mathcal{X}$. Every picking-exchange mechanism is completely determined by the seven sets $N_1$, $N_2$, $\mathcal{O}_1$, $\mathcal{O}_2$, $E_1$, $E_2$, and $D$ mentioned in Definition \ref{def:p/e_mech} (plus a deterministic tie-breaking rule). 
Below we try to identify these sets. Later we show that the mechanism's behavior is identical to that of a picking-exchange mechanism defined by them.

To proceed, we will need to consider the collection of all maximal sets controlled by each player. For $i\in\{1, 2\}$, let 
\[\mathcal{A}_i=\{S\subseteq M \ |\ \text{player $i$ controls $S$ and for any $T\supsetneq S$, $i$ does not control $T$}\} \,.\]
Clearly, every set controlled by player $i$ is a subset of an element of $\mathcal{A}_i$. According to Lemma \ref{lem:control}, if we consider the set $C_i=\bigcup_ {S \in\mathcal{A}_i}S$, i.e., the union of all the sets in  $\mathcal{A}_i$, this is exactly the set of items that are controlled---as singletons---by player $i$. 

\begin{corollary}\label{cor:part}
The sets $C_1$ and $C_2$ define a partition of $M$.
\end{corollary}

Using the $\mathcal{A}_i$s and the $C_i$s, we define the sets of interest that determine the mechanism. 
We begin with $E_i=\bigcap_ {S \in \mathcal{A}_i}S$ for $i\in \{1, 2\}$. As we are going to see eventually in Lemma \ref{lem:final}, 
the ``exchange component'' of $\mathcal{X}$ is observed on $E_1 \cup E_2$. 

Defining the corresponding valid set of exchange deals $D$ is trickier, and we need some terminology. Recall that $X_i^S(\mathbf{v}) = X_i(\mathbf{v}) \cap S$.
For $S \subseteq E_1$ and $T \subseteq E_{2}$, we say that $(S,T)$ is a \emph{feasible exchange}, if there exists a profile $\mathbf{v}$, such that $X_1^{E_1 \cup E_2}(\mathbf{v})=(E_1 \mysetminus S) \cup T$. In such a case, each of $S$ and $T$ is called \emph{exchangeable}.
An exchangeable set $S$ is called \emph{minimally exchangeable} if any $S' \subsetneq S$ is not exchangeable. Finally, a feasible exchange $(S,T)$ is a \emph{minimal feasible exchange}, if at least  one of $S$ and $T$ is minimally exchangeable. Now let
\[D=\{(S, T) \ | \ (S, T) \text{ is a minimal feasible exchange with respect to $\mathcal{X}$} \} \,.\]
Of course, at this point it is not clear whether $D$ is well defined as a valid set of exchange deals, and this is probably the most challenging part  of the characterization. 

Next, we define $N_i= C_i \mysetminus E_i$  and  $\mathcal{O}_i=\{S \mysetminus E_i\ |\ S \in \mathcal{A}_i\}$  for $i\in \{1, 2\}$. As shown in Lemmata \ref{output} and \ref{best}, we identify the ``picking component'' of  $\mathcal{X}$ on $N_1 \cup N_2$, and
$\mathcal{O}_i$ will correspond to the set of offers. 

Note that by Corollary \ref{cor:part} and the above definitions,  $(N_1, N_2, E_1, E_2)$ is a partition of $M$. The intuition behind breaking $C_i$ into $N_i$ and $E_i$ is that player $i$ has different levels of control on those two sets. The fact that $E_i$ is contained in every maximal set controlled by player $i$ will turn out to mean that $\mathcal{X}$ gives the ownership of $E_i$ to player $i$. On the other hand, the control of player $i$ on $N_i$ is much more restricted as shown below.

\subsubsection{Cracking the Picking Component}

The first step is to show that the $\mathcal{O}_i$s defined above, greatly restrict the possible allocations of the items of $N_1 \cup N_2$. In particular, whatever player $i$ receives from $N_i$ must be contained in some set of $\mathcal{O}_i$.
\begin{lemma}\label{output}
	For every profile $\mathbf{v}$ and every $i\in \{1, 2\}$, there exists $S \in \mathcal{O}_i$ such that $X_i^{N_i}(\mathbf{v})\subseteq S$.
\end{lemma}
The idea behind the proof of Lemma \ref{output} is that by receiving some $X_i^{N_i}(\mathbf{v})$ not contained in any set of $\mathcal{O}_i$, player $i$ is able to extend his control to subsets not contained in $C_i$, thus leading to contradiction. The proof, as many of the proofs of the remaining lemmata, includes the careful construction of a series of profiles, where in each step one has to argue about how the allocation does or does not change.

Given the restriction implied by Lemma \ref{output}, next we can prove that the subset of $N_i$ that player $i$ receives must be the best possible from his perspective, hence the mechanism behaves as a picking mechanism on each $N_i$. 
Intuitively,  suppose that player 1 receives a subset $S$ of $N_1$ which is not an element of $\mathcal{O}_1$. By Lemma \ref{output}, $S$ is contained in an element $S'$ of $\mathcal{O}_1$. Since player 1 controls $S'$, this means that he gave up part of his control to gain something that he was not supposed to. Actually, it can be shown that it is the case where player 2 also gave part of his control (either on $N_2$ or $E_2$). This mutual transfer of control, combined with truthfulness, eventually leads to profiles where some of the items must be given to both players at the same time, hence a contradiction. 
\begin{lemma}\label{best}
	For every profile $ \mathbf{v}$ and every $i\in \{1, 2\}$ we have  $X_i^{N_i}(\mathbf{v}) \in \argmax_{S \in \mathcal{O}_i}v_i(S)$.
\end{lemma}
Now we know that $\mathcal{X}$ behaves as the ``right''  picking-exchange mechanism on $N_1 \cup N_2$. For most of the rest of the proof we would like to somehow ignore this part of $\mathcal{X}$ and focus on $E_1 \cup E_2$.


\subsubsection{Separating the Two Components}

As mentioned right after Definition \ref{def:p/e_mech}, there is some kind of independence between the two components of a picking-exchange mechanism, at least when  restricted on $\mathcal{V}^{\neq}_{\! m}$. This independence should be present in $\mathcal{X}$ as well; in fact we are going to exploit it to get rid of $N_1 \cup N_2$  until the last part of the proof. 

\begin{lemma}\label{indie}
Let $\mathbf{v}=(v_1, v_2), \mathbf{v}'=(v'_1, v'_2)\in \mathcal{V}^{\neq}_{\! m}$ such that $v_{ij}=v'_{ij}$ for all $i\in \{1, 2\}$ and $j\in E_1 \cup E_2$. Then $X_1^{E_1 \cup E_2}(\mathbf{v})=X_1^{E_1 \cup E_2}(\mathbf{v}')$.
\end{lemma}

The lemma states that assuming strict preferences over all subsets, the allocation of $E_1 \cup E_2$ does not depend on the values of either player for the items in $ N_1 \cup  N_2$. What allows this separation is the complete lack of ties in the restricted profile space.

Without loss of generality we may assume that $E_1 \cup E_2=[\ell]$. We can define a mechanism $\mathcal{X}_E$ for allocating the items of $[\ell]$ to two players with valuation profiles in $\mathcal{V}^{\neq}_{\! \ell}$ as
\[\mathcal{X}_E(\mathbf{v}) = (X_1^{E_1 \cup E_2}(\mathbf{v}'), X_2^{E_1 \cup E_2}(\mathbf{v}')), \text{ for every } \mathbf{v}\in \mathcal{V}^{\neq}_{\! \ell},\]
where $\mathbf{v}'$ is any profile in $\mathcal{V}^{\neq}_{\! m}$  with  $v_{ij}=v'_{ij}$ for all $i\in \{1, 2\}$ and $j\in [\ell]$. This new mechanism is just the projection of $\mathcal{X}$ on $E_1 \cup E_2$ restricted on a domain where it is well-defined. The truthfulness of $\mathcal{X}_E$ on $\mathcal{V}^{\neq}_{\! \ell}$ follows directly from the truthfulness of $\mathcal{X}$ on $\mathcal{V}^{\neq}_{\! m}$. Moreover, it is easy to see that player $i$ controls $E_i$ with respect to $\mathcal{X}_E$, for $i\in\{1, 2\}$.

The plan is to study $\mathcal{X}_E$ instead of $\mathcal{X}$, show that $\mathcal{X}_E$ is an exchange mechanism, and finally sew the two parts of $\mathcal{X}$ back together and show that everything works properly for any profile in  $\mathcal{V}_{\! m}$. One issue here is that maybe the set of feasible exchanges with respect to $\mathcal{X}_E$ is greatly reduced, in comparison to the set of feasible exchanges with respect to $\mathcal{X}$, because of the restriction on the domain. In such a case, it will not be possible to argue about exchanges in $D$ that are not feasible anymore. 
It turns out that this is not the case; the set of possible allocations (of $E_1 \cup E_2$) is the same, whether we consider profiles in $\mathcal{V}_{\! m}$ or in $\mathcal{V}^{\neq}_{\! m}$.

\begin{lemma}\label{strict}
For every profile $\mathbf{v}\in \mathcal{V}_{\! m}$ there exists a profile $\mathbf{v}' \in \mathcal{V}^{\neq}_{\! m}$ such that $\mathcal X(\mathbf{v})= \mathcal X(\mathbf{v}')$.
\end{lemma}

In particular, the set of feasible exchanges on $E_1 \cup E_2$ is exactly the same for $\mathcal{X}$ and $\mathcal{X}_E$, and 
thus we will utilize the following set of exchanges.
\[D=\{(S, T) \ | \ (S, T) \text{ is a minimal feasible exchange with respect to $\mathcal{X}_E$} \} \,.\]

\subsubsection{Cracking the Exchange Component}
In the attempt to show that $\mathcal{X}_E$ is an exchange mechanism, the first step is to show that $D$ is indeed a valid set of exchange deals. 

\begin{lemma}\label{lem:proper}
$D$ is a valid set of exchange deals on $(E_1, E_2)$.
\end{lemma}

The above lemma involves three main steps. First we show that each minimally exchangeable set is involved in exactly one exchange deal. Then, we guarantee that minimally exchangeable sets can be exchanged only with minimally exchangeable sets, and finally, we show that minimally exchangeable sets are always disjoint.
There is a common underlying idea in the proofs of these steps: whenever there exist two feasible exchanges that overlap in any way, we can construct a profile where both of them are favorable but the two players disagree on which of them is best. On a high level, each player can ``block'' his least favorable of the conflicting exchanges, and this leads to violation of truthfulness.


Lemma \ref{lem:proper} implies that every  exchangeable set  $S\subseteq E_1$ can be decomposed as $S=W\cup \bigcup_{i \in I} S_i$, where  $W=S \mysetminus \bigcup_{i \in I} S_i$ does not contain any minimally exchangeable sets. 
Ideally, we would like two things. First, the set $W$  in the above decomposition to always be empty, i.e., every  exchangeable set should  be a union of minimally exchangeable sets.
Second, we want every union of minimally exchangeable subsets of $E_1$ to be exchangeable only with the corresponding union of minimally exchangeable subsets of $E_2$, and vice versa. 
It takes several lemmas and a rather involved induction to prove those. A key ingredient of the inductive step is a carefully constructed argument about the value that each player must gain from any exchange
(see also Lemma \ref{lem:values} in Appendix \ref{app:characterization}). 

\begin{lemma}\label{lem:umes}
For every exchangeable set $S\subseteq E_1$, there exists some $I\subseteq[k]$ such that $S=\bigcup_{i \in I}S_i$. Moreover, $S$ is  exchangeable with $T=\bigcup_{i \in I}T_i$ and only with $T$.
\end{lemma}

Finally, we have all the ingredients to fully describe $\mathcal{X}_E$ as an exchange mechanism on $E_1 \cup E_2$ and set of exchange deals $D$. 

\begin{lemma}\label{lem:favor}
Given any profile $\mathbf{v}\in \mathcal{V}^{\neq}_{\! \ell}$, each exchange in $D$ is  performed if and only if it is favorable, i.e., 
$X_1^{E_1\cup E_2}(\mathbf{v})= \left( E_1\mysetminus \bigcup_{i\in I} S_i\right)\cup \bigcup_{i\in I} T_i$, where $I\subseteq [k]$ contains exactly the indices of all favorable exchange deals in $D$. 
\end{lemma}

\subsubsection{Putting the Mechanism Back Together}

As a result of Lemma \ref{lem:favor} (combined, of course, with Lemmata \ref{best} and \ref{indie}), the characterization is complete for truthful mechanisms defined on $\mathcal{V}^{\neq}_{\! m}$. For general additive valuation functions, however, we need a little more work. This is to counterbalance the fact that  in the presence of ties the allocations of $N_1 \cup N_2$ and $E_1 \cup E_2$ may not be independent.

By Lemmata \ref{strict} and \ref{lem:umes}, we know that for any $\mathbf{v}\in \mathcal{V}_{\! m}$, $X_1^{E_1\cup E_2}(\mathbf{v})$  is the result of some exchanges of $D$ taking place. There are two things that can go wrong: $\mathcal{X}$ performs an unfavorable exchange, or it does not perform a favorable one. In either of these cases it is possible to construct some profile in $\mathcal{V}^{\neq}_{\! m}$ that leads to contradiction. Hence we have the following lemma.

\begin{lemma}\label{lem:final}
Given any profile $\mathbf{v}\in \mathcal{V}_{\! m}$, $X_1^{E_1\cup E_2}(\mathbf{v})= \left( E_1\mysetminus \bigcup_{i\in I} S_i\right)\cup \bigcup_{i\in I} T_i$, where $I\subseteq [k]$ contains the indices of all favorable exchange deals in $D$, but no indices of unfavorable exchange deals.
\end{lemma}

Clearly, Lemma \ref{lem:final}, together with Lemma \ref{best} concludes the proof of Theorem \ref{thm:truthful-->p/e_mech}.

\subsection{Immediate Implications of Theorem \ref{thm:truthful-->p/e_mech}}
As mentioned in Section \ref{subsec:relwork}, there are several works characterizing truthful mechanisms in combination with other notions, such as Pareto efficiency, nonbossiness, and neutrality (these results are usually for unrestricted, not necessarily additive valuations).
Pareto efficiency means that there is no other allocation where one player strictly improves and none of the others are worse-off.
Non-bossiness means that a player cannot affect the outcome of the mechanism without changing his own bundle of items.
Finally, neutrality refers to a mechanism being consistent with a permutation on the items, i.e., permuting the items results in the corresponding permuted allocation. 

Although such notions are not our main focus, the purpose of this short discussion is twofold. On one hand, we illustrate how our characterization immediately implies a characterization for mechanisms that satisfy these extra properties under additive valuations, and on the other hand we see how these properties are either incompatible with fairness or irrelevant in our context.

To begin with, nonbossiness comes for free in our case, since we have two players and all the items must be allocated. Neutrality and Pareto efficiency, however, greatly reduce the space of available mechanisms. Note that it makes more sense to study neutral mechanisms  when the valuation functions induce a strict preference order over all sets of items. 

\begin{corollary}
Every neutral, truthful mechanism $\mathcal{X}$  on $\mathcal{V}^{\neq}_{\! m}$ 
can be implemented as a picking-exchange mechanism, such that
\begin{enumerate}
\item there exists $i\in \{1, 2\}$ such that $E_i=[m]$, \emph{or}
\item there exists $i\in \{1, 2\}$ such that $N_i=[m]$ and  $\mathcal{O}_i=\{S\subseteq [m] \ | \ |S|=\kappa\}$ for some  $\kappa< m$.
\end{enumerate}
\end{corollary}

\begin{corollary}
Every Pareto efficient, truthful mechanism $\mathcal{X}$ 
can be implemented as a picking-exchange mechanism, such that
\begin{enumerate}
\item there exists $i\in \{1, 2\}$ such that $E_i=[m]$, \emph{or}
\item there exists $j\in [m]$ such that $E_{i_1}=\{j\}$, $E_{i_2}=[m]\mysetminus \{j\}$, where $\{i_1, i_2\} = \{1, 2\}$, and $D=\{(E_1, E_2)\}$, \emph{or}
\item there exists $i\in \{1, 2\}$ such that $N_i=[m]$ and $\mathcal{O}_i=\{S\subseteq N_i \ | \ |S|=m-1\}$.
\end{enumerate}
\end{corollary}

The proofs are deferred to the full version.

It is somewhat surprising that the resulting mechanisms are a strict superset of dictatorships, even when we impose both properties together. Pareto efficiency, however, allows only mechanisms that are rather close to being dictatorial, and thus cannot guarantee fairness of any type. On the other hand, most of the mechanisms defined and studied in Section \ref{SEC:FAIRNESS} are neutral, yet neutrality is not implied by the fairness concepts we consider, nor the other way around.

\section{A Necessary Fairness Condition and its Implications} \label{SEC:FAIRNESS}
In this section, we explore some implications of Theorem
\ref{thm:truthful-->p/e_mech} on fairness properties, i.e., on the
design of mechanisms where on top of truthfulness, we would like
to 
achieve fairness guarantees.
All the missing proofs from this section are 
in Appendix \ref{app:fairness}.

In Section~\ref{sec:control} we show that the Control Lemma implies
that truthfulness prevents any bounded approximation for envy-freeness
and proportionality.  Then, we move on  describing  a necessary fairness condition, in terms of our notion of ``control'', that summarizes a common feature of several relaxations of fairness 
and provide a restricted version of our characterization that follows this fairness condition. This
will allow us, in Section~\ref{sec:app-fairness}, to examine what this new class of mechanisms can achieve in each
of these fairness concepts.

\subsection{Implications of the Control Lemma.}
\label{sec:control}

\subsubsection{Control of singletons.}
\label{sec:sincontrol}
The basic restriction that truthfulness imposes to every mechanism (leading to poor results for some fairness concepts) comes from Corollary \ref{lem:controlsin}, an immediate corollary of the
Control Lemma, stating that every single item is controlled by some
player.

We begin by studing how the above corollary affects two of the most researched notions in the fair
division literature, namely {\it proportionality} and {\it
  envy-freeness}. It is well known that even without the requirement
for truthfulness, it is impossible to achieve any of these two
objectives, simply because in the presence of indivisible goods,
envy-free or proportional allocations may not
exist.\footnote{Consider, for instance, a profile where both players
  desire only the first item and have a negligible value for the other
  items. Then one of the players will necessarily remain unsatisfied
  and receive a value close to zero, no matter what the allocation
  is. 

} 

This leads to the definition of approximation versions of these two
concepts for settings with indivisible goods. 
Namely \citet{LMMS04} considered the {\em minimum envy} problem and tried to construct algorithms 
such that for every instance, an approximation to the minimum possible envy admitted by the instance is guaranteed. 
Similarly \citet{MP11} considered approximate proportionality, i.e., find allocations that achieve an approximation to the best possible value that an instance can guarantee to all agents. See also the discussion in Section \ref{sec:notation} on defining the approximation versions of these problems.
Note that if time complexity is not an issue, we can always identify the allocation with the best possible envy or with the best possible proportionality, achieveable by a given instance. 


We are now ready to state our first application, showing that truthfulness prohibits us from having any approximation to the minimum envy or to proportionality. This greatly improves the conclusions of \citet{LMMS04} and \citet{CKKK09} that truthful mechanisms cannot attain the optimal minimum envy allocation.

\begin{application} 
	\label{app:envy}
  For any truthful mechanism that allocates all the items to two
  players with additive valuations, the approximation achieved for
  either proportionality or the minimum envy is arbitrarily bad (i.e.,
  not lower bounded by any positive function of $m$).
\end{application}

\begin{proof}
Consider a setting with $m$ items, and a truthful mechanism $\mathcal{X}$.
Suppose now that item
$1$ is controlled by player $1$ with respect to $\mathcal{X}$. This
means that in the profile $\mathbf{v}=([m\ 1\ 1\ \ldots\ 1],\allowbreak [m^d\ 1\
1\ \ldots\ 1])$ player 1 must obtain item $1$, and player 2 ends up with a
negligible fraction of his total value for large enough $d$. The
optimal solution would be to assign the first item to the second
player and the last $m$ items to the first player, which provides an envy-free and proportional allocation.
We conclude that the approximation guarantee that can be obtained
by a truthful mechanism is arbitrarily high. 
\end{proof}

So far, the conclusion is that even approximate proportionality or
envy-freeness are quite stringent and incompatible with truthfulness
because of the Control Lemma.  The next step would be to relax these
notions.  There have been already a few approaches on relaxing
proportionality and envy-freeness under indivisible goods, leading to
solutions such as the maximin share fairness, envy-freeness
up to one item~\citep{Budish11}, as well as the type of worst-case guarantees
proposed by \citet{hill87} (recall Definitions \ref{def:uptoone},
\ref{def:mms} and \ref{def:Vn} in Section
\ref{sec:notation}). 
The fact that a truthful mechanism $\mathcal{X}$ yields control of
singletons does not seem to have such detrimental effects on these
notions. However, if even a single pair of items is controlled by a player,
the same situation arises.

\subsubsection{Control of pairs}
\label{sec:pcontrol}
We propose the following {\em necessary} (but not sufficient)
condition that captures a common aspect of all these relaxations of
fairness. This allows us to treat all the above concepts of fairness
in a unified way.

\begin{definition}\label{def:2-set}
We say that a mechanism $\mathcal{X}$ \emph{yields control of pairs}
if there exists $i\in\{1, 2\}$ and $S\subseteq [m]$ with $|S|=2$, such
that player $i$ controls $S$ with respect to
$\mathcal{X}$. 
\end{definition}

The following lemma states that in order to obtain impossibility
results for the above concepts, it is enough to focus on mechanisms
with control of pairs.

\begin{lemma}\label{lem:control-2sets}
  In order to achieve (either exactly or within a bounded approximation) 
  the above mentioned relaxed fairness criteria, a truthful mechanism that allocates all the items
  to two players with additive valuations cannot yield control of
  pairs.
\end{lemma}

\begin{proof}
  Assuming that $\{1, 2\}$ is controlled by player 1, in a setting
  with $m$ items, we may consider $\mathbf{v}'=([m\ m\ 1\ \ldots\
  1], [m^d\ m^d\ 1\ \ldots\ 1])$, in analogy to profile $\mathbf{v}$
  in the proof of Application \ref{app:envy}. Given that player 1 can always get both items 1 and 2 when he
  strongly desires them, it is easy to see that envy-freeness up to
  one item cannot be achieved, while by choosing large enough $d$ the
  approximation of $\mms_2$ or $V_2(\alpha_2)$ can be arbitrarily bad.
\end{proof}

So now we are ready to move to a complete characterization of truthful
mechanisms that do not yield control of pairs. Of course such
mechanisms are picking-exchange mechanisms, but our fairness condition
allows only singleton offers, and the exchange part is
completely degenerate.

\begin{definition}
  A mechanism $\mathcal{X}$ for allocating all the items in $[m]$ to
  two players is a \emph{singleton picking-exchange mechanism} if it
  is a picking-exchange mechanism where for each $i\in\{1, 2\}$ at
  most one of $N_i$ and $E_i$ is nonempty, $|E_i|\le 1$, and
\begin{equation*}
\mathcal{O}_i =
\begin{cases}
\left\lbrace \{x\} \ | \ x\in N_i\right\rbrace  & \text{when } N_i\neq\emptyset\\
\{\emptyset\} & \text{otherwise } 
\end{cases}
\end{equation*}
i.e., the sets of offers contain all  possible singletons.
\end{definition} 

Hence, typically, in a singleton picking-exchange mechanism player $i$
receives from $N_i \cup E_i$ only his best item. Moreover, for $m\geq 3$, no
exchanges are allowed.\footnote{The only exceptions---and the only
  such mechanisms where both $E_1$ and $E_2$ are nonempty---are two
  mechanisms for the degenerate case of $m=2$, e.g.,
  $N_1=N_2=\emptyset$, $\mathcal{O}_1=\mathcal{O}_2=\{\emptyset\}$,
  $E_1=\{a\}$, $E_2=\{b\}$ and $D=\{(\{a\}, \{b\})\}$, where $\{a,
  b\}=\{1, 2\}$.}

\begin{lemma}\label{lem:fairness_mechs}
  Every truthful mechanism for allocating all the items to two players
  with additive valuation functions that does not yield control of
  pairs can be implemented as a singleton picking-exchange mechanism.
\end{lemma}

It is interesting to note that, in contrast to Application \ref{app:envy},  proving Lemma \ref{lem:fairness_mechs} without Theorem \ref{thm:truthful-->p/e_mech} is not straightforward. In fact, it  requires a partial characterization which (on a high level) is similar to characterizing the picking component of general mechanisms.

\subsection{Applications to Relaxed Notions of Fairness}
\label{sec:app-fairness}

It is now possible to apply Lemma~\ref{lem:fairness_mechs} on each
fairness notion separately, and characterize every truthful mechanism
that achieves each criterion.  

\paragraph{Envy-freeness up to one item.} We start with a relaxation of envy-freeness. Below we provide a complete description of the mechanisms that satisfy this criterion.

\begin{application} 
\label{cor:envy} For $m\le 3$,
  every singleton picking-exchange mechanism achieves envy-freeness up
  to one item. For $m = 4$ every singleton picking-exchange mechanism
  with $|N_1|=|N_2|=2$ achieves envy-freeness up to one item. Finally,
  for $m\ge 5$ there is no truthful mechanism that allocates all the
  items to two players and achieves envy-freeness up to one item.
\end{application}

\paragraph{Maximin share fairness and related notions.} For maximin share allocations a truthful
mechanism was suggested by \citet{ABM16} for any number of items and
any number of players. For two players, their mechanism is the
singleton picking-exchange mechanism with $N_1=[m]$ and produces an
allocation that guarantees to each player a $\frac{1}{\left\lfloor m/2
  \right\rfloor}$-approximation of his maximin share. 
It was left as an open problem whether a better truthful approximation exists.
Here we show
that this approximation is tight; in fact, almost any other singleton
picking-exchange mechanism performs strictly worse. Note that the best
previously known lower bound for two players was $1/2$.

\begin{application}\label{cor:mms}
  For any $m$ there exists a singleton picking-exchange mechanism that
  guarantees to player $i$ a ${\left\lfloor \max\{2, m\}/2
    \right\rfloor}^{-1}$-approximation of $\mms_i$, for $i\in
  \{1,2\}$. There is no truthful mechanism that allocates all the
  items to two players and achieves a better guarantee with respect to
  maximin share fairness.
\end{application}

Regarding now allocations that guarantee an approximation of the function $V_2(\alpha_i)$ defined by \citet{hill87} (recall the definition in Section \ref{sec:notation}), 
the singleton picking-exchange mechanism with
$N_1=[m]$ was also suggested by \citet{MP11} as a
$\frac{1}{\left\lfloor m/2 \right\rfloor}$-approximation of
$V_2(\alpha_i)$.\footnote{The approximation factor in \cite{MP11} is
  expressed in terms of $V_2(1/m)$, but it simplifies to
  ${\left\lfloor m/2 \right\rfloor}^{-1}$.} This comes as no surprise,
since there exists a strong connection between maximin shares and the
function $V_n$, especially for two players. This is illustrated in the
following 
corollary, where both the positive and the negative results coincide with the ones for the maximin share fairness.

\begin{application}\label{cor:Vn}
  For any $m$ there exists a singleton picking-exchange mechanism that
  guarantees to player $i$ a ${\left\lfloor \max\{2, m\}/2
    \right\rfloor}^{-1}$-approximation of $V_2(\alpha_i)$, for $i\in
  \{1,2\}$, where $\alpha _i = \max_{j\in [m]}v_{ij}$. There is no
  truthful mechanism that allocates all the items to two players and
  achieves a better guarantee with respect to the $V_2(\alpha_i)$s.
\end{application}

Again, the best previously known lower bound for two players was constant, namely $2/3$ due to \citet{MP11}.
In Applications \ref{cor:mms} and \ref{cor:Vn}, it is stated that there
exists a $\frac{1}{\left\lfloor m/2 \right\rfloor}$-approximate
singleton picking-exchange mechanism. It is interesting that
\emph{any} singleton picking-exchange mechanism does not perform much
worse. Following the corresponding proofs, we have that even the worst
singleton picking-exchange mechanism achieves a
$\frac{1}{m-1}$-approximation in each case.

\begin{remark}\label{rem:Wn}
  \citet{GMT15} introduced a variant of $V_n$, called $W_n$, and
  showed that there always exists an allocation such that each
  player $i$ receives $W_n(\alpha_i) \ge V_n(\alpha_i)$ (where the
  inequality is often strict). Since the definition of $W_n$ is rather
  involved even for $n=2$, we defer a formal discussion about it to
  the full version of the paper. However, it is not hard to show that
  for every valuation function $v_i$ we have $V_2(\alpha_i) \le
  W_2(\alpha_i) \le \mms_i$ and thus the analog of Application
  \ref{cor:Vn} holds.
\end{remark}

\begin{remark}
\citet{ABM16} made the following interesting observation: every single known truthful mechanism achieving a bounded approximation of maximin share fairness is \emph{ordinal}, in the sense that it only needs a ranking of the items for each player rather than his whole valuation function. Finding truthful mechanisms that explicitly take into account the players' valuation functions in order to achieve better guarantees was posed as a major open problem. Note that, weird tie-breaking aside, all singleton picking-exchange mechanisms are ordinal! Therefore, from the mechanism designer's perspective, it is impossible to exploit the extra cardinal information given as input and at the same time maintain truthfulness and some nontrivial fairness guarantee.
\end{remark}

\section{Truthful Mechanisms for Many Players} \label{sec:further_dir}
We introduce a family of non-dictatorial, truthful mechanisms for any number of players.
Our mechanisms are defined recursively; in analogy to serial dictatorships, the choices of a player define the sub-mechanism used to allocate the items to the remaining players. Here, however, this serial behavior is observed ``in parallel'' in several sets of a partition of $M$. 

A \emph{generalized deal} between $k$ players is a collection of (up to $k(k-1)$) exchange deals between pairs of players. A set $D$ of generalized deals is called \emph{valid} if all the sets involved in all these exchange deals are nonempty and pairwise disjoint. Given a profile $\mathbf{v}=(v_1,v_2, \ldots, v_n)$ we say that a generalized deal is \emph{favorable} if it strictly improves all the players involved, while it is \emph{unfavorable} if  there exists a player involved whose utility strictly decreases. 

\begin{definition}\label{def:s/p/e}
A mechanism $\mathcal{X}$ for allocating all the items in $[m]$ to $n$ players is called a \emph{serial picking-exchange mechanism} if
\begin{enumerate}
\item when $n=1$, $\mathcal{X}$ always allocates the whole $[m]$ to player 1.
\item when $n\ge 2$, there exist a partition $(N_1, \ldots, N_n, E_1, \ldots, E_n)$  of $[m]$, sets of offers $\mathcal{O}_i$ on $N_i$ for $i\in[n]$, a valid set $D$ of generalized deals, and  a mapping  $f$ from subsets of $M$ to serial picking-exchange mechanisms for $n-1$ players, such that for every profile $\mathbf{v}=(v_1, \ldots, v_n)$ we have for all $i\in[n]$:
\begin{itemize}
\item $X_i^{N_i}(\mathbf{v})\in \argmax_{S \in \mathcal{O}_i} v_i(S)$,
\item $X_i^{E}(\mathbf{v})$, where $E=\bigcup_{j\in [n]} E_j$, is the result of starting with $E_i$ and performing some of the deals in $D$, including all the favorable deals but no unfavorable ones,
\item the items of $N_i \mysetminus X_i^{N_i}(\mathbf{v})$ are allocated to players in $[n]\mysetminus \{i\}$ using the serial picking-exchange mechanism $f\big( N_i \mysetminus X_i^{N_i}(\mathbf{v})\big)$.
\end{itemize}
\end{enumerate}
\end{definition}

Clearly, serial picking-exchange mechanisms are a generalization of picking-exchange mechanisms studied in Section \ref{SEC:CHARACTERIZATION}.
The following example illustrates how such a mechanism looks like for three players.

\begin{example}\label{ex:3_players}
Suppose that we have three players with additive valuations. For simplicity, assume that each player's valuation induces a strict preference over all possible subsets of items.
Let $M= [100]$ be the set of items, and consider the following relevant ingredients of our mechanism:
\begin{itemize}
\item $N_1=\{1, 2, \ldots, 20\},\ \mathcal{O}_1=\{\{1, 2, 3\}, N_1\mysetminus \{1\}\}$
\item $N_2=\{21, 22, \ldots,  50\},\ \mathcal{O}_2=\{S\subseteq N_2 \ |\ |S|=6 \}$
\item $N_3=\{51, 52, \ldots, 70\},\ \mathcal{O}_3=\{\{51, \ldots, 60\}, \{61, \ldots, 70\}\}$
\item $E_1=\{71, \ldots, 80\},\ E_2=\{81, \ldots, 90\},\ E_3=\{91, \ldots, 100\}$
\item $D = \left\lbrace  \left[ (\{75, 79\}, \{83\})^{1, 3}\right] , \left[ (\{71\}, \{88\})^{1, 2}, (\{72, 80\}, \{95\})^{1, 3}, (\{85\}, \{99, 100\})^{2, 3}\right] \right\rbrace $
\item $f$ is a mapping from subsets of $M$ to picking-exchange mechanisms (for 2 players)
\end{itemize}
The above sets are the analog of the corresponding sets of a picking-exchange mechanism. The deals, however, are a bit more complex. E.g., by $\big[ (\{71\}, \{88\})^{1, 2},\allowbreak (\{72, 80\}, \{95\})^{1, 3},\allowbreak (\{85\}, \{99, 100\})^{2, 3}\big]$ we denote the deal in which:
\begin{itemize}
\item[--] player 1 gives item 71 to player 2 and items 72, 80 to player 3
\item[--] player 2 gives item 88 to player 1 and item 85 to player 3
\item[--] player 3 gives item 95 to player 1 and items 99, 100 to player 2
\end{itemize}
The mapping $f$ suggests which truthful mechanism should be used every time there are items left to be allocated to only two players.

We are ready to describe our mechanism $\mathcal{X}$:
\begin{enumerate}
\item The mechanism gives endowments $E_1, E_2, E_3$ to the three players and then performs each exchange deal that strictly improves all the players involved. 
\item Then, for each $i\in \{1, 2, 3\}$, the mechanism gives to player $i$ his best set in $\mathcal{O}_i$, say $S_i$.
\item Finally, for each $i\in \{1, 2, 3\}$, $\mathcal{X}$  uses mechanism $f(N_i \mysetminus S_i)$ to allocate the items of $N_i \mysetminus S_i$ to players in $\{1, 2, 3\}\mysetminus i$.
\end{enumerate}
\end{example}

Like picking-exchange mechanisms, serial picking-exchange mechanisms are truthful, given an appropriate tie-breaking rule (e.g., a label-based tie-breaking rule).
To bypass a general discussion about tie-breaking, however, we may assume that each player's valuation induces a strict preference over all  subsets of $M$. 
We denote by $\mathcal{V}^{\neq}_{\! n, m}$ the set of profiles that only include such valuation functions. Following almost the same proof, however, we have that for general additive valuations every serial picking-exchange mechanism is truthful when using label-based tie-breaking.

\begin{theorem}\label{thm:s/p/e_mechs}
When restricted to $\mathcal{V}^{\neq}_{\! n, m}$, every serial picking-exchange mechanism $\mathcal{X}$ for allocating $m$ items to $n$ players is truthful.
\end{theorem}

The proof is similar in spirit with the proof of Theorem \ref{thm:p/e_mech-->truthful}, and is deferred to the full version of the paper.

%

\section{Discussion}
We obtained a nontrivial characterization for truthful mechanisms, that has immediate implications on fairness.
A natural question to ask is whether our characterization can be extended for more than two players. Characterizing the truthful mechanisms without money for any number of additive players is, undoubtedly, a fundamental open problem. However, as indicated by Definition \ref{def:s/p/e}, there seems to be a much richer structure when one attempts to describe such mechanisms, even though serial picking-exchange mechanisms are only a subset of nonbossy truthful mechanisms.
In particular, the notion of control that was crucial for identifying the structure of truthful mechanisms for two players does not convey enough information anymore. Instead, there seem to exist several different levels of control, and understanding this structure still remains a very interesting and intriguing question.

\section*{Aknowledgements}	
This work has been partly supported by the COST Action IC1205 on Computational Social Choice, and by an internal grant of the Athens University of Economics and Business.
George Christo\-dou\-lou was supported by EPSRC EP/M008118/1 and Royal Society LT140046.	
We also wish to acknowledge the Simons institute for hosting the program on Economics and Computation, as some ideas and preliminary discussions began there.

%
\bibliographystyle{plainnat}
\bibliography{fairdivrefs}
%

%
\appendix
\section{{Missing Material from Section \ref{SEC:CHARACTERIZATION}}} \label{app:characterization}

\begin{proof}[\textbf{Proof of Theorem \ref{thm:p/e_mech-->truthful}}]
Assume $\mathcal{X}$ is a picking-exchange mechanism with partition $(N_1, N_2, E_1,\allowbreak E_2)$, offer sets $\mathcal{O}_i$ on $N_i$, for $i\in\{1, 2\}$, and set of exchange deals $D$.
Let $\mathbf{v}=(v_1,v_2)\in \mathcal{V}^{\neq}_{\! m}$ be a profile, and fix $v_2$. We are going to show that there is no $\mathbf{v'}=(v'_1,v_2)\in \mathcal{V}^{\neq}_{\! m}$ such that $v_1(X_1(\mathbf{v}'))> v_1(X_1(\mathbf{v}))$.

For any $\mathbf{v'}=(v'_1,v_2)\in \mathcal{V}^{\neq}_{\! m}$ there exist  the following  possibilities: \smallskip

\noindent (a) $X_1(\mathbf{v}') = X_1(\mathbf{v})$. Then clearly $v_1(X_1(\mathbf{v}')) = v_1(X_1(\mathbf{v}))$. \smallskip

\noindent (b) $X_1^{N_1\cup N_2}(\mathbf{v}') \neq X_1^{N_1\cup N_2}(\mathbf{v})$, but $X_1^{E_1\cup E_2}(\mathbf{v}') = X_1^{E_1\cup E_2}(\mathbf{v})$. Then it must be the case where $X_1^{N_1}(\mathbf{v}') \neq X_1^{N_1}(\mathbf{v})$. Indeed, player 1 has no power over $N_2$ where the items that he is allocated depend only on the unique best offer to player 2, i.e., $X_1^{N_2}(\mathbf{v}') = X_1^{N_2}(\mathbf{v})$. But this can only mean $v_1(X_1^{N_1}(\mathbf{v}')) < v_1(X_1^{N_1}(\mathbf{v}))$ by the definition of a picking-exchange mechanism and the fact that there are no subsets of equal value. So in total, $v_1(X_1(\mathbf{v}')) < v_1(X_1(\mathbf{v}))$. \smallskip
	
\noindent (c) $X_1^{N_1\cup N_2}(\mathbf{v}') = X_1^{N_1\cup N_2}(\mathbf{v})$, but $X_1^{E_1\cup E_2}(\mathbf{v}') \neq X_1^{E_1\cup E_2}(\mathbf{v})$. By the definition of picking-exchange mechanisms, player 1 can never force an exchange that is good for him but not for player 2. That is,  by deviating he will  lose one or more  exchanges that were good for him, and/or force one or more   exchanges that were bad for him.
We conclude it is the case where $v_1(X_1^{E_1\cup E_2}(\mathbf{v}')) < v_1(X_1^{E_1\cup E_2}(\mathbf{v}))$, and therefore, $v_1(X_1(\mathbf{v}')) < v_1(X_1(\mathbf{v}))$. \smallskip

\noindent (d) $X_1^{N_1\cup N_2}(\mathbf{v}') \neq X_1^{N_1\cup N_2}(\mathbf{v})$, and $X_1^{E_1\cup E_2}(\mathbf{v}') \neq X_1^{E_1\cup E_2}(\mathbf{v})$. By the fact that we are restricted to $\mathcal{V}^{\neq}_{\! m}$, we can derive that the ``picking part'' on $N_1\cup N_2$ and the ``exchange part'' on $E_1\cup E_2$ are independent. So, by cases (b) and (c) above we have $v_1(X_1^{N_1}(\mathbf{v}')) < v_1(X_1^{N_1}(\mathbf{v}))$ and  $v_1(X_1^{E_1\cup E_2}(\mathbf{v}')) < v_1(X_1^{E_1\cup E_2}(\mathbf{v}))$. Therefore, $v_1(X_1(\mathbf{v}')) < v_1(X_1(\mathbf{v}))$. \smallskip

We  conclude that every picking-exchange mechanism on  $\mathcal{V}^{\neq}_{\! m}$ is truthful.
\end{proof}


\begin{remark}\label{rem:ties}
With only slight modifications of the above proof, we have that for general additive valuations every picking-exchange mechanism is truthful when using the following two interesting families of tie-breaking rules: \smallskip

\noindent\emph{Tie-breaking with labels.} Every set in $\mathcal{O}_1 \cup \mathcal{O}_2$ has a distinct label, and whenever $\argmax_{S \in \mathcal{O}_i} v_i(S)$ is not a singleton, player $i$ receives the set with the smallest label. Further, every deal in $D$ has a label with five possible values, each indicating one of the following: (i) the exchange takes place every time it is not unfavorable, (ii) it only takes place every time it is not unfavorable and at least one player is strictly improved, (iii) it only takes place every time it is not unfavorable and player 1 is strictly improved, (iv) it only takes place every time it is not unfavorable and player 2 is strictly improved, and (v) it only takes place every time it is favorable.\smallskip

\noindent\emph{Welfare maximizing tie-breaking.} When $\argmax_{S \in \mathcal{O}_i} v_i(S)$ is not a singleton, player $i$ receives the set that leaves in $N_i$ as much value as possible for the other player. If there are still ties, labels are used to resolve those. Further, for every deal in $D$ the exchange takes place every time it is not unfavorable and at least one player is strictly improved.
\end{remark}
\medskip



\begin{proof}[\textbf{Proof of Corollary \ref{cor:part}}]
	From the definition of the $C_i$s and Corollary \ref{lem:controlsin}, $C_1 \cup  C_2= M$ follows. On the other hand, if $z \in C_1 \cap C_2$, then there exist a set $A\in \mathcal{A}_1$, such that $z\in A$, and a set  $B\in \mathcal{A}_2$ such that $z\in B$. By Lemma \ref{lem:control}, this implies that the singleton $\{z\}$ is controlled by both players, which is a contradiction. Thus, we have $C_1 \cap C_2= \emptyset$.
\end{proof}
\medskip


\begin{proof}[\textbf{Proof of Lemma \ref{output}}]
Due to symmetry, it suffices to prove the statement for $i=1$. If $N_1 = \emptyset$ then the statement is trivially true. So assume $N_1 \neq \emptyset$ and suppose that the statement does not hold. That is, there exists a profile $\mathbf{v}=(v_1,v_2)$ such that for any $S \in \mathcal{O}_1$ we have  $X_1^{N_1}(\mathbf{v}) \nsubseteq S$. 
This means $X_1^{N_1}(\mathbf{v}) \neq \emptyset$.
Since the sets in $\mathcal{O}_1$ cover $N_1$, there exists $S'$ such that $S'\cap X_1^{N_1}(\mathbf{v})\neq \emptyset$. Let $Z$ be a maximum cardinality such intersection between some $S' \in \mathcal{O}_1$ and $X_1^{N_1}(\mathbf{v})$, and $x$ be any element of $X_1^{N_1}(\mathbf{v}) \mysetminus Z$. Note that $x$ is guaranteed to exist since $X_1^{N_1}(\mathbf{v})$ is not contained in any set of $\mathcal{O}_1$. Also, there is no $S'' \in \mathcal{O}_1$ such that $Z \cup \{x\} \subseteq S''$ due to the maximality of $Z$. 

 The generic values that may appear in $\mathbf{v}$ restrict our ability to argue about the allocation, so our first goal is to reach a profile $\mathbf{u}$ that contradicts the lemma's statement, like $\mathbf{v}$, but has appropriately selected values. Then, having $\mathbf{u}$ as a starting point we can create profiles in which the allocations contradict truthfulness.

Now, recall that in  profile $\mathbf{v}$,
%
	  player 1 gets $Z \cup \{x\}$ (notice that he may get more items as well), and consider profiles $\mathbf{v'}=(v_1^{\textsc{i}},v_2)$ and $\mathbf{v''}=(v_1^{\textsc{ii}},v_2)$, where 
	\begin{equation*}
	\begin{IEEEeqnarraybox}[
	\IEEEeqnarraystrutmode
	\IEEEeqnarraystrutsizeadd{2pt}
	{0.5pt}
	][c]{c/v/c/v/c/v/c/}
	&& Z && x && M\mysetminus (Z \cup \{x\})  \\\hline
	v_1^{\textsc{i}} && \text{---}\ m^2\ \text{---} && m&& \text{---}\ 1\ \text{---}   
	\end{IEEEeqnarraybox}
	\end{equation*}
and
	\begin{equation*}
	\begin{IEEEeqnarraybox}[
	\IEEEeqnarraystrutmode
	\IEEEeqnarraystrutsizeadd{2pt}
	{0.5pt}
	][c]{c/v/c/v/c/v/c/}
	&& Z && x && M\mysetminus (Z \cup \{x\})  \\\hline
	v_1^{\textsc{ii}} && \text{---}\ m\ \text{---} && m^2 && \text{---}\ 1\ \text{---}   
	\end{IEEEeqnarraybox}
	\end{equation*}
By truthfulness, player 1 continues to get $Z \cup \{x\}$ in both cases, i.e., $Z \cup \{x\} \subseteq X_1^{N_1}(\mathbf{v'})$ and $Z \cup \{x\} \subseteq X_1^{N_1}(\mathbf{v''})$.

We proceed by changing the values of player 2 this time. Assuming that $M\mysetminus (Z \cup \{x\})=\{i_1, i_2, \ldots, \allowbreak i_\ell\}$ let $f_{i_j}=m$ if $i_j \in X_2(\mathbf{v''})$ and $f_{i_j}=1$ otherwise. Consider the next profile $\mathbf{u}=(v_1^{\textsc{ii}}, v_2^{\textsc{i}})$:
	\begin{equation*}
	\begin{IEEEeqnarraybox}[
	\IEEEeqnarraystrutmode
	\IEEEeqnarraystrutsizeadd{2pt}
	{0.5pt}
	][c]{c/v/c/v/c/v/c/}
	&& Z && x && M\mysetminus (Z \cup \{x\})  \\\hline
	v_1^{\textsc{ii}} && \text{---}\ m\ \text{---} && m^2&& \text{---}\ 1\ \text{---}    \\\hline
	v_2^{\textsc{i}} && \text{---}\ 1\ \text{---} && m^2  && f_{i_1}, \ldots, f_{i_\ell}
	\end{IEEEeqnarraybox}
	\end{equation*}
	Now notice that player 1 must get item $x$, since $x \in N_1$ and thus he controls $\{x\}$. On the other hand, since player 2 can not get  $x$ he must continue to get at least the items in $X_2(\mathbf{v''})$ by truthfulness (otherwise he would play $v_2$ instead). 
Since this the case, he can not get a strict superset of $X_2(\mathbf{v''})$ either. Indeed, if this was not the case he would deviate from $\mathbf{v''}$ to $\mathbf{u}$. So we can conclude that $X_2(\mathbf{u})=X_2(\mathbf{v''})$. 

Now we move to a profile  $\mathbf{u'}=(v_1^{\textsc{i}}, v_2^{\textsc{ii}})$ where eventually player 2 gets item $x$:
	\begin{equation*}
	\begin{IEEEeqnarraybox}[
	\IEEEeqnarraystrutmode
	\IEEEeqnarraystrutsizeadd{2pt}
	{0.5pt}
	][c]{c/v/c/v/c/v/c/}
	&& Z && x && M\mysetminus (Z \cup \{x\})  \\\hline
	v_1^{\textsc{i}} && \text{---}\ m^2\ \text{---} && m&& \text{---}\ 1\ \text{---}    \\\hline
	v_2^{\textsc{ii}} && \text{---}\ m^2\ \text{---} && m^2  && f_{i_1}, \ldots, f_{i_\ell}
	\end{IEEEeqnarraybox}
	\end{equation*}
	In $\mathbf{u'}$, both players strongly desire  $Z \cup \{x\}$. But 
player 1 cannot get both set $Z$ and item $x$, or by Lemma \ref{lem:control} he controls $Z \cup \{x\}$ and thus $Z \cup \{x\} \subseteq S$ for some $S \in \mathcal{O}_1$. However, he controls $Z$, since there exists some $S'\in \mathcal{O}_1$ such that  $Z = S'\cap X_1^{N_1}(\mathbf{v}) \subseteq S'$. So, player 1 has to get $Z$ since he strongly desires it, and item $x$  is given to player 2 (probably with other items in  $M\mysetminus (Z \cup \{x\}$). 

Finally, consider our final profile $\mathbf{u''}=(v_1^{\textsc{i}}, v_2^{\textsc{i}})$
	\begin{equation*}
	\begin{IEEEeqnarraybox}[
	\IEEEeqnarraystrutmode
	\IEEEeqnarraystrutsizeadd{2pt}
	{0.5pt}
	][c]{c/v/c/v/c/v/c/}
	&& Z && x && M\mysetminus (Z \cup \{x\})  \\\hline
	v_1^{\textsc{ii}} && \text{---}\ m^2\ \text{---} && m&& \text{---}\ 1\ \text{---}    \\\hline
	v_2^{\textsc{ii}} && \text{---}\ 1\ \text{---} && m^2  && f_{i_1}, \ldots, f_{i_\ell}
	\end{IEEEeqnarraybox}
	\end{equation*}
By truthfulness,  player 2 must get item $x$, or he would deviate from $\mathbf{u''}$ to $\mathbf{u'}$. However, now player 1 can strictly improve his utility by deviating from profile $\mathbf{u''}$ to $\mathbf{u}$, something that contradicts  truthfulness.
\end{proof}
\medskip


\begin{proof} [\textbf{Proof of Lemma \ref{best}}]
	Due to symmetry, it suffices to prove the statement for $i=1$. If $N_1 = \emptyset$ then the statement is trivially true. So assume $N_1 \neq \emptyset$ and suppose, towards a contradiction, that the statement does not hold. 
	That is, there exists a profile $\mathbf{v}=(v_1, v_2)$ such that 
$X_1^{N_1}(\mathbf{v}) \notin \argmax_{S\in \mathcal{O}_1} v_1(S)$. We consider two cases, depending on whether $X_2^{N_2}(\mathbf{v})$ is in $\mathcal{O}_2$ or not. In both cases, we create a series of deviations that eventually contradict truthfulness. Like in the proof of Lemma \ref{output}, our first goal is to reach a profile $\mathbf{u}$ that contradicts the statement, like $\mathbf{v}$, but has appropriately selected values. Using $\mathbf{u}$ as a starting point we create profiles in which the allocations dictated by truthfulness are in conflict.
	
	\noindent \textbf{Case 1.} Assume $X_2^{N_2}(\mathbf{v}) \in \mathcal{O}_2$ (note that this includes the case where $\mathcal{O}_2 = \{\emptyset\}$). Intuitively this is the case where the two players trade value between $N_1$ and $E_2$. 
	
	Consider the profile $\mathbf{v}'=(v_1, v_2^{\textsc{i}})$, where
	\begin{equation*}
	\begin{IEEEeqnarraybox}[
	\IEEEeqnarraystrutmode
	\IEEEeqnarraystrutsizeadd{2pt}
	{0.5pt}
	][c]{c/v/c/v/c/V/c/v/c/V/c/V/c/v/c}
	&& X_1^{N_1}(\mathbf{v}) && X_2^{N_1}(\mathbf{v})  && X_1^{N_2}(\mathbf{v}) && X_2^{N_2}(\mathbf{v})  && E_1 && X_1^{E_2}(\mathbf{v}) && X_2^{E_2}(\mathbf{v})   \\\hline
	v_2^{\textsc{i}} && \text{---}\ m\ \text{---} && \text{---}\ m^3\ \text{---}  && \text{---}\ m\ \text{---} && \text{---}\ m^4\ \text{---}  && \text{---}\ 1 \ \text{---} && \text{---}\ m^2 \ \text{---} && \text{---}\ m^4 \ \text{---}
	\end{IEEEeqnarraybox}
	\end{equation*}
	By truthfulness, $X_2(\mathbf{v}') \supseteq X_2^{N_1}(\mathbf{v}) \cup X_2^{N_2}(\mathbf{v}) \cup X_2^{E_2}(\mathbf{v})$. This implies $X_2^{N_2}(\mathbf{v}') = X_2^{N_2}(\mathbf{v})$ due to the maximality of $X_2^{N_2}(\mathbf{v})$ and Lemma \ref{output}, as well as $X_1^{N_1}(\mathbf{v}')\subseteq X_1^{N_1}(\mathbf{v})$. The latter implies that $X_1^{N_1}(\mathbf{v}') \notin \argmax_{S\in \mathcal{O}_1} v_1(S)$. 
	
	\begin{claim}\label{cl:n1a}
		$X_1^{E_2}(\mathbf{v}')\neq \emptyset$.
	\end{claim}
	
	\begin{proof}[Proof of Claim \ref{cl:n1a}]\renewcommand{\qedsymbol}{$\triangleleft$}
		Suppose $X_1^{E_2}(\mathbf{v}') = \emptyset$ and let $S'\in \argmax_{S\in \mathcal{O}_1} v_1(S)$. Then player 1, whose total received value in $\mathbf{v}'$ would be strictly less than $v_1(S'\cup (X_1^{N_2}(\mathbf{v}))\cup E_1)$, could force the mechanism to give him at least that by playing  
		\begin{equation*}
		\begin{IEEEeqnarraybox}[
		\IEEEeqnarraystrutmode
		\IEEEeqnarraystrutsizeadd{2pt}
		{0.5pt}
		][c]{c/v/c/v/c/V/c/V/c/V/c}
		&& S' && N_1 \mysetminus S'  &&  N_2  && E_1 && E_2    \\\hline
		v_1^{\textsc{i}} && \text{---}\ m\ \text{---}   && \text{---}\ 1\ \text{---} && \text{---}\ 1\ \text{---}  && \text{---}\ m \ \text{---} && \text{---}\ 1 \ \text{---} 
		\end{IEEEeqnarraybox}
		\end{equation*}
		By the definition of $N_1$, $N_2$, $E_1$, and Lemma \ref{output}, player 1 gets $S'$, $N_2\mysetminus X_2^{N_2}(\mathbf{v})$, and $E_1$ (and possibly something from $E_2$). Since this contradicts truthfulness, it must be the case that $X_1^{E_2}(\mathbf{v}')\neq \emptyset$. (In fact, this settles Case 1 when $E_2=\emptyset$.)
	\end{proof}

	Next, let $S_1\in \mathcal{O}_1$ be such that $X_1^{N_1}(\mathbf{v}')\subseteq S_1$ (they could possibly be equal). Consider the profile $\mathbf{u}=(v_1^{\textsc{ii}}, v_2^{\textsc{i}})$, where
	\begin{equation*}
	\begin{IEEEeqnarraybox}[
	\IEEEeqnarraystrutmode
	\IEEEeqnarraystrutsizeadd{2pt}
	{0.5pt}
	][c]{c/v/c/v/c/v/c/V/c/V/c/v/c/V/c/v/c}
	&& X_1^{N_1}(\mathbf{v}') && S_1 \mysetminus X_1^{N_1}(\mathbf{v}')  && N_1 \mysetminus S_1 && N_2  && X_1^{E_1}(\mathbf{v}') && X_2^{E_1}(\mathbf{v}') && X_1^{E_2}(\mathbf{v}') && X_2^{E_2}(\mathbf{v}')   \\\hline
	v_1^{\textsc{ii}} && \text{---}\  m \ \text{---} && \text{---}\ 1 + m^{-1}\ \text{---}  && \text{---}\ 1\ \text{---} && \text{---}\ 1 \ \text{---}  && \text{---}\ m^3 \ \text{---} && \text{---}\ 1 \ \text{---} && \text{---}\ m^2 \ \text{---} && \text{---}\ 1 \ \text{---}
	\end{IEEEeqnarraybox}
	\end{equation*}
	Notice that $S_1$ is the unique set in $\argmax_{S\in \mathcal{O}_1} v_1^{\textsc{ii}}(S)$. By truthfulness, $X_1(\mathbf{u}) \supseteq X_1^{N_1}(\mathbf{v}') \cup X_1^{E_1}(\mathbf{v}') \cup X_1^{E_2}(\mathbf{v}')$. 
	
	\begin{claim}\label{cl:n1b}
		$S_1 \nsubseteq X_1(\mathbf{u})$, and therefore $X_1^{N_1}(\mathbf{u})\notin \argmax_{S\in \mathcal{O}_1} v_1^{\textsc{ii}}(S)$.
	\end{claim}
	
	\begin{proof}[Proof of Claim \ref{cl:n1b}]\renewcommand{\qedsymbol}{$\triangleleft$}
		Suppose $S_1 \subseteq X_1(\mathbf{u})$. By Lemma \ref{output} this means $S_1 = X_1^{N_1}(\mathbf{u})$. Then player 2, whose total received value in $\mathbf{u}$ would be strictly less than $v_2^{\textsc{ii}}((N_1\mysetminus S_1)\cup X_2^{N_2}(\mathbf{v}')\cup X_2^{E_2}(\mathbf{v}')) + m$, could force the mechanism to give him more than that by playing  
		\begin{equation*}
		\begin{IEEEeqnarraybox}[
		\IEEEeqnarraystrutmode
		\IEEEeqnarraystrutsizeadd{2pt}
		{0.5pt}
		][c]{c/v/c/V/c/v/c/V/c/V/c}
		&& N_1 && X_1^{N_2}(\mathbf{v}')  &&  X_2^{N_2}(\mathbf{v}')  && E_1 && E_2    \\\hline
		v_2^{\textsc{ii}} && \text{---}\ 1\ \text{---}   && \text{---}\ 1\ \text{---} && \text{---}\ m\ \text{---}  && \text{---}\ 1 \ \text{---} && \text{---}\ m \ \text{---} 
		\end{IEEEeqnarraybox}
		\end{equation*}
		By the definition of $N_2$, $E_2$, in $\mathbf{v}''=(v_1^{\textsc{ii}}, v_2^{\textsc{ii}})$ player 2 gets $X_2^{N_2}(\mathbf{v}')$ and $E_2$ (and possibly something from $N_1$ and $E_1$). Given that, the maximum value that player 1 could achieve in $\mathbf{v}''$ is $v_1^{\textsc{ii}}(S_1 \cup X_1^{N_2}(\mathbf{v}') \cup E_1)$ and there is no subset of $M\mysetminus (X_2^{N_2}(\mathbf{v}')\cup E_2)$ giving this value other than $S_1 \cup X_1^{N_2}(\mathbf{v}') \cup E_1$. In fact,  player 1 can achieve exactly this by increasing his reported value for each item in $S_1\cup E_1$ to $m^3$. Thus $X_1(\mathbf{v}'')= S_1\cup X_1^{N_2}(\mathbf{v}') \cup E_1$ and $v_2^{\textsc{ii}}(X_2(\mathbf{v}'')) = v_2^{\textsc{ii}}((N_1\mysetminus S_1) \cup X_2^{N_2}(\mathbf{v}') \cup E_2) \ge v_2^{\textsc{ii}}((N_1\mysetminus S_1) \cup X_2^{N_2}(\mathbf{v}') \cup X_2^{E_2}(\mathbf{v}')) + m^2$. Since this contradicts truthfulness, it must be the case that $S_1 \nsubseteq X_1(\mathbf{u})$ (and thus $X_1^{N_1}(\mathbf{u})\notin \argmax_{S\in \mathcal{O}_1} v_1^{\textsc{ii}}(S)$).
	\end{proof}

	Claim \ref{cl:n1b} implies that $S_1 \mysetminus X_1^{N_1}(\mathbf{u}) \neq \emptyset$. Since the sets in $\mathcal{O}_1$ have empty intersection, there must exist some $T\in \mathcal{O}_1$ such that $S_1 \mysetminus X_1^{N_1}(\mathbf{u}) \nsubseteq T$. We are going to concentrate most of player 2's value from $N_1$ on $W = (S_1 \mysetminus X_1^{N_1}(\mathbf{u}))\mysetminus T \subseteq X_2^{N_1}(\mathbf{u})$. Notice that $W\neq  \emptyset$.
	
	So consider the profile $\mathbf{u}'=(v_1^{\textsc{ii}}, v_2^{\textsc{iii}})$, where
	\begin{equation*}
	\begin{IEEEeqnarraybox}[
	\IEEEeqnarraystrutmode
	\IEEEeqnarraystrutsizeadd{2pt}
	{0.5pt}
	][c]{c/v/c/v/c/V/c/v/c/V/c/V/c/v/c}
	&& N_1\mysetminus W && W  && X_1^{N_2}(\mathbf{v}) && X_2^{N_2}(\mathbf{v})  && E_1 && X_1^{E_2}(\mathbf{v}') && X_2^{E_2}(\mathbf{v}')   \\\hline
	v_2^{\textsc{iii}} && \text{---}\ m\ \text{---} && \text{---}\ m^3\ \text{---}  && \text{---}\ m\ \text{---} && \text{---}\ m^4\ \text{---}  && \text{---}\ 1 \ \text{---} && \text{---}\ m^2 \ \text{---} && \text{---}\ m^4 \ \text{---}
	\end{IEEEeqnarraybox}
	\end{equation*}
	By the definition of $N_2$, $E_2$ and truthfulness, $X_2(\mathbf{u}') \supseteq W \cup X_2^{N_2}(\mathbf{v}) \cup X_2^{E_2}(\mathbf{v}')$.
	
	\begin{claim}\label{cl:n1c}
		$X_1^{E_2}(\mathbf{u}')\neq \emptyset$.
	\end{claim}
	
	\begin{proof}[Proof of Claim \ref{cl:n1c}]\renewcommand{\qedsymbol}{$\triangleleft$}
		This is very similar to the proof of Claim \ref{cl:n1a}. Suppose $X_1^{E_2}(\mathbf{u}') = \emptyset$. Then player 1, whose total received value in $\mathbf{u}'$ would be strictly less than $v_1^{\textsc{ii}}(S_1 \cup X_1^{N_2}(\mathbf{v}) \cup E_1)$, could force the mechanism to give him at least that by playing  
		\begin{equation*}
		\begin{IEEEeqnarraybox}[
		\IEEEeqnarraystrutmode
		\IEEEeqnarraystrutsizeadd{2pt}
		{0.5pt}
		][c]{c/v/c/v/c/V/c/V/c/V/c}
		&& S_1 && N_1 \mysetminus S_1  &&  N_2  && E_1 && E_2    \\\hline
		v_1^{\textsc{iii}} && \text{---}\ m\ \text{---}   && \text{---}\ 1\ \text{---} && \text{---}\ 1\ \text{---}  && \text{---}\ m \ \text{---} && \text{---}\ 1 \ \text{---} 
		\end{IEEEeqnarraybox}
		\end{equation*}
		Since this contradicts truthfulness, it must be the case where $X_1^{E_2}(\mathbf{u}')\neq \emptyset$.
	\end{proof}
	
	
	Before we examine the final profile of the proof, let us consider the following simple profile $\mathbf{u}''=(v_1^{\textsc{iv}}, v_2^{\textsc{iv}})$:
	\begin{equation*}
	\begin{IEEEeqnarraybox}[
	\IEEEeqnarraystrutmode
	\IEEEeqnarraystrutsizeadd{2pt}
	{0.5pt}
	][c]{c/v/c/v/c/V/c/v/c/V/c/V/c/v/c}
	&& T && N_1\mysetminus T && X_1^{N_2}(\mathbf{v}) && X_2^{N_2}(\mathbf{v})  && E_1 && X_1^{E_2}(\mathbf{u}') && X_2^{E_2}(\mathbf{u}') \\\hline
	v_1^{\textsc{iv}} && \text{---}\ m\ \text{---} && \text{---}\ 1\ \text{---} && \text{---}\ 1\ \text{---}  && \text{---}\ 1\ \text{---}  && \text{---}\ 1\ \text{---}  && \text{---}\ m^2\ \text{---} && \text{---}\ 1 \ \text{---}   \\\hline
	v_2^{\textsc{iv}} && \text{---}\ 1\ \text{---} && \text{---}\ 1\ \text{---}  && \text{---}\ 1\ \text{---} && \text{---}\ m\ \text{---}  && \text{---}\ 1\ \text{---} && \text{---}\ m \ \text{---} && \text{---}\ m \ \text{---}
	\end{IEEEeqnarraybox}
	\end{equation*}
	By the definition of $N_2$, $E_2$, in $\mathbf{u}''$ player 2 gets $X_2^{N_2}(\mathbf{v})$ and $E_2$ (and possibly something from $N_1$ and $E_1$). Given that, the maximum value that player 1 could achieve in $\mathbf{u}''$ is $|T|\cdot m + |X_1^{N_2}(\mathbf{v}) \cup E_1|$, and there is no subset of $M\mysetminus (X_2^{N_2}(\mathbf{v})\cup E_2)$ giving this value other than $T \cup X_1^{N_2}(\mathbf{v}) \cup E_1$. In fact, player 1 can achieve exactly this by increasing his reported value for each item in $T\cup E_1$ to $m^3$. Thus $X_1(\mathbf{u}'')= T\cup X_1^{N_2}(\mathbf{v}) \cup E_1$ and $X_2(\mathbf{u}'') = (N_1\mysetminus T)\cup X_2^{N_2}(\mathbf{v}) \cup E_2$.
	
	The final profile we need is $\mathbf{u}'''=(v_1^{\textsc{iv}}, v_2^{\textsc{iii}})$, and the contradiction follows from the allocation of the items in $X_1^{E_2}(\mathbf{u}')$. If $X_1^{E_2}(\mathbf{u}') \nsubseteq X_1(\mathbf{u}''')$ then player 1 has incentive to deviate to profile $\mathbf{u}'=(v_1^{\textsc{ii}}, v_2^{\textsc{iii}})$. So, it must be the case where $X_1^{E_2}(\mathbf{u}') \subseteq X_1(\mathbf{u}''')$, and therefore $v_2^{\textsc{iii}}(X_2(\mathbf{u}''')) \le v_2^{\textsc{iii}}(M\mysetminus X_1^{N_2}(\mathbf{u}')) < v_2^{\textsc{iii}}(W \cup X_2^{N_2}(\mathbf{v}) \cup X_2^{E_2}(\mathbf{u}')) +m^2$. On the other hand, notice that $W\subseteq N_1 \mysetminus T$. Using the allocation for $\mathbf{u}''$ we derived above, by truthfulness we have that $v_2^{\textsc{iii}}(X_2(\mathbf{u}''')) \ge v_2^{\textsc{iii}}(W \cup X_2^{N_2}(\mathbf{v}) \cup E_2) \ge v_2^{\textsc{iii}}(W \cup X_2^{N_2}(\mathbf{v}) \cup X_2^{E_2}(\mathbf{u}')) +m^2$, which is a contradiction.  \smallskip


	\noindent \textbf{Case 2.} Assume $X_2^{N_2}(\mathbf{v}) \notin \mathcal{O}_2$. Case 1 implies that not only $X_1^{N_1}(\mathbf{v}) \notin \argmax_{S\in \mathcal{O}_1} v_1(S)$ but $X_1^{N_1}(\mathbf{v}) \notin \mathcal{O}_1$. Intuitively this is the case where the two players trade value between $N_1$ and $N_2$. The proof uses a  sequence of profiles similar to Case 1.  
	
	Consider the profile $\mathbf{v}'=(v_1, v_2^{\textsc{i}})$, where
	\begin{equation*}
	\begin{IEEEeqnarraybox}[
	\IEEEeqnarraystrutmode
	\IEEEeqnarraystrutsizeadd{2pt}
	{0.5pt}
	][c]{c/v/c/v/c/V/c/v/c/V/c/V/c/}
	&& X_1^{N_1}(\mathbf{v}) && X_2^{N_1}(\mathbf{v})  && X_1^{N_2}(\mathbf{v}) && X_2^{N_2}(\mathbf{v})  && E_1 && E_2   \\\hline
	v_2^{\textsc{i}} && \text{---}\ 1\ \text{---} && \text{---}\ m^2\ \text{---}  && \text{---}\ m\ \text{---} && \text{---}\ m^3\ \text{---}  && \text{---}\ 1 \ \text{---} && \text{---}\ 1 \ \text{---} 
	\end{IEEEeqnarraybox}
	\end{equation*}
	By truthfulness, $X_2(\mathbf{v}') \supseteq X_2^{N_1}(\mathbf{v}) \cup X_2^{N_2}(\mathbf{v})$. This implies  $X_1^{N_1}(\mathbf{v}')\subseteq X_1^{N_1}(\mathbf{v})$, and thus $X_1^{N_1}(\mathbf{v}') \notin \mathcal{O}_1$. By Case 1, this means that $X_2^{N_2}(\mathbf{v}') \notin \mathcal{O}_2$.

	Next, let $S_1\in \mathcal{O}_1$ be a minimal set of $ \mathcal{O}_1$ such that $X_1^{N_1}(\mathbf{v}')\subseteq S_1$. Since $X_1^{N_1}(\mathbf{v}') \notin \mathcal{O}_1$, we have $X_1^{N_1}(\mathbf{v}')\subsetneq S_1$. Consider the profile $\mathbf{u}=(v_1^{\textsc{i}}, v_2^{\textsc{i}})$, where
	\begin{equation*}
	\begin{IEEEeqnarraybox}[
	\IEEEeqnarraystrutmode
	\IEEEeqnarraystrutsizeadd{2pt}
	{0.5pt}
	][c]{c/v/c/v/c/v/c/V/c/v/c/V/c/V/c}
	&& X_1^{N_1}(\mathbf{v}') && S_1 \mysetminus X_1^{N_1}(\mathbf{v}')  && N_1 \mysetminus S_1 && X_1^{N_2}(\mathbf{v}') && X_2^{N_2}(\mathbf{v}')  && E_1  && E_2  \\\hline
	v_1^{\textsc{i}} && \text{---}\  m \ \text{---} && \text{---}\ 1 + m^{-1}\ \text{---}  && \text{---}\ 1\ \text{---} && \text{---}\ m \ \text{---} && \text{---}\ 1 \ \text{---}   && \text{---}\ 1 \ \text{---}  && \text{---}\ 1 \ \text{---}
	\end{IEEEeqnarraybox}
	\end{equation*}
	Notice that $S_1$ is the unique set in $\argmax_{S\in \mathcal{O}_1} v_1^{\textsc{i}}(S)$. 
	By truthfulness, $X_1(\mathbf{u}) \supseteq X_1^{N_1}(\mathbf{v}') \cup X_1^{N_2}(\mathbf{v}')$. 
	
	\begin{claim}\label{cl:n1d}
		$S_1 \nsubseteq X_1(\mathbf{u})$, and therefore $X_1^{N_1}(\mathbf{u})\notin \argmax_{S\in \mathcal{O}_1} v_1^{\textsc{i}}(S)$.
	\end{claim}
	
	\begin{proof}[Proof of Claim \ref{cl:n1d}]\renewcommand{\qedsymbol}{$\triangleleft$}
		This is similar to the proof of Claim \ref{cl:n1b}.Suppose $S_1 \subseteq X_1(\mathbf{u})$. By Lemma \ref{output} this means $S_1 = X_1^{N_1}(\mathbf{u})$. Then player 2, whose total received value in $\mathbf{u}$ would be strictly less than $v_2^{\textsc{i}}(X_2^{N_2}(\mathbf{v}')) + m$, could force the mechanism to give him at least that by playing  
		\begin{equation*}
		\begin{IEEEeqnarraybox}[
		\IEEEeqnarraystrutmode
		\IEEEeqnarraystrutsizeadd{2pt}
		{0.5pt}
		][c]{c/v/c/V/c/v/c/V/c/V/c}
		&& N_1 && N_2 \mysetminus S_2 &&  S_2 && E_1 && E_2    \\\hline
		v_2^{\textsc{ii}} && \text{---}\ 1\ \text{---}    && \text{---}\ 1\ \text{---}  && \text{---}\ m\ \text{---}  && \text{---}\ 1 \ \text{---} && \text{---}\ m \ \text{---} 
		\end{IEEEeqnarraybox}
		\end{equation*}
		where $S_2\in \mathcal{O}_2$ is such that $X_2^{N_2}(\mathbf{v}')\subseteq S_2$.  
		By the definition of $N_2$, $E_2$, in $\mathbf{v}''=(v_1^{\textsc{i}}, v_2^{\textsc{ii}})$ player 2 gets $S_2$ and $E_2$ (and possibly something from $N_1$ and $E_1$). Note, however, that $X_2^{N_2}(\mathbf{v}') \notin \mathcal{O}_2$ and thus $X_2^{N_2}(\mathbf{v}')\subsetneq S_2$. Therefore, $v_2^{\textsc{i}}(X_2(\mathbf{v}'')) \ge v_2^{\textsc{i}}(S_2) \ge v_2^{\textsc{i}}(X_2^{N_2}(\mathbf{v}')) + m$.
		Since this contradicts truthfulness, it must be the case that $S_1 \nsubseteq X_1(\mathbf{u})$ (and thus $X_1^{N_1}(\mathbf{u})\notin \argmax_{S\in \mathcal{O}_1} v_1^{\textsc{i}}(S)$).
	\end{proof}
	

	This implies that $S_1 \mysetminus X_1^{N_1}(\mathbf{u}) \neq \emptyset$. Since the sets in $\mathcal{O}_1$ have empty intersection, there must exist some $T\in \mathcal{O}_1$ such that $S_1 \mysetminus X_1^{N_1}(\mathbf{u}) \nsubseteq T$. We are going to concentrate most of player 2's value from $N_1$ on $W = (S_1 \mysetminus X_1^{N_1}(\mathbf{u}))\mysetminus T \neq \emptyset$. So consider the profile $\mathbf{u}'=(v_1^{\textsc{i}}, v_2^{\textsc{iii}})$, where
	\begin{equation*}
	\begin{IEEEeqnarraybox}[
	\IEEEeqnarraystrutmode
	\IEEEeqnarraystrutsizeadd{2pt}
	{0.5pt}
	][c]{c/v/c/v/c/V/c/v/c/V/c/V/c/}
	&& N_1\mysetminus W && W  && X_1^{N_2}(\mathbf{u}) && X_2^{N_2}(\mathbf{u})  && E_1 && E_2   \\\hline
	v_2^{\textsc{iii}} && \text{---}\ 1\ \text{---} && \text{---}\ m^2\ \text{---}  && \text{---}\ m\ \text{---} && \text{---}\ m^3\ \text{---}  && \text{---}\ 1 \ \text{---} && \text{---}\ 1 \ \text{---} 
	\end{IEEEeqnarraybox}
	\end{equation*}
	
\sloppy	By truthfulness, $X_2(\mathbf{u}') \supseteq W \cup X_2^{N_2}(\mathbf{u})$. This implies that $S_1 \nsubseteq X_1(\mathbf{u}')$ and thus $X_1^{N_1}(\mathbf{u}')\notin \allowbreak \argmax_{S\in \mathcal{O}_1} v_1^{\textsc{i}}(S)$. By Case 1, this means that $X_2^{N_2}(\mathbf{u}') \notin \mathcal{O}_2$. Therefore, $X_1^{N_2}(\mathbf{u}') \neq \emptyset$.

	Now let $S'_2\in \mathcal{O}_2$ is such that $X_2^{N_2}(\mathbf{u}')\subsetneq S'_2$.
	Before we examine the final profile of the proof, let us consider the following profile $\mathbf{u}''=(v_1^{\textsc{ii}}, v_2^{\textsc{iv}})$:
	\begin{equation*}
	\begin{IEEEeqnarraybox}[
	\IEEEeqnarraystrutmode
	\IEEEeqnarraystrutsizeadd{2pt}
	{0.5pt}
	][c]{c/v/c/v/c/V/c/v/c/v/c/V/c/V/c}
	&& T && N_1\mysetminus T &&  N_2\mysetminus S'_2 && S'_2 \mysetminus X_2^{N_2}(\mathbf{u}')  && X_2^{N_2}(\mathbf{u}')  && E_1 && E_2 \\\hline
	v_1^{\textsc{ii}} && \text{---}\ m\ \text{---} && \text{---}\ 1\ \text{---} && \text{---}\ m^2 \ \text{---}  && \text{---}\ m^2\ \text{---}  && \text{---}\ 1\ \text{---} && \text{---}\ 1\ \text{---}  && \text{---}\ 1 \ \text{---}    \\\hline
	v_2^{\textsc{iv}} && \text{---}\ 1\ \text{---} && \text{---}\ 1\ \text{---}  && \text{---}\ 1\ \text{---} && \text{---}\ m\ \text{---} && \text{---}\ m\ \text{---}  && \text{---}\ 1\ \text{---} && \text{---}\ m \ \text{---} 
	\end{IEEEeqnarraybox}
	\end{equation*}
	By the definition of $N_2$, $E_2$, in $\mathbf{u}''$ player 2 gets $S'_2$ and $E_2$ (and possibly something from $N_1$ and $E_1$). Given that, the maximum value that player 1 could achieve in $\mathbf{u}''$ is $|T|\cdot m + |N_2\mysetminus S'_2|\cdot m^2 + |E_1|$. In fact, player 1 can achieve exactly this by increasing his reported value for each item in $T\cup E_1$ to $m^3$. 
	Thus $X_1(\mathbf{u}'')= T\cup (N_2\mysetminus S'_2) \cup E_1$ and $X_2(\mathbf{u}'') = (N_1\mysetminus T)\cup S'_2 \cup E_2$.
	
	The final profile we need is $\mathbf{u}'''=(v_1^{\textsc{ii}}, v_2^{\textsc{iii}})$, and the contradiction follows from the allocation of the items in $X_1^{N_2}(\mathbf{u}')$. If $X_1^{N_2}(\mathbf{u}') \nsubseteq X_1(\mathbf{u}''')$ then player 1 has incentive to deviate to profile $\mathbf{u}'=(v_1^{\textsc{i}}, v_2^{\textsc{iii}})$. So, it must be the case where $X_1^{N_2}(\mathbf{u}') \subseteq X_1(\mathbf{u}''')$ and therefore $v_2^{\textsc{iii}}(X_2(\mathbf{u}''')) \le v_2^{\textsc{iii}}(M\mysetminus X_1^{N_2}(\mathbf{u}')) < |W| \cdot m^2 + |X_2^{N_2}(\mathbf{u})| \cdot m^3 +m$. On the other hand, notice that $W\subseteq N_1 \mysetminus T$ and recall that $X_2^{N_2}(\mathbf{u})\subseteq X_2^{N_2}(\mathbf{u}')\subsetneq S'_2$. Using the allocation for $\mathbf{u}''$ we calculated above, by truthfulness we have that $v_2^{\textsc{iii}}(X_2(\mathbf{u}''')) \ge v_2^{\textsc{iii}}((N_1 \mysetminus T)\cup S'_2) \ge |W| \cdot m^2 + |X_2^{N_2}(\mathbf{u})| \cdot m^3 +m$, which is a contradiction. 
\end{proof}
\medskip



\begin{proof}[\textbf{Proof of Lemma \ref{indie}}]
	Suppose that this not true. So there are profiles $\mathbf{v}=(v_1,v_2),  \mathbf{v}'=(v'_1,v'_2) \in \mathcal{V}^{\neq}_{\! m}$ such that $v_{ij}=v'_{ij}$ for all $i\in \{1, 2\}$ and $j\in E_1 \cup E_2$, but $X_1^{E_1 \cup E_2}(\mathbf{v}) \neq X_1^{E_1 \cup E_2}(\mathbf{v}')$. In such a case, either $\mathbf{v}=(v_1,v_2), \hat{\mathbf{v}}=(v'_1,v_2)$, or $ \hat{\mathbf{v}}=(v'_1,v_2),\mathbf{v}'=(v'_1,v'_2)$ is also  a pair of profiles that violates the statement. Without loss of generality we assume that $\mathbf{v}, \hat{\mathbf{v}}$ is such a pair, and  that $v_1(X_1^{E_1}(\mathbf{v}))>v_1(X_1^{E_1}(\hat{\mathbf{v}}))$. Now let $S_1, \hat{S}_1\in \mathcal{O}_1$ be the single best offer in each case. If $S_1= \hat{S}_1$ then player 1 would deviate from $\hat{\mathbf{v}}$ to $\mathbf{v}$ and  strictly improve. So assume that $S_1\neq \hat{S}_1$ and  multiply the values in $E_1 \cup E_2$ for  player 1 with a large enough constant $K$, so that  $K\big(\hat{v}_1(X_1^{E_1}(\mathbf{v}))-\hat{v}_1(X_1^{E_1}(\hat{\mathbf{v}}))\big)> \hat{v}_1( N_1 \cup  N_2)$. 

Call $\mathbf{v}^*=(v^*_1,v_2)$ and $\hat{\mathbf{v}}^*=(v'^*_1,v_2)$ the new profiles and notice that they are still in $\mathcal{V}^{\neq}_{\! m}$. Also, it is easy to see that truthfulness implies $X_1(\mathbf{v})=X_1(\mathbf{v}^*)$ and $X_1(\hat{\mathbf{v}})=X_1(\hat{\mathbf{v}}^*)$.
 Indeed, by Lemma \ref{best}, we have $X_1^{N_1\cup N_2}(\mathbf{v})=X_1^{N_1\cup N_2}(\mathbf{v}^*)$, and if it was the case where $X_1^{E_1\cup E_2}(\mathbf{v}) \neq X_1^{E_1\cup E_2}(\mathbf{v}^*)$, then player 1 would deviate from profile $\mathbf{v}$ to $\mathbf{v}^*$ or vice versa to strictly improve his utility. The same holds for $\hat{\mathbf{v}}$ to $\hat{\mathbf{v}}^*$. 

Now, however, player 1 would deviate from $\hat{\mathbf{v}}^*$ to $\mathbf{v}^*$ in order to  improve by at least $\hat{v}^*_1(X_1^{E_1}(\mathbf{v}^*))-\hat{v}^*_1(X_1^{E_1}(\hat{\mathbf{v}}^*)) - \hat{v}^*_1( N_1 \cup  N_2) = K\big(\hat{v}_1(X_1^{E_1}(\mathbf{v}))-\hat{v}_1(X_1^{E_1}(\hat{\mathbf{v}}))\big) - \hat{v}_1( N_1 \cup  N_2) >0$, and this contradicts truthfulness.
\end{proof}
\medskip

\remark \label{rem:Xetc}
Since we are talking about $\mathcal{X}_E$ in many of the following proofs, it is  correct to write $X_i^{E_1\cup E_2}(\cdot)$, not  $X_i(\cdot)$. For the sake of readability, though, we drop the superscript wherever it is not necessary.  
Similarly, in order to avoid the unnecessary use of extra symbols, we prove the statements for $m$ items, although in Subsection \ref{subsec:characterization}  $\mathcal{X}_E$ is a mechanism on $\ell \le m$ items.
\medskip

\remark \label{rem:notequal}
For most of the following proofs we need to construct profiles in $\mathcal{V}^{\neq}_{\! m}$. To facilitate the presentation, however, the valuation functions we construct only use a few powers of $m$. As a result, the corresponding  profiles typically are not in $\mathcal{V}^{\neq}_{\! m}$. Still, this is without loss of generality; when defining such valuation functions we can add $2^i/2^{\kappa}$ to the value of item $i$, for $i\in [m]$.  When $\kappa\in \mathbb N$ is large enough (usually $\kappa=m+1$ suffices), our arguments about the allocation are not affected, and a strict preference over all subsets is induced. 
\medskip

\begin{proof}[\textbf{Proof of Lemma \ref{strict}}]
Let $\mathbf{v}=(v_1,v_2)\in \mathcal{V}_{\! m}$, and consider the intermediate profile  $\mathbf{v}^*=(v'_1,v_2)$ where $v'_{1x}=m$, if $x \in X_1(\mathbf{v})$, and $v'_{1x}=1$ otherwise. By truthfulness, we have that $X_1(\mathbf{v}^*)=X_1(\mathbf{v})$. By defining $v'_2$ in a similar way (i.e., $v'_{2x}=m$, if $x \in X_2(\mathbf{v})$, and $v'_{2x}=1$ otherwise), we get the profile  $\mathbf{v}'=(v'_1,v'_2)$. Again by truthfulness, we have $\mathcal{X}(\mathbf{v}')=\mathcal{X}(\mathbf{v})$. If $\mathbf{v}^*$ and $\mathbf{v}'$ where defined as described in Remark \ref{rem:notequal}, the same arguments would apply, and moreover, $\mathbf{v}'\in \mathcal{V}^{\neq}_{\! m}$.
\end{proof}
\medskip


\begin{proof}[\textbf{Proof of Lemma \ref{lem:proper}}]
To show that $D$ is indeed a valid set of exchange deals, we need to show that for any two distinct deals $(S, T), (S', T')\in D$ we have $S\cap S'=T\cap T'=\emptyset$ and $S, T, S', T'$ are all nonempty. The latter is straightforward due to truthfulness and the fact that all values are positive. The former is done through the next three lemmata, the first of which states that each minimally exchangeable set is involved in exactly one exchange deal. 

\begin{lemma}\label{lem:fes}
If $S \subseteq E_1$ is a minimally exchangeable set, then there exists a unique $T \subseteq E_2$ such that $(S, T)$ is a feasible exchange.
\end{lemma}

The lemma is stated in terms of minimally exchangeable subsets of $E_1$, but due to symmetry it is true for all minimally exchangeable sets. This is done for the following statements as well, for the sake of readability. The   three lemmata are proved right after this proof.

It is implied that  every minimally exchangeable set appears in exactly one exchange deal in $D$. The second lemma, below, guarantees that minimally exchangeable sets can be exchanged only with minimally exchangeable sets. 

\begin{lemma}\label{lem:min}
Let $S \subseteq E_1$ be a minimally exchangeable set and $(S, T)$ be the only feasible exchange involving $S$. Then $T$ is a minimally exchangeable set as well. 
\end{lemma}

The result of the two lemmata combined is that $D=\{(S_1, T_1), (S_2, T_2), \ldots, (S_{k}, T_{k})\}$, where $S_1, \ldots,\allowbreak S_{k}, T_1, \ldots, T_{k}$ are all the minimally exchangeable sets and are all different from each other. What is still needed is that the intersection between any two minimally exchangeable sets is always empty. The third lemma states something stronger (that is indeed needed later in the proof of \ref{lem:good}), namely that the intersection between a minimally exchangeable set and any other exchangeable set is always empty, unless the latter contains the former.

\begin{lemma}\label{lem:inter}
Let $S \subseteq E_1$ be a minimally exchangeable set and $S' \subseteq E_1$ be an  exchangeable set such that $S' \cap S \neq \emptyset$. Then $ S \subseteq S'$.
\end{lemma}

If the intersection between any two minimally exchangeable sets was nonempty, then by Lemma \ref{lem:inter} one is contained in the other, which contradicts minimality. We can conclude that $D$ is a valid set of exchange deals.
\end{proof}
\medskip


\begin{proof}[\textbf{Proof of Lemma \ref{lem:fes}}]
Suppose that this does not hold. Without loss of generality, assume that  there is some $S_1\subseteq E_1$ and two profiles $\mathbf{v}^{\textsc{i}}=(v_1^{\textsc{i}}, v_2^{\textsc{i}})$ and $\mathbf{v}^{\textsc{ii}}=(v_1^{\textsc{ii}}, v_2^{\textsc{ii}})$, such that $X^{E_1}_1(\mathbf{v}^{\textsc{i}})=E_1 \mysetminus S_1=X^{E_1}_1(\mathbf{v}^{\textsc{ii}})$ and  $X^{E_2}_1(\mathbf{v}^{\textsc{i}})= S_2 \neq S'_2=X^{E_2}_1(\mathbf{v}^{\textsc{ii}})$. 

For the sake of readability, let $A=S_2 \mysetminus S'_2$, $B = S_2 \cap S'_2$, $C = S'_2 \mysetminus S_2$, and $D = M\mysetminus (S_2 \cup S'_2)$. Since $S_2 \neq S'_2$, either  $A\neq \emptyset$ or $C\neq \emptyset$. Without loss of generality, suppose that $A\neq \emptyset$. Using this notation, $X_1(\mathbf{v}^{\textsc{i}}) =(E_1 \mysetminus S_1) \cup A\cup B$ and $X_2(\mathbf{v}^{\textsc{i}}) =S_1 \cup C \cup D$, while $X_1(\mathbf{v}^{\textsc{ii}}) = (E_1\mysetminus S_1) \cup B \cup C$ and $X_2(\mathbf{v}^{\textsc{ii}}) =  S_1\cup A \cup D$.

We proceed to profile $\mathbf{v}^{\textsc{iii}}=(v_1^{\textsc{i}}, v_2^{\textsc{iii}})$ by changing the values of player 2: 
\begin{equation*}
\begin{IEEEeqnarraybox}[
\IEEEeqnarraystrutmode
\IEEEeqnarraystrutsizeadd{2pt}
{0.5pt}
][c]{c/v/c/v/c/V/c/v/c/v/c/v/c}
 && E_1\mysetminus S_1 && S_1 && A  && B  && C  && D \\\hline
 v_2^{\textsc{iii}} && \text{---}\ 1\ \text{---} &&  \text{---}\ m^2\ \text{---} && \text{---}\ 1\ \text{---}  && \text{---}\ 1\ \text{---}&&\text{---}\ m^3\ \text{---}&&\text{---}\ m^3\ \text{---}
\end{IEEEeqnarraybox}
\end{equation*}
Since the most valuable items of player 2 are those which he was allocated in profile $\mathbf{v}^{\textsc{i}}$, by truthfulness,  he should still get them, but he should not get any other item. Thus $\mathcal{X}_E(\mathbf{v}^{\textsc{iii}}) = \mathcal{X}_E(\mathbf{v}^{\textsc{i}})$.

We move to profile $\mathbf{v}^{\textsc{iv}}=(v_1^{\textsc{iii}}, v_2^{\textsc{iii}})$ by changing the values of player 1:
\begin{equation*}
\begin{IEEEeqnarraybox}[
\IEEEeqnarraystrutmode
\IEEEeqnarraystrutsizeadd{2pt}
{0.5pt}
][c]{c/v/c/v/c/V/c/v/c/v/c/v/c}
 && E_1\mysetminus S_1 && S_1 && A  && B  && C  && D \\\hline
 v_1^{\textsc{iii}} && \text{---}\ m^3\ \text{---} &&  \text{---}\ m\ \text{---} && \text{---}\ m^2\ \text{---}  && \text{---}\ 1\ \text{---}&&\text{---}\ 1\ \text{---}&&\text{---}\ 1\ \text{---}
\end{IEEEeqnarraybox}
\end{equation*}
By truthfulness we have that player 1 must get   $E_1\mysetminus S_1$ and $A$ (or else he could deviate to profile  $\mathbf{v}^{\textsc{iii}}$ and  strictly improve). Since he gets  $A$,  an exchange takes place. Due to the minimality of  $S_1$, we can derive that player 2 receives the whole $S_1$. In addition, player 2 continues to get $D$, since he strongly desires it and $D\subseteq E_2$. So we can conclude that $(E_1\mysetminus S_1) \cup A \subseteq X_1(\mathbf{v}^{\textsc{iv}})$ and $ S_1 \cup D \subseteq X_2(\mathbf{v}^{\textsc{iv}})$, while we do not care about the allocation of the remaining items.

Now let us return to profile $\mathbf{v}^{\textsc{ii}}=(v_1^{\textsc{ii}}, v_2^{\textsc{ii}})$. Starting from here, we change the values of player 2 and to get profile $\mathbf{v}^{\textsc{v}}=(v_1^{\textsc{ii}}, v_2^{\textsc{iv}})$.
\begin{equation*}
\begin{IEEEeqnarraybox}[
\IEEEeqnarraystrutmode
\IEEEeqnarraystrutsizeadd{2pt}
{0.5pt}
][c]{c/v/c/v/c/V/c/v/c/v/c/v/c}
 && E_1\mysetminus S_1 && S_1 && A  && B  && C  && D \\\hline
 v_2^{\textsc{iv}} && \text{---}\ 1\ \text{---} &&  \text{---}\ m\ \text{---} && \text{---}\ m^2\ \text{---}  && \text{---}\ 1\ \text{---}&&\text{---}\ 1\ \text{---}&&\text{---}\ m^2\ \text{---}
\end{IEEEeqnarraybox}
\end{equation*}
By truthfulness, like in profile $\mathbf{v}^{\textsc{iii}}$, we have $\mathcal{X}_E(\mathbf{v}^{\textsc{v}}) = \mathcal{X}_E(\mathbf{v}^{\textsc{ii}})$.


Next, we proceed to  profile $\mathbf{v}^{\textsc{vi}}=(v_1^{\textsc{iv}}, v_2^{\textsc{iv}})$, where
\begin{equation*}
\begin{IEEEeqnarraybox}[
\IEEEeqnarraystrutmode
\IEEEeqnarraystrutsizeadd{2pt}
{0.5pt}
][c]{c/v/c/v/c/V/c/v/c/v/c/v/c}
 && E_1\mysetminus S_1 && S_1 && A  && B  && C  && D \\\hline
 v_1^{\textsc{iv}} && \text{---}\ m^4\ \text{---} &&  \text{---}\ m\ \text{---} && \text{---}\ m^3\ \text{---}  && \text{---}\ m^2\ \text{---}&&\text{---}\ m^2\ \text{---}&&\text{---}\ 1\ \text{---}
\end{IEEEeqnarraybox}
\end{equation*}
Player 2 continues to get $A,D$ since he strongly desires them and $A,D \subseteq E_2$.
By the same argument, player 1 gets $E_1\mysetminus S_1$. Additionally, we know that an exchange  happens (otherwise player 1 would deviate to profile $\mathbf{v}^{\textsc{v}}$ in order to get the items of $B\cup C$), so player 2 gets set the whole $S_1$ due to its minimality. Thus we can conclude that $X_1(\mathbf{v}^{\textsc{vi}}) = (E_1\mysetminus S_1) \cup B \cup C$ and $X_2(\mathbf{v}^{\textsc{vi}}) =  S_1\cup A \cup D$.

Next, we move to profile $\mathbf{v}^{\textsc{vii}}=(v_1^{\textsc{iv}}, v_2^{\textsc{v}})$ by changing player 2 this time:
\begin{equation*}
\begin{IEEEeqnarraybox}[
\IEEEeqnarraystrutmode
\IEEEeqnarraystrutsizeadd{2pt}
{0.5pt}
][c]{c/v/c/v/c/V/c/v/c/v/c/v/c}
 && E_1\mysetminus S_1 && S_1 && A  && B  && C  && D \\\hline
 v_2^{\textsc{v}} && \text{---}\ 1\ \text{---} &&  \text{---}\ m^2\ \text{---} && \text{---}\ m\ \text{---}  && \text{---}\ 1\ \text{---}&&\text{---}\ 1\ \text{---}&&\text{---}\ m^3\ \text{---}
\end{IEEEeqnarraybox}
\end{equation*}
By truthfulness, the allocation does not change, i.e., $X_1(\mathbf{v}^{\textsc{vii}}) = (E_1\mysetminus S_1) \cup B \cup C$ and $X_2(\mathbf{v}^{\textsc{vii}}) =  S_1\cup A \cup D$.

Finally, we move to profile $\mathbf{v}^{\textsc{viii}}=(v_1^{\textsc{iii}}, v_2^{\textsc{v}})$ by changing the values of player 1 back to the values that he had in profile $\mathbf{v}^{\textsc{iv}}$.
Now recall that $X_2(\mathbf{v}^{\textsc{iv}})\supseteq S_1 \cup D$. Since in this profile $S_1 \cup D$ contains player 2's most valuable items,  he must continue to get them by  truthfulness.  This means that there is an exchange. Player 1 however must get some items from  $A$ in any exchange; if not he can declare that he strongly desires $E_1$ and strictly improve.  This, however, contradicts the truthfulness of the mechanism, since player 1 can deviate from  $\mathbf{v}^{\textsc{vii}}$ to $\mathbf{v}^{\textsc{viii}}$ and become strictly better.
\end{proof}
\medskip


\begin{proof}[\textbf{Proof of Lemma \ref{lem:min}}]
Suppose that this does not hold, i.e., there exists some minimally exchangeable $S_1\in E_1$, such that $(S_1, S_2)$ is the only feasible exchange involving $S_1$, but $S_2$ is not minimally exchangeable.  So there exists $S'_2 \subseteq S_2$ that is minimally exchangeable. So let $S'_1$ be such that $(S'_1,S'_2)$ is a feasible exchange (notice that $S_1 \neq S'_1$ by lemma \ref{lem:fes}). 

For the sake of readability, let $A=E_1\mysetminus (S_1 \cup S'_1)$, $B = S'_1 \cap S_1$, $C = S_1 \mysetminus S'_1$, $D = S_2 \cap S'_2$, $E=S'_2$, and $F=S_2 \mysetminus S'_2$.

So there is a profile $\mathbf{v}^{\textsc{i}}=(v_1^{\textsc{i}}, v_2^{\textsc{i}})$, where $X_1(\mathbf{v}^{\textsc{i}})=(E_1\mysetminus S_1)\cup S_2 = A \cup B \cup E \cup F$ and  $X_2(\mathbf{v}^{\textsc{i}})=(E_2\mysetminus S_2)\cup S_1 =C \cup D \cup (E_2\mysetminus S_2)$. Also there is another profile  $\mathbf{v}^{\textsc{ii}}=(v_1^{\textsc{ii}}, v_2^{\textsc{ii}})$ where $X_1(\mathbf{v}^{\textsc{ii}})=(E_1\mysetminus S'_1)\cup S'_2 =A \cup D \cup E$ and  $X_2(\mathbf{v}^{\textsc{ii}})=(E_2\mysetminus S'_2)\cup S'_1=B \cup C \cup F\cup (E_2\mysetminus S_2)$. 

We start from profile $\mathbf{v}^{\textsc{i}}=(v_1^{\textsc{i}}, v_2^{\textsc{i}})$ and we proceed to profile $\mathbf{v}^{\textsc{iii}}=(v_1^{\textsc{iii}}, v_2^{\textsc{i}})$ by changing the values of player 1: 
\begin{equation*}
\begin{IEEEeqnarraybox}[
\IEEEeqnarraystrutmode
\IEEEeqnarraystrutsizeadd{2pt}
{0.5pt}
][c]{c/v/c/v/c/v/c/v/c/V/c/v/c/v/c}
 &&  A && B && C  && D  && E  && F && E_2\mysetminus S_2\\\hline
 v_1^{\textsc{iii}} && \text{---}\ m^4\ \text{---} && \text{---}\ m^4\ \text{---} &&  \text{---}\ m\ \text{---} && \text{---}\ m\ \text{---}  && \text{---}\ m^3\ \text{---}&&\text{---}\ m^2\ \text{---}&&\text{---}\ 1\ \text{---}
\end{IEEEeqnarraybox}
\end{equation*}
Since player's 1 most valuable items are those  he was allocated in profile $\mathbf{v}^{\textsc{i}}$, due to the truthfulness of the mechanism,  he must continue to get them while not getting any other item. Thus the allocation does not change, i.e., $X_1(\mathbf{v}^{\textsc{iii}}) = A \cup B \cup E \cup F$ and $X_2(\mathbf{v}^{\textsc{iii}}) = C \cup D \cup (E_2\mysetminus S_2)$.

Next, move to profile $\mathbf{v}^{\textsc{iv}}=(v_1^{\textsc{iii}}, v_2^{\textsc{iii}})$ by changing the values of player 2:
\begin{equation*}
\begin{IEEEeqnarraybox}[
\IEEEeqnarraystrutmode
\IEEEeqnarraystrutsizeadd{2pt}
{0.5pt}
][c]{c/v/c/v/c/v/c/v/c/V/c/v/c/v/c}
 &&  A && B && C  && D  && E  && F && E_2\mysetminus S_2\\\hline
 v_2^{\textsc{iii}} && \text{---}\ m\ \text{---} && \text{---}\ m^4\ \text{---} &&  \text{---}\ m\ \text{---} && \text{---}\ m^3\ \text{---}  && \text{---}\ 1 \text{---}&&\text{---}\ m^2\ \text{---}&&\text{---}\ m^5\ \text{---}
\end{IEEEeqnarraybox}
\end{equation*}
Player 2 must get $E_2\mysetminus S_2$ since he strongly desires them  and $E_2\mysetminus S_2 \subseteq E_2$. Similarly, player 1 gets  $A \cup B$. Moreover, we know that an exchange should take place (otherwise player 2 would deviate to $\mathbf{v}^{\textsc{iii}}$ and become strictly better). What can be exchanged from $E_1$ is a subset of $C \cup D$, and since $C \cup D =S_1$  is  minimal, it is exchanged with  $S_2 =E \cup F$ (the only set that is exchangeable with $S_1$, by Lemma \ref{lem:fes}). Thus we conclude that the allocation here is $X_1(\mathbf{v}^{\textsc{iv}})=A \cup B \cup E \cup F$ and $X_2(\mathbf{v}^{\textsc{iv}})=C \cup D \cup (E_2\mysetminus S_2)$.

Finally we move to profile $\mathbf{v}^{\textsc{v}}=(v_1^{\textsc{iv}}, v_2^{\textsc{iii}})$, by changing the values of player 1:
\begin{equation*}
\begin{IEEEeqnarraybox}[
\IEEEeqnarraystrutmode
\IEEEeqnarraystrutsizeadd{2pt}
{0.5pt}
][c]{c/v/c/v/c/v/c/v/c/V/c/v/c/v/c}
 &&  A && B && C  && D  && E  && F && E_2\mysetminus S_2\\\hline
 v_1^{\textsc{iv}} && \text{---}\ m^4\ \text{---} && \text{---}\ m^2\ \text{---} &&  \text{---}\ m\ \text{---} && \text{---}\ m\ \text{---}  && \text{---}\ m^3\ \text{---}&&\text{---}\ m^2\ \text{---}&&\text{---}\ 1\ \text{---}
\end{IEEEeqnarraybox}
\end{equation*}
By truthfulness, like above,  the allocation does not change, i.e., $X_1(\mathbf{v}^{\textsc{v}}) = A \cup B \cup E \cup F$ and $X_2(\mathbf{v}^{\textsc{v}}) = C \cup D \cup (E_2\mysetminus S_2)$.

Now let us return to profile $\mathbf{v}^{\textsc{ii}}=(v_1^{\textsc{ii}}, v_2^{\textsc{ii}})$. Starting from this profile we change the values of player 2 to get profile $\mathbf{v}^{\textsc{vi}}=(v_1^{\textsc{ii}}, v_2^{\textsc{iv}})$.
\begin{equation*}
\begin{IEEEeqnarraybox}[
\IEEEeqnarraystrutmode
\IEEEeqnarraystrutsizeadd{2pt}
{0.5pt}
][c]{c/v/c/v/c/v/c/v/c/V/c/v/c/v/c}
 &&  A && B && C  && D  && E  && F && E_2\mysetminus S_2\\\hline
 v_2^{\textsc{iv}} && \text{---}\ 1\ \text{---} && \text{---}\ m^3\ \text{---} &&  \text{---}\ 1\ \text{---} && \text{---}\ 1\ \text{---}  && \text{---}\ m \text{---}&&\text{---}\ m^2\ \text{---}&&\text{---}\ m^4\ \text{---}
\end{IEEEeqnarraybox}
\end{equation*}
 Player 2 must get (at least) $B \cup F \cup (E_2\mysetminus S_2)$ since else he could deviate to profile $\mathbf{v}^{\textsc{ii}}$ and become strictly better. Now since player 1 loses  $B$ we know that an exchange  takes place with some of the available items in $E$. By the minimality of $E=S'_2$, player 1 gets the whole $E$ and he  loses $B \cup C= S'_1$. Thus we can conclude that the allocation here is $X_1(\mathbf{v}^{\textsc{vi}})=A \cup D \cup E$, $X_2(\mathbf{v}^{\textsc{vi}})=B \cup C \cup F\cup (E_2\mysetminus S_2)$.

In order to conclude, we move to profile $\mathbf{v}^{\textsc{vii}}=(v_1^{\textsc{iv}}, v_2^{\textsc{iv}})$ by changing the values of player 1 back to what he played in $\mathbf{v}^{\textsc{v}}$,
\begin{equation*}
\begin{IEEEeqnarraybox}[
\IEEEeqnarraystrutmode
\IEEEeqnarraystrutsizeadd{2pt}
{0.5pt}
][c]{c/v/c/v/c/v/c/v/c/V/c/v/c/v/c}
 &&  A && B && C  && D  && E  && F && E_2\mysetminus S_2\\\hline
 v_1^{\textsc{iv}} && \text{---}\ m^4\ \text{---} && \text{---}\ m^2\ \text{---} &&  \text{---}\ m\ \text{---} && \text{---}\ m\ \text{---}  && \text{---}\ m^3\ \text{---}&&\text{---}\ m^2\ \text{---}&&\text{---}\ 1\ \text{---}\\\hline
 v_2^{\textsc{iv}} && \text{---}\ 1\ \text{---} && \text{---}\ m^3\ \text{---} &&  \text{---}\ 1\ \text{---} && \text{---}\ 1\ \text{---}  && \text{---}\ m \text{---}&&\text{---}\ m^2\ \text{---}&&\text{---}\ m^4\ \text{---}
\end{IEEEeqnarraybox}
\end{equation*}
Player 2 gets $E_2 \mysetminus S_2$ because he strongly desires it. We also know that an exchange should take place, otherwise player 1 would deviate to $\mathbf{v}^{\textsc{vi}}$ and  strictly improve his total value. As a result, player 2 gets at least one item from set $B$, or he could increase to $m^4$ his value for any item in $E_2$ and improve by getting $E_2$. However, now player 2 can deviate from profile  $\mathbf{v}^{\textsc{v}}$ to $\mathbf{v}^{\textsc{vii}}$ and become strictly better, something that contradicts the truthfulness of the mechanism.
\end{proof}
\medskip


\begin{proof}[\textbf{Proof of Lemma \ref{lem:inter}}]
Suppose that this does not hold, i.e., there exists a minimally exchangeable set $S_1\in E_1$ and an exchangeable set $S'_1\in E_1$, such that $S_1\cap S'_1\neq \emptyset$ and $S_1\nsubseteq S'_1$. 
Choose $S'_1$ to be minimal, i.e., if $S''_1 \subsetneq S'_1$ then either $S_1 \cap S''_1= \emptyset$ or $S''_1$ is not exchangeable. Let $S_2, S'_2$ be such that $(S_1,S_2)$, $(S'_1,S'_2)$ are feasible exchanges and $S'_2$ is minimal in the sense that there is no $S''_2 \subsetneq S'_2$ where $(S'_1, S''_2)$ being a feasible exchange. From Lemmata \ref{lem:fes} and \ref{lem:min} we have that $S'_2 \mysetminus S_2 \neq \emptyset$.

For the sake of readability, let $A=E_1\mysetminus (S_1 \cup S'_1)$, $B = S'_1 \mysetminus S_1$, $C = S'_1 \cap S_1$, $D = S_1 \mysetminus S'_1$, $E = S_2 \cap S'_2$, $F=S_2 \mysetminus S'_2$, $G=E_2\mysetminus (S_2 \cup S'_2)$, and $H=S'_2 \mysetminus S_2$. 

So there is a profile $\mathbf{v}^{\textsc{i}}=(E_1\mysetminus S_1)\cup S_2=(v_1^{\textsc{i}}, v_2^{\textsc{i}})$, where $X_1(\mathbf{v}^{\textsc{i}})=A \cup B \cup E \cup F$, $X_2(\mathbf{v}^{\textsc{i}})=(E_2\mysetminus S_2)\cup S_1=C \cup D \cup G \cup H$. There is also a profile $\mathbf{v}^{\textsc{ii}}=(v_1^{\textsc{ii}}, v_2^{\textsc{ii}})$, where $X_1(\mathbf{v}^{\textsc{ii}})=(E_1\mysetminus S'_1)\cup S'_2=A \cup D \cup E \cup H$, $X_2(\mathbf{v}^{\textsc{ii}})=(E_2\mysetminus S'_2)\cup S'_1=B \cup C \cup F\cup G$.

We start from profile $\mathbf{v}^{\textsc{i}}=(v_1^{\textsc{i}}, v_2^{\textsc{i}})$ and we proceed to profile $\mathbf{v}^{\textsc{iii}}=(v_1^{\textsc{i}}, v_2^{\textsc{iii}})$ by changing the values of player 2: 
\begin{equation*}
\begin{IEEEeqnarraybox}[
\IEEEeqnarraystrutmode
\IEEEeqnarraystrutsizeadd{2pt}
{0.5pt}
][c]{c/v/c/v/c/v/c/v/c/V/c/v/c/v/c/v/c}
 &&  A && B && C  && D  && E  && F && G && H\\\hline
 v_2^{\textsc{iii}} && \text{---}\ 1\ \text{---} && \text{---}\ 1\ \text{---} &&  \text{---}\ m\ \text{---} && \text{---}\ m\ \text{---}  && \text{---}\ 1\ \text{---}&&\text{---}\ 1\ \text{---}&&\text{---}\ m^2\ \text{---} && \text{---}\ m^2\ \text{---}
\end{IEEEeqnarraybox}
\end{equation*}
By truthfulness, we can conclude that the allocation remains the same, i.e.,   player 1 gets $A \cup B \cup E \cup F$,  while player 2 gets $C \cup D \cup G \cup H$ . 

Next, we move to profile $\mathbf{v}^{\textsc{iv}}=(v_1^{\textsc{iii}}, v_2^{\textsc{iii}})$ by changing the values of player 1:
\begin{equation*}
\begin{IEEEeqnarraybox}[
\IEEEeqnarraystrutmode
\IEEEeqnarraystrutsizeadd{2pt}
{0.5pt}
][c]{c/v/c/v/c/v/c/v/c/V/c/v/c/v/c/v/c}
 &&  A && B && C  && D  && E  && F && G && H\\\hline
 v_1^{\textsc{iii}} && \text{---}\ m^3\ \text{---} && \text{---}\ m^3\ \text{---} &&  \text{---}\ 1\ \text{---} && \text{---}\ 1\ \text{---}  && \text{---}\ m^2\ \text{---}&&\text{---}\ m^2\ \text{---}&&\text{---}\ 1\ \text{---} && \text{---}\ m\ \text{---}
\end{IEEEeqnarraybox}
\end{equation*}
Again, by truthfulness player 1 gets $A \cup B \cup E \cup F$,  and player 2 gets $C \cup D \cup G \cup H$ .

We continue by moving to profile $\mathbf{v}^{\textsc{v}}=(v_1^{\textsc{iii}}, v_2^{\textsc{iv}})$ by changing the values of player 2:
\begin{equation*}
\begin{IEEEeqnarraybox}[
\IEEEeqnarraystrutmode
\IEEEeqnarraystrutsizeadd{2pt}
{0.5pt}
][c]{c/v/c/v/c/v/c/v/c/V/c/v/c/v/c/v/c}
 &&  A && B && C  && D  && E  && F && G && H\\\hline
 v_2^{\textsc{iv}} && \text{---}\ 1\ \text{---} && \text{---}\ m\ \text{---} &&  \text{---}\ m^3\ \text{---} && \text{---}\ 1\ \text{---}  && \text{---}\ 1\ \text{---}&&\text{---}\ m^2\ \text{---}&&\text{---}\ m^4\ \text{---} && \text{---}\ 1\ \text{---}
\end{IEEEeqnarraybox}
\end{equation*}
Player 2 must get $G$ since he strongly desires it and $H \subseteq E_2$. The same goes for player 1 and  $A \cup B$. Now we know that an exchange should take place, otherwise player 2 would deviate to $\mathbf{v}^{\textsc{iii}}$ and become strictly better. Since the only available exchangeable set here is $C \cup D =S_1$ (because it is  minimal), it is exchanged with set $S_2 =E \cup F$ (the only set  exchangeable with $S_1$ by lemma \ref{lem:fes}). Thus we  conclude that the allocation remains the same,  player 1 gets $A \cup B \cup E \cup F$,  while player 2 gets $C \cup D \cup G \cup H$ . 

Next proceed to profile $\mathbf{v}^{\textsc{vi}}=(v_1^{\textsc{iv}}, v_2^{\textsc{iv}})$ by changing the values of player 1:
\begin{equation*}
\begin{IEEEeqnarraybox}[
\IEEEeqnarraystrutmode
\IEEEeqnarraystrutsizeadd{2pt}
{0.5pt}
][c]{c/v/c/v/c/v/c/v/c/V/c/v/c/v/c/v/c}
 &&  A && B && C  && D  && E  && F && G && H\\\hline
 v_1^{\textsc{iv}} && \text{---}\ m^4\ \text{---} && \text{---}\ m\ \text{---} &&  \text{---}\ 1\ \text{---} && \text{---}\ 1\ \text{---}  && \text{---}\ m^2\ \text{---}&&\text{---}\ m^3\ \text{---}&&\text{---}\ 1\ \text{---} && \text{---}\ m^2\ \text{---}
\end{IEEEeqnarraybox}
\end{equation*}
we can derive by truthfulness that player 1 must get (at least) $A \cup F$, or else he would deviate to profile $\mathbf{v}^{\textsc{v}}$ and improve. Currently, this is all what we need to know for $\mathbf{v}^{\textsc{vi}}$.

Now let us return to profile $\mathbf{v}^{\textsc{ii}}=(v_1^{\textsc{ii}}, v_2^{\textsc{ii}})$. Starting from here we change the values of player 1 to  get profile $\mathbf{v}^{\textsc{vii}}=(v_1^{\textsc{v}}, v_2^{\textsc{ii}})$.
\begin{equation*}
\begin{IEEEeqnarraybox}[
\IEEEeqnarraystrutmode
\IEEEeqnarraystrutsizeadd{2pt}
{0.5pt}
][c]{c/v/c/v/c/v/c/v/c/V/c/v/c/v/c/v/c}
 &&  A && B && C  && D  && E  && F && G && H\\\hline
 v_1^{\textsc{v}} && \text{---}\ m^2\ \text{---} && \text{---}\ 1\ \text{---} &&  \text{---}\ 1\ \text{---} && \text{---}\ m^2\ \text{---}  && \text{---}\ m\ \text{---}&&\text{---}\ 1\ \text{---}&&\text{---}\ 1\ \text{---} && \text{---}\ m\ \text{---}
\end{IEEEeqnarraybox}
\end{equation*}
By truthfulness,  the allocation remains the same, i.e.,  player 1 gets $A \cup D \cup E \cup H$,  while player 2 gets $B \cup C \cup F \cup G$ .

We now move to profile $\mathbf{v}^{\textsc{viii}}=(v_1^{\textsc{v}}, v_2^{\textsc{v}})$ and change the values of player 2.
\begin{equation*}
\begin{IEEEeqnarraybox}[
\IEEEeqnarraystrutmode
\IEEEeqnarraystrutsizeadd{2pt}
{0.5pt}
][c]{c/v/c/v/c/v/c/v/c/V/c/v/c/v/c/v/c}
 &&  A && B && C  && D  && E  && F && G && H\\\hline
 v_2^{\textsc{v}} && \text{---}\ 1\ \text{---} && \text{---}\ \alpha m^3\ \text{---} &&  \text{---}\ \alpha m^4\ \text{---} && \text{---}\ 1\ \text{---}  && \text{---}\ m^4\ \text{---}&&\text{---}\ m^5\ \text{---}&&\text{---}\ m^5\ \text{---} && \text{---}\ m^4\ \text{---}
\end{IEEEeqnarraybox}
\end{equation*}
The values in  $B \cup C$ are set in such a way so that $v_2^{\textsc{v}}(B \cup C) > v_2^{\textsc{v}}(E \cup H)$, but $v_2^{\textsc{v}}(T) < v_2^{\textsc{v}}(E \cup H)$ for any $T\subsetneq B \cup C$.\footnote{This is always possible. In particular, if $|B|>0$ then $\alpha = \frac{|E \cup H|m^4-m}{(|B|-1)m^3+|C|m^4}$ works. If $|B|=0$, then $\alpha = \frac{|E \cup H|m^4-m}{(|C|-1)m^4}$. In order to apply the idea mentioned in Remark \ref{rem:notequal}, one can multiply the whole profile with the denominator of $\alpha$.} 

Notice that  player 2 must get $G \cup F$ since he strongly desires it. The same goes for player 1 and  $A \cup D$. 
We know that an exchange should take place, otherwise player 2 would deviate to $\mathbf{v}^{\textsc{vii}}$ and improve. In this exchange, values are such that player 2 should get the whole $S'_1$. Thus we  conclude that the allocation remains the same, i.e.,  player 1 gets $A \cup D \cup E \cup H$,  while player 2 gets $B \cup C \cup F \cup G$ .

We now move to profile $\mathbf{v}^{\textsc{ix}}=(v_1^{\textsc{vi}}, v_2^{\textsc{v}})$ and change the values of player 1.
\begin{equation*}
\begin{IEEEeqnarraybox}[
\IEEEeqnarraystrutmode
\IEEEeqnarraystrutsizeadd{2pt}
{0.5pt}
][c]{c/v/c/v/c/v/c/v/c/V/c/v/c/v/c/v/c}
 &&  A && B && C  && D  && E  && F && G && H\\\hline
 v_1^{\textsc{vi}} && \text{---}\ m^4\ \text{---} && \text{---}\ m\ \text{---} &&  \text{---}\ 1\ \text{---} && \text{---}\ m^4\ \text{---}  && \text{---}\ m^2\ \text{---}&&\text{---}\ m^3\ \text{---}&&\text{---}\ 1\ \text{---} && \text{---}\ m^2\ \text{---}
\end{IEEEeqnarraybox}
\end{equation*}
Again player 2 must get $G \cup F$. Given that, player 1 gets at least $A \cup D \cup E \cup H$, and by truthfulness he cannot receive strictly more items. 
Therefore, the allocation remains the same, i.e.,  player 1 gets $A \cup D \cup E \cup H$,  while player 2 gets $B \cup C \cup F \cup G$.

We now move to profile $\mathbf{v}^{\textsc{x}}=(v_1^{\textsc{iv}}, v_2^{\textsc{v}})$ by changing  the values of player 1 back to what he had at profile $\mathbf{v}^{\textsc{vi}}$. Recall:
\begin{equation*}
\begin{IEEEeqnarraybox}[
\IEEEeqnarraystrutmode
\IEEEeqnarraystrutsizeadd{2pt}
{0.5pt}
][c]{c/v/c/v/c/v/c/v/c/V/c/v/c/v/c/v/c}
 &&  A && B && C  && D  && E  && F && G && H\\\hline
 v_1^{\textsc{iv}} && \text{---}\ m^4\ \text{---} && \text{---}\ m\ \text{---} &&  \text{---}\ 1\ \text{---} && \text{---}\ 1\ \text{---}  && \text{---}\ m^2\ \text{---}&&\text{---}\ m^3\ \text{---}&&\text{---}\ 1\ \text{---} && \text{---}\ m^2\ \text{---}
\end{IEEEeqnarraybox}
\end{equation*}
Like above Player 2 gets $F \cup G$. The same goes for player 1 $A$. By truthfulness, an exchange must happen and player 1 gets at least the set $E \cup H$  (else he would deviate to  $\mathbf{v}^{\textsc{ix}}$ and improve). Moreover, since player 2 loses $E \cup H$ he must  at least get the set $B \cup C$. We conclude that   player 1 gets $A \cup E \cup H$,  player 2 gets $B \cup C \cup F \cup G$, while we do not care what happens for items in $D$.

Now notice that player 2 can deviate from profile $\mathbf{v}^{\textsc{vi}}$ to profile $\mathbf{v}^{\textsc{x}}$ and become strictly better (recall that at profile $\mathbf{v}^{\textsc{vi}}$ player 2 loses $G$, while $A, D, E, H$ all have very small value) and this  contradicts  truthfulness.
\end{proof}
\medskip


\begin{proof}[\textbf{Proof of Lemma \ref{lem:umes}}]
We begin with a direct implication of the  Lemmata \ref{lem:fes}--\ref{lem:inter}. Although we are not guaranteed yet that any feasible exchange can be expressed as a union of exchange deals from $D$ as it should, the following corollary is a step towards this direction. Recall that $S_1, \ldots , S_{k}$ and $T_1, \ldots , T_{k}$ are all the minimally exchangeable subsets of $E_1$ and $E_2$ respectively, and that $(S_i, T_i)$ is the only feasible exchange involving either one of $S_i$ and $T_i$, for every $i\in[k]$.

\begin{corollary}\label{col:combo}
For every  exchangeable set  $S\subseteq E_1$, we have that $S=W\cup \bigcup_{i \in I} S_i$, where $I\subseteq [k]$ with $|I|\ge 1$, while $W=S \mysetminus \bigcup_{i \in I} S_i$ does not contain any minimally exchangeable sets. Furthermore, this decomposition is unique.
\end{corollary}

Ideally, we would like two things. First, the $W$ part in the above decomposition to always be empty, i.e., we want every  exchangeable set to be a \emph{u}nion of \emph{m}inimally \emph{e}xchangeable \emph{s}ets (\emph{umes} for short).
Second, we want every umes of $E_1$ to be exchangeable only with the corresponding umes of $E_2$, and vice versa. To be more precise, we say that an umes $S=\bigcup_{i \in I}S_i$ is \emph{nice} if it is exchangeable with $T=\bigcup_{i \in I}T_i$ and only with $T$. The definition of a nice umes of $E_2$ is symmetric.
As it turns out, every umes is nice, but it takes a rather involved induction to prove it. Especially the fact that $\big(\bigcup_{i \in I}S_i, \bigcup_{i \in I}T_i \big)$ is exchangeable needs a carefully constructed argument about the value that each player must  gain from any exchange (see also Lemma \ref{lem:values}). 

\begin{lemma}\label{lem:good}
Every umes is nice.
\end{lemma}

Given the above lemma, we can now show that the set $W$ in the  decomposition of Corollary \ref{col:combo} is always empty. In fact the proof idea is the same as the one for Lemmata \ref{lem:fes}--\ref{lem:inter}.

\begin{lemma}\label{lem:change}
Every exchangeable set is an umes.
\end{lemma}
The above two lemmata complete the proof. They are proved below, right after Lemma \ref{lem:values}.
\end{proof}
\bigskip


For the following lemmas, recall that \emph{umes} is short for \emph{u}nion of \emph{m}inimally \emph{e}xchangeable \emph{s}ets!
\bigskip
%

\begin{lemma}\label{lem:values}
Let $(S, T)$ be a feasible exchange such that $S$ is a nice umes with the property that if $S' \subseteq S$ is exchangeable, then $S'$ is a nice umes. In particular, let $S=\bigcup_{i \in [r]}S_i$, where $S_i$ is minimally exchangeable for all $i\in[r]$. If $\mathbf{v}$ is a profile where $(S_i,T_i)$ is favorable for all  $i \in [r]$ then $(S,T)$ gives a lower bound on the value gained from exchanges in profile $\mathbf{v}$ for each player.
\end{lemma}

\begin{proof}[Proof of Lemma \ref{lem:values}]
Due to symmetry, it suffices to prove the lower bound for player 1.
Let $\mathbf{v}=(v_1,v_2)$ be  a profile like in the statement, where the values are $v_{i1},v_{i2}, ..., v_{im}$ for $i=1,2$.. Since $(S,T)$ is a feasible exchange, there exists a profile $\mathbf{v}^{\textsc{i}}=(v_1^{\textsc{i}}, v_2^{\textsc{i}})\in \mathcal{V}^{\neq}_{\! m}$ such that the exchange $(S,T)$ takes place, i.e.,$X_1(\mathbf{v}^{\textsc{i}})=(E_1 \mysetminus S) \cup T $ and  $X_2(\mathbf{v}^{\textsc{i}})=S\cup (E_2\mysetminus T)$. Starting from this profile we will use a series of intermediate profiles in order to reach  $\mathbf{v}=(v_1,v_2)$. Initially consider profile $\mathbf{v}^{\textsc{ii}}=(v_1^{\textsc{ii}}, v_2^{\textsc{i}})$ where we change the values of player 1.
\begin{equation*}
v_{1j}^{\textsc{ii}} =
\begin{cases}
\frac{m\cdot\max_i v_{1i}}{\min_i v_{1i}}\cdot v_{1j} & \text{if } j\in E_1\mysetminus S\\
v_{1j}  & \text{if } j\in S\cup T\\
\frac{\min_i v_{1i}}{m\cdot\max_i v_{1i}}\cdot v_{1j} & \text{otherwise } 
\end{cases}
\end{equation*}
In this profile each item in  $E_1\mysetminus S$ has a value which is higher from the sum of the  values in all the other sets. On the other hand,  items in  $E_2\mysetminus T$ have total value  less than the value of a single item in the other sets.\footnote{Notice that the values are chosen in a way such that if $\mathbf{v}\in \mathcal{V}^{\neq}_{\! m}$, then  $\mathbf{v}^{\textsc{i}}\in \mathcal{V}^{\neq}_{\! m}$ as well.} Since this is the case, player 1 must get $E_1\mysetminus S$ since he strongly desires it. In addition, an exchange must take place, or player 1 could deviate to profile $\mathbf{v}^{\textsc{i}}$ and become strictly better. Thus  an exchange takes place and must involve a subset $S'$ of  $S$. Now if $S'$ was a proper subset of $S$, then it would be a nice umes, i.e., $S'=\bigcup_{j \in I}S_j, I \subsetneq [r]$, and it is exchanged only with $T'=\bigcup_{j \in I}T_j$. 
However, since exchanges $S_j, T_j, j \in [r] \mysetminus I$ are also favorable, player 1 would deviate to profile $\mathbf{v}^{\textsc{i}}$ and become strictly better. Therefore, the exchange involves the whole  $S$, and since $S$ is a nice umes  it should be exchanged with $T$. So the allocation here is $X_1(\mathbf{v}^{\textsc{ii}})=(E_1 \mysetminus S) \cup T $, $X_2(\mathbf{v}^{\textsc{ii}})=S\cup (E_2\mysetminus T)$.

By moving to profile $\mathbf{v}^{\textsc{iii}}=(v_1^{\textsc{ii}}, v_2)$ where we change the values of player 2,
we have that, once again, player 1 must get the items in $E_1\mysetminus S$. Moreover, an exchange must take place, or player 2 could deviate to profile $\mathbf{v}^{\textsc{ii}}$ and become strictly better (recall that he prefers  $S$ from $T$). By following the same arguments as in the previous case, if the exchange involves  a  proper subset of $S$, player 2 would deviate to profile $\mathbf{v}^{\textsc{ii}}$ and become strictly better. Hence player 2 gets the whole $S$, i.e., the allocation here is again $X_1(\mathbf{v}^{\textsc{iii}})=(E_1 \mysetminus S) \cup T $ and $X_2(\mathbf{v}^{\textsc{iii}})=S\cup (E_2\mysetminus T)$.

Finally we move to profile $\mathbf{v}=(v_1,v_2)$ by changing the values of player 1.
It is easy to see that if there is no exchange that improves player 1 by at least $v_1(T)-v_1(S)$, then he could deviate to profile $\mathbf{v}^{\textsc{iii}}=(v_1^{\textsc{ii}}, v_2)$ and gain exactly that.
\end{proof}
\medskip


\begin{proof}[\textbf{Proof of Lemma \ref{lem:good}}]
We will use induction in the number of minimally exchangeable sets contained in an umes; let us call this number \emph{index} of the umes. Lemmata \ref{lem:fes} and \ref{lem:min} imply that every umes of index 1 is nice. That is the basis of our induction.

Assume that every umes of index lower or equal to $k$ is nice and notice that Lemma \ref{lem:inter} implies that every exchangeable subset of an umes is also an umes.

Let  $S$ be an umes of index $k+1$. In particular, let $S=\bigcup_{i \in [k+1]}S_i$, where for any $i \in [k+1]$ we have that $S_i$ is minimally exchangeable and $(S_i,T_i)$ is a feasible exchange. By the inductive hypothesis we have that both $S_1$ and $S'=\bigcup_{i=2}^{k+1}S_i$ are nice umes and uniquely exchangeable with $S_1$ and $T'=\bigcup_{i=2}^{k+1}T_i$ respectively.

We first prove that $(S,T)$ is a feasible exchange. Consider the following profile $\mathbf{v}=(v_1,v_2)$, 
\begin{equation*}
\begin{IEEEeqnarraybox}[
\IEEEeqnarraystrutmode
\IEEEeqnarraystrutsizeadd{2pt}
{0.5pt}
][c]{c/v/c/v/c/v/c/V/c/v/c/v/c}
 &&  E_1 \mysetminus (S' \cup S_1) && S' && S_1  && T'  && T_1  &&E_2\mysetminus (T' \cup T_1)\\\hline
 v_1 && \text{---}\ \Delta\ \text{---} && \text{---}\ \delta\ \text{---} &&  \text{---}\ \epsilon\ \text{---} && \text{---}\ 1\ \text{---}&&\text{---}\ \zeta \ \text{---}&&\text{---}\ \delta\ \text{---}\\\hline
 v_2 && \text{---}\ \delta\ \text{---} && \text{---}\ n_j\ \text{---} &&  \text{---}\ 1\ \text{---} && \text{---}\ \theta_j \ \text{---}  && \text{---}\ \delta \text{---}&&\text{---}\ \Delta\ \text{---}
\end{IEEEeqnarraybox}
\end{equation*}
where $\Delta >> 1 >> \zeta, n_j, \theta_j, \epsilon >> \delta >> \lambda_i$.\footnote{In order to be able to apply the idea mentioned in Remark \ref{rem:notequal}, one can use $m^7$ instead of 1, and $\Delta=m^8$, $\delta=m^3$, $\lambda_i=|T_i|\cdot|S_i|$, $n_i=|T_j|\cdot m^{4}$, $\theta_i=|S_j|\cdot (m^{4}-1)$, $\zeta = |S_1|\cdot m^{4}$, and $\epsilon = |T_1|\cdot (m^{4}-1)$.} Regarding the rest values, $|T_1|\cdot \zeta = |S_1|\cdot \epsilon + \lambda_1$ and for all $j\in [k+1]\mysetminus\{1\}$ we have that $|S_j|\cdot n_j = |T_j|\cdot \theta_j + \lambda_j$. Now notice that $S'$ is a nice umes such that every exchangeable $S'' \subseteq S'$ is a nice umes and for all $j$, $(S_j, T_j)$ is a favorable exchange with respect to $\mathbf{v}$. Lemma \ref{lem:values}  guarantees that in $\mathbf{v}$, player 1 gains at least $v_1(T')-v_1(S')=|T'|-\delta|S'|$ from the exchanges. So player 1 gets a superset of $T'$, i.e., $T'\subseteq X_1^{E_2}(\mathbf{v})$. By lemma \ref{lem:inter}, this means that $ X_1^{E_2}(\mathbf{v})$ is either $T'$ or $T$.

On the other hand, if we apply lemma \ref{lem:values} for $(S_1, T_1)$ we have that in profile $\mathbf{v}$, player 2 should gain at least $v_2(S_1)-v_2(T_1)=|S_1|-\delta|T_1|$ from the exchanges. So, $X_2^{E_1}(\mathbf{v})  \supseteq S_1$. Since $T'$ is nice, however, we have that $X_1^{E_2}(\mathbf{v})=T'$ implies $X_2^{{E_1}}(\mathbf{v}) =S'\nsupseteq S_1$.  Therefore, it must be the case where $ X_1^{E_2}(\mathbf{v}) \supsetneq T'$ or else player 2 does not get enough value.

We conclude that $ X_2^{E_1}(\mathbf{v})= T$. Now we claim that $X_2^{E_1}(\mathbf{v}) =S$ and therefore $(S,T)$ is a feasible exchange. Indeed, every $S'' \subsetneq S$ that is exchangeable is an umes of index lower or equal to $k$ and therefore is nice. So $S'', T$ cannot be a feasible exchange, due to the fact that $S''$ has a unique pair $T'' \subsetneq T$. 
\smallskip

Next we show that there is no $\hat{T} \neq T$ such that $(S, \hat{T})$ is a feasible exchange. By the proof so far we have that if such a $\hat{T}$ existed, then it is not a subset of $T$. So suppose that there is a $\hat{T} \neq T$ such that $(S, \hat{T})$is a feasible exchange and let $T^*$ be a minimal such set (that is, if $R \subsetneq T^*$ then $(S,R)$ is not a feasible exchange or $R \subseteq T$).

Thus there are two profiles $\mathbf{v}^{\textsc{i}}=(v_1^{\textsc{i}}, v_2^{\textsc{i}})$ and $\mathbf{v}^{\textsc{ii}}=(v_1^{\textsc{ii}}, v_2^{\textsc{ii}})$ where we have that $X^{E_1}(\mathbf{v}^{\textsc{i}})=(E_1 \mysetminus S)=X^{E_1}_1(\mathbf{v}^{\textsc{ii}})$ and  $X^{E_2}_1(\mathbf{v}^{\textsc{i}})= T^* \neq T=X^{E_2}_1(\mathbf{v}^{\textsc{ii}})$. 

For the sake of readability, let $A=T^*\mysetminus T$, $B = T^* \cap T$, $C = T \mysetminus T^*$, $D =  E_2\mysetminus (T^* \cup T)$.

We start from profile $\mathbf{v}^{\textsc{i}}$ where  the allocation is $X_1(\mathbf{v}^{\textsc{i}})= (E_1 \mysetminus S) \cup A \cup B$, $X_2(\mathbf{v}^{\textsc{i}})=  S \cup C \cup D$  and we proceed to profile $\mathbf{v}^{\textsc{iii}}=(v_1^{\textsc{i}}, v_2^{\textsc{iii}})$ by changing the values of player 2: 
\begin{equation*}
\begin{IEEEeqnarraybox}[
\IEEEeqnarraystrutmode
\IEEEeqnarraystrutsizeadd{2pt}
{0.5pt}
][c]{c/v/c/v/c/V/c/v/c/v/c/v/c}
 && E_1\mysetminus S && S && A  && B  && C  && D \\\hline
 v_2^{\textsc{iii}} && \text{---}\ 1\ \text{---} &&  \text{---}\ m\ \text{---} && \text{---}\ 1\ \text{---}  && \text{---}\ 1\ \text{---}&&\text{---}\ m^2\ \text{---}&&\text{---}\ m^2\ \text{---}
\end{IEEEeqnarraybox}
\end{equation*}
By truthfulness, the allocation remains the same, i.e.,
$X_1(\mathbf{v}^{\textsc{iii}})= (E_1 \mysetminus S) \cup A \cup B$, $X_2(\mathbf{v}^{\textsc{iii}})=  S \cup C \cup D$. 

Next we move to profile  $\mathbf{v}^{\textsc{iv}}=(v_1^{\textsc{iii}}, v_2^{\textsc{iii}})$ by changing the values of player 1:
\begin{equation*}
\begin{IEEEeqnarraybox}[
\IEEEeqnarraystrutmode
\IEEEeqnarraystrutsizeadd{2pt}
{0.5pt}
][c]{c/v/c/v/c/V/c/v/c/v/c/v/c}
 && E_1\mysetminus S && S && A  && B  && C  && D \\\hline
 v_1^{\textsc{iii}} && \text{---}\ m^3\ \text{---} &&  \text{---}\ m\ \text{---} && \text{---}\ m^2\ \text{---}  && \text{---}\ 1\ \text{---}&&\text{---}\ 1\ \text{---}&&\text{---}\ 1\ \text{---}
\end{IEEEeqnarraybox}
\end{equation*}
Notice that player 1 must receive $E_1 \mysetminus S$ since he strongly desires it. The same goes for player 2 and  $C \cup D $. Now we know that an exchange should take place and that in this exchange player 1 must get at least set $A =T^*\mysetminus T$ (otherwise he would deviate to $\mathbf{v}^{\textsc{iii}}$ and become strictly better). 

We claim that player 1 gets the whole $T^*$. If this was not the case then he would get some set $R \supseteq A \neq \emptyset$. Since $R \subsetneq T^*$ and $R \nsubseteq T$ we have that the exchange $(S,R)$ is not feasible due to the minimality of $T^*$. Thus $R$ is exchanged with some $\hat{S} \subsetneq S$. However $\hat{S}$ is an umes (by Lemma \ref{lem:fes}) and by inductive hypothesis it is exchangeable only with strict subsets of $T$ which is a contradiction. 
Similarly, player 2 must get set the whole $S$, or otherwise he would get some $\hat{S} \subsetneq S$  which is  exchangeable only with strict subsets of $T$, something that can not happen. Thus the allocation here is $X_1(\mathbf{v}^{\textsc{iv}})= (E_1 \mysetminus S) \cup A \cup B$, $X_2(\mathbf{v}^{\textsc{iv}})=  S \cup C \cup D$.

Next we move to profile  $\mathbf{v}^{\textsc{v}}=(v_1^{\textsc{iii}}, v_2^{\textsc{iv}})$ by changing the values of player 2.
\begin{equation*}
\begin{IEEEeqnarraybox}[
\IEEEeqnarraystrutmode
\IEEEeqnarraystrutsizeadd{2pt}
{0.5pt}
][c]{c/v/c/v/c/V/c/v/c/v/c/v/c}
 && E_1\mysetminus S && S && A  && B  && C  && D \\\hline
 v_2^{\textsc{iv}} && \text{---}\ 1\ \text{---} &&  \text{---}\ m^2\ \text{---} && \text{---}\ 1\ \text{---}  && \text{---}\ 1\ \text{---}&&\text{---}\ m\ \text{---}&&\text{---}\ m^3\ \text{---}
\end{IEEEeqnarraybox}
\end{equation*}
By truthfulness, the allocation remains the same, i.e.,
$X_1(\mathbf{v}^{\textsc{v}})= (E_1 \mysetminus S) \cup A \cup B$, $X_2(\mathbf{v}^{\textsc{v}})=  S \cup C \cup D$.

Now let us return to profile $\mathbf{v}^{\textsc{ii}}=(\mathbf{v}_1^{\textsc{ii}}, \mathbf{v}_2^{\textsc{ii}})$. Starting from this profile we change the values of player 2 and get profile $\mathbf{v}^{\textsc{vi}}=(v_1^{\textsc{ii}}, v_2^{\textsc{v}})$.
\begin{equation*}
\begin{IEEEeqnarraybox}[
\IEEEeqnarraystrutmode
\IEEEeqnarraystrutsizeadd{2pt}
{0.5pt}
][c]{c/v/c/v/c/V/c/v/c/v/c/v/c}
 && E_1\mysetminus S && S && A  && B  && C  && D \\\hline
 v_2^{\textsc{v}} && \text{---}\ 1\ \text{---} &&  \text{---}\ m\ \text{---} && \text{---}\ m^2\ \text{---}  && \text{---}\ 1\ \text{---}&&\text{---}\ 1\ \text{---}&&\text{---}\ m^2\ \text{---}
\end{IEEEeqnarraybox}
\end{equation*}
Since player's 2 most valuable items are those which he was allocated in profile $\mathbf{v}^{\textsc{ii}}$, by truthfulness, the allocation remains the same, i.e.,
$X_1(\mathbf{v}^{\textsc{vi}})= (E_1 \mysetminus S) \cup B \cup C$, $X_2(\mathbf{v}^{\textsc{vi}})=  S \cup A \cup D$. 

Next we move to profile  $\mathbf{v}^{\textsc{vii}}=(v_1^{\textsc{iv}}, v_2^{\textsc{v}})$ by changing the values of player 1.
\begin{equation*}
\begin{IEEEeqnarraybox}[
\IEEEeqnarraystrutmode
\IEEEeqnarraystrutsizeadd{2pt}
{0.5pt}
][c]{c/v/c/v/c/V/c/v/c/v/c/v/c}
 && E_1\mysetminus S && S && A  && B  && C  && D \\\hline
 v_1^{\textsc{iv}} && \text{---}\ m^4 \ \text{---} &&  \text{---}\ m\ \text{---} && \text{---}\ m^3\ \text{---}  && \text{---}\ m^2\ \text{---}&&\text{---}\ m^2\ \text{---}&&\text{---}\ 1\ \text{---}
\end{IEEEeqnarraybox}
\end{equation*}
Notice that player 1 must get $E_1 \mysetminus S$ since he strongly desires it. The same goes for player 2 and $A \cup D $. Given that, an exchange  takes place and  in this exchange player 1 must get the whole $B \cup C= T$ (otherwise he would deviate to $\mathbf{v}^{\textsc{v}}$ and  strictly improve). On the other hand, player 2 must get the whole $S$, or player 1 would deviate from $\mathbf{v}^{\textsc{v}}$ to $\mathbf{v}^{\textsc{vi}}$ and  strictly improve. Thus the allocation here remains the same: $X_1(\mathbf{v}^{\textsc{vii}})= (E_1 \mysetminus S) \cup B \cup C$, $X_2(\mathbf{v}^{\textsc{vii}})=  S \cup A \cup D$. 

Next we move to profile  $\mathbf{v}^{\textsc{viii}}=(v_1^{\textsc{iv}}, v_2^{\textsc{vi}})$ by changing the values of player 2.
\begin{equation*}
\begin{IEEEeqnarraybox}[
\IEEEeqnarraystrutmode
\IEEEeqnarraystrutsizeadd{2pt}
{0.5pt}
][c]{c/v/c/v/c/V/c/v/c/v/c/v/c}
 && E_1\mysetminus S && S && A  && B  && C  && D \\\hline
 v_2^{\textsc{vi}} && \text{---}\ 1 \ \text{---} &&  \text{---}\ m^2\ \text{---} && \text{---}\ m\ \text{---}  && \text{---}\ 1\ \text{---}&&\text{---}\ 1\ \text{---}&&\text{---}\ m^3\ \text{---}
\end{IEEEeqnarraybox}
\end{equation*} 
By truthfulness, the allocation remains the same, i.e.,
$X_1(\mathbf{v}^{\textsc{viii}})= (E_1 \mysetminus S) \cup B \cup C$, $X_2(\mathbf{v}^{\textsc{viii}})=  S \cup A \cup D$. 

Finally we move to profile  $\mathbf{v}^{\textsc{ix}}=(v_1^{\textsc{iii}}, v_2^{\textsc{vi}})$ by changing the values of player 1 back to what he had in profile $\mathbf{v}^{\textsc{v}}$. Recall:
\begin{equation*}
\begin{IEEEeqnarraybox}[
\IEEEeqnarraystrutmode
\IEEEeqnarraystrutsizeadd{2pt}
{0.5pt}
][c]{c/v/c/v/c/V/c/v/c/v/c/v/c}
 && E_1\mysetminus S && S && A  && B  && C  && D \\\hline
 v_1^{\textsc{iii}} && \text{---}\ m^3\ \text{---} &&  \text{---}\ m\ \text{---} && \text{---}\ m^2\ \text{---}  && \text{---}\ 1\ \text{---}&&\text{---}\ 1\ \text{---}&&\text{---}\ 1\ \text{---}
\end{IEEEeqnarraybox}
\end{equation*}
Notice that player 1 must get $E_1 \mysetminus S$ since he strongly desires it. The same goes for player 2 and $D$.
Now if player 1 gets nothing from set $A$ then there is no exchange at all. However, in this case player 2 would deviate to profile $\mathbf{v}^{\textsc{v}}$ and become strictly better. Thus player 1 should get at least one item from  $A$. As a result, however, player 1 would deviate from profile $\mathbf{v}^{\textsc{viii}}$ to $\mathbf{v}^{\textsc{ix}}$ and become strictly better, something that leads to contradiction. 

This completes the inductive step. 
\end{proof}
\medskip


\begin{proof}[\textbf{Proof of Lemma \ref{lem:change}}]

Let $S$ be an exchangeable subset of $E_1$. Then according to corollary \ref{col:combo} $S= \bigcup_{ i \in I}S_i \cup W$ for some $I\subseteq [k]$, with $|I| \geq 1$. We are going to show that $W= \emptyset$. So suppose, towards a   contradiction, that $W \neq \emptyset$. In fact, choose $S$ so that it is a minimal exchangeable non-umes subset of $E_1$, i.e., for all $S' \subsetneq S$, $S'$ is either umes or non-exchangeable. In addition, notice that $W$ does not contain any exchangeable sets. 

Let $T$ be such that $(S,T)$ is a feasible exchange. In fact let $T$ be a minimal such set, i.e., for all $T' \subsetneq T$, either $(S,T')$ is not a feasible exchange or $T'$ is not exchangeable at all. Finally, let $S^*= \bigcup_{ i \in I}S_i$, $T^*= \bigcup_{ i \in I}T_i$ and notice that $T\mysetminus T^* \neq \emptyset$ since otherwise $T$ would be an umes (as an exchangeable subset of an umes, by Lemma \ref{lem:inter}).

For the sake of readability, let $A=E_1\mysetminus S$, $B = T \mysetminus T^*$, $C = T^* \cap T$, $D = T^* \mysetminus T$, and $E=E_2 \mysetminus (T \cup T^*)$.

So there are two profiles, $\mathbf{v}^{\textsc{i}}=(v_1^{\textsc{i}}, v_2^{\textsc{i}})$ where $X_1(\mathbf{v}^{\textsc{i}})=A \cup B \cup C$, $X_2(\mathbf{v}^{\textsc{i}})=S \cup D \cup E$ and $\mathbf{v}^{\textsc{ii}}=(v_1^{\textsc{ii}}, v_2^{\textsc{ii}})$ where $X_1(\mathbf{v}^{\textsc{ii}})=A \cup W \cup C \cup D$ and $X_2(\mathbf{v}^{\textsc{ii}})=S^* \cup B \cup E$. 

We start from profile $\mathbf{v}^{\textsc{i}}=(v_1^{\textsc{i}}, v_2^{\textsc{i}})$ and we proceed to profile $\mathbf{v}^{\textsc{iii}}=(v_1^{\textsc{i}}, v_2^{\textsc{iii}})$ by changing the values of player 2: 
\begin{equation*}
\begin{IEEEeqnarraybox}[
\IEEEeqnarraystrutmode
\IEEEeqnarraystrutsizeadd{2pt}
{0.5pt}
][c]{c/v/c/v/c/v/c/V/c/v/c/v/c/v/c}
 && A&& S^* && W  && B  && C  && D && E \\\hline
 v_2^{\textsc{iii}} && \text{---}\ 1\ \text{---} &&  \text{---}\ m^2\ \text{---} &&\text{---}\ m^2\ \text{---}&& \text{---}\ 1\ \text{---}  && \text{---}\ 1\ \text{---}&&\text{---}\ m^3\ \text{---}&&\text{---}\ m^3\ \text{---}
\end{IEEEeqnarraybox}
\end{equation*}
Since player's 2 most valuable items are those which he was allocated in profile $\mathbf{v}^{\textsc{i}}$, by truthfulness, the allocation remains the same, i.e., $X_1(\mathbf{v}^{\textsc{iii}})=A \cup B \cup C$, $X_2(\mathbf{v}^{\textsc{iii}})=S \cup D \cup E$.  

Next we move to profile  $\mathbf{v}^{\textsc{iv}}=(v_1^{\textsc{iii}}, v_2^{\textsc{iii}})$ by changing the values of player 1:
\begin{equation*}
\begin{IEEEeqnarraybox}[
\IEEEeqnarraystrutmode
\IEEEeqnarraystrutsizeadd{2pt}
{0.5pt}
][c]{c/v/c/v/c/v/c/V/c/v/c/v/c/v/c}
 && A&& S^* && W  && B  && C  && D && E \\\hline
 v_1^{\textsc{iii}} && \text{---}\ m^3\ \text{---} &&  \text{---}\ m\ \text{---} &&\text{---}\ m\ \text{---} && \text{---}\ m^2\ \text{---}  && \text{---}\ 1\ \text{---}&&\text{---}\ 1\ \text{---}&&\text{---}\ 1\ \text{---}
\end{IEEEeqnarraybox}
\end{equation*}
Now notice that player 1 must get $A$ since he strongly desires it. The same goes for player 2 and  $D \cup E $. Also we know that an exchange should take place and that in this exchange player 1 must get at least  $B =T\mysetminus T^*$ (otherwise he would deviate to $\mathbf{v}^{\textsc{iii}}$). 

We claim that player 1 gets the whole $T$. If this was not the case then he would get some set $R \supseteq B \neq \emptyset$. Since $R \subsetneq T$ and $R \nsubseteq T^*$ we have that the exchange $(S,R)$ is not feasible due to the fact that $T$ is minimal. Thus $R$ should be exchanged with some $\hat{S} \subsetneq S$. However, by the minimality of $S$,  $\hat{S}$ is an umes and  it is exchangeable only with strict subsets of $T^*$, which is a contradiction. On the other hand, player 2 must get the whole $S$, or otherwise he would get some $\hat{S} \subsetneq S$  which is  exchangeable only with strict subsets of $T^*$, something that can not happen. Thus the allocation is $X_1(\mathbf{v}^{\textsc{iv}})=A \cup B \cup C$ and $X_2(\mathbf{v}^{\textsc{iv}})=S \cup D \cup E$.

Next we move to profile  $\mathbf{v}^{\textsc{v}}=(v_1^{\textsc{iii}}, v_2^{\textsc{iv}})$ by changing the values of player 2:
\begin{equation*}
\begin{IEEEeqnarraybox}[
\IEEEeqnarraystrutmode
\IEEEeqnarraystrutsizeadd{2pt}
{0.5pt}
][c]{c/v/c/v/c/v/c/V/c/v/c/v/c/v/c}
 && A&& S^* && W  && B  && C  && D && E \\\hline
 v_2^{\textsc{iii}} && \text{---}\ 1\ \text{---} &&  \text{---}\ m^2\ \text{---} &&\text{---}\ m\ \text{---}&& \text{---}\ 1\ \text{---}  && \text{---}\ 1\ \text{---}&&\text{---}\ m^2\ \text{---}&&\text{---}\ m^3\ \text{---}
\end{IEEEeqnarraybox}
\end{equation*}
Since player's 2 most valuable items are those which he was allocated in profile $\mathbf{v}^{\textsc{iv}}$, by the truthfulness of the mechanism,  he must continue to get them but he can not get any other item. Thus the allocation remains the same, i.e., $X_1(\mathbf{v}^{\textsc{v}})=A \cup B \cup C$, $X_2(\mathbf{v}^{\textsc{v}})=S \cup D \cup E$.

Now let us return to profile $\mathbf{v}^{\textsc{ii}}=(v_1^{\textsc{ii}}, v_2^{\textsc{ii}})$. Starting from this profile we change the values of player 2 and get profile $\mathbf{v}^{\textsc{vi}}=(v_1^{\textsc{ii}}, v_2^{\textsc{iv}})$:
\begin{equation*}
\begin{IEEEeqnarraybox}[
\IEEEeqnarraystrutmode
\IEEEeqnarraystrutsizeadd{2pt}
{0.5pt}
][c]{c/v/c/v/c/v/c/V/c/v/c/v/c/v/c}
 && A&& S^* && W  && B  && C  && D && E \\\hline
 v_2^{\textsc{iv}} && \text{---}\ 1\ \text{---} &&  \text{---}\ m^2\ \text{---} &&\text{---}\ m\ \text{---}&& \text{---}\ m^4\ \text{---}  && \text{---}\ 1\ \text{---}&&\text{---}\ 1\ \text{---}&&\text{---}\ m^4\ \text{---}
\end{IEEEeqnarraybox}
\end{equation*}
Again, player's 2 most valuable items are those which he was allocated in profile $\mathbf{v}^{\textsc{ii}}$. So, by truthfulness, the allocation remains the same, i.e., $X_1(\mathbf{v}^{\textsc{vi}})=A \cup W \cup C \cup D$ and $X_2(\mathbf{v}^{\textsc{vi}})=S^* \cup B \cup E$. 

Next we move to profile  $\mathbf{v}^{\textsc{vii}}=(v_1^{\textsc{iv}}, v_2^{\textsc{iv}})$ by changing the values of player 1:
\begin{equation*}
\begin{IEEEeqnarraybox}[
\IEEEeqnarraystrutmode
\IEEEeqnarraystrutsizeadd{2pt}
{0.5pt}
][c]{c/v/c/v/c/v/c/V/c/v/c/v/c/v/c}
 && A&& S^* && W  && B  && C  && D && E \\\hline
 v_1^{\textsc{iv}} && \text{---}\ m^5\ \text{---} &&  \text{---}\ m\ \text{---} &&\text{---}\ m^2\ \text{---} && \text{---}\ m^4\ \text{---}  && \text{---}\ m^3\ \text{---}&&\text{---}\ m^3\ \text{---}&&\text{---}\ 1\ \text{---}
\end{IEEEeqnarraybox}
\end{equation*}
Notice that player 1 must get $A$ and player 2 must get $B \cup E $. Given that, player 1 must get  $W \cup C \cup D=$ (otherwise he could deviate to $\mathbf{v}^{\textsc{v}}$ and strictly improve). Thus the allocation  remains the same, i.e.,  $X_1(\mathbf{v}^{\textsc{vii}})=A \cup W \cup C \cup D$ and $X_2(\mathbf{v}^{\textsc{vii}})=S^* \cup B \cup E$.

Next we move to profile  $\mathbf{v}^{\textsc{viii}}=(v_1^{\textsc{iv}}, v_2^{\textsc{v}})$ by changing the values of player 2:
\begin{equation*}
\begin{IEEEeqnarraybox}[
\IEEEeqnarraystrutmode
\IEEEeqnarraystrutsizeadd{2pt}
{0.5pt}
][c]{c/v/c/v/c/v/c/V/c/v/c/v/c/v/c}
 && A&& S^* && W  && B  && C  && D && E \\\hline
 v_2^{\textsc{v}} && \text{---}\ 1\ \text{---} &&  \text{---}\ m^3\ \text{---} &&\text{---}\ m\ \text{---}&& \text{---}\ m^2\ \text{---}  && \text{---}\ 1\ \text{---}&&\text{---}\ 1\ \text{---}&&\text{---}\ m^4\ \text{---}
\end{IEEEeqnarraybox}
\end{equation*}
Again, by truthfulness, the allocation remains the same, i.e.,  $X_1(\mathbf{v}^{\textsc{viii}})=A \cup W \cup C \cup D$ and $X_2(\mathbf{v}^{\textsc{viii}})=S^* \cup B \cup E$.
 
Finally we move to profile  $\mathbf{v}^{\textsc{ix}}=(v_1^{\textsc{iii}}, v_2^{\textsc{v}})$ by changing the values of player 1 back to what he had in profile $\mathbf{v}^{\textsc{v}}$.
\begin{equation*}
\begin{IEEEeqnarraybox}[
\IEEEeqnarraystrutmode
\IEEEeqnarraystrutsizeadd{2pt}
{0.5pt}
][c]{c/v/c/v/c/v/c/V/c/v/c/v/c/v/c}
 && A&& S^* && W  && B  && C  && D && E \\\hline
 v_1^{\textsc{iii}} && \text{---}\ m^3\ \text{---} &&  \text{---}\ m\ \text{---} &&\text{---}\ m\ \text{---} && \text{---}\ m^2\ \text{---}  && \text{---}\ 1\ \text{---}&&\text{---}\ 1\ \text{---}&&\text{---}\ 1\ \text{---}
\end{IEEEeqnarraybox}
\end{equation*}
Player 1 must get $A$ and player 2 must get $ E $.
Now if player 1 gets nothing from  $B$ then there will be no exchange. However, in this case player 2 would deviate to profile $\mathbf{v}^{\textsc{v}}$ and become strictly better. Thus player 1 should get at least one item from  $B$. As a result, player 1 would deviate from profile $\mathbf{v}^{\textsc{viii}}$ to $\mathbf{v}^{\textsc{ix}}$ and  strictly improve, something that leads to contradiction.
\end{proof}
\medskip


\begin{proof}[\textbf{Proof of Lemma \ref{lem:favor}}]
Without loss of generality, assume that $(S_1,T_1), ...,(S_r,T_r)$ is the set of all favorable exchanges. Then $(S,T)$ where $S=\bigcup_{i \in [r]}S_i$ and $T=\bigcup_{i \in [r]}T_i$ will give a lower bound on the value of each player. 
Indeed, $S$ is an umes an using Lemmata \ref{lem:values} and \ref{lem:good}, we have that player 1 should gain at least $v_1(T)-v_1(S)$, while player 2  should gain at least $v_2(S)-v_2(T)$ from the exchanges.

Since $\mathbf{v}\in \mathcal{V}^{\neq}_{\! m}$, it suffices to show that $v_1(X_1(\mathbf{v}))=v_1((E_1 \cup T)\mysetminus S) = v_1(E_1)+v_1(T)-v_1(S)$. So suppose that  $v_1(X_1(\mathbf{v}))>v_1(E_1)+v_1(T)-v_1(S)$ and notice that this also implies that $v_2(X_2(\mathbf{v}))>v_2(E_2)+v_2(S)-v_2(T)$, since otherwise it would be $v_2(X_2(\mathbf{v})) = v_2(E_2)+v_2(S)-v_2(T)$ and we have  the desired allocation. 

As a result, there exists some $S^* \subseteq X_2^{E_1}(\mathbf{v})$, such that $S^*$ is an umes but $(S^*, T^*)$---where $T^*$ is the ``pair'' of $S^*$---is unfavorable. Without loss of generality, we may assume that $v_1(T^*)<v_1(S^*)$. 
Now let $S'$ to be the union of all minimally exchangeable sets $S_j \subseteq X_2^{E_1}(\mathbf{v})$ such that $v_1(T_j) < v_1(S_j)$, and notice that  $S' \subsetneq X_2^{E_1}(\mathbf{v})$ and  $v_1(T')<v_1(S')$. 

Let $S^*=X_2^{E_1}(\mathbf{v})$ and  $T^*=X_1^{E_2}(\mathbf{v})$. 
We begin with profile $\mathbf{v}=(v_1, v_2)$ where the allocation is $X_1(\mathbf{v})=(E_1\mysetminus S^*) \cup T^*$ and  $X_2(\mathbf{v})=(E_2\mysetminus T^*) \cup S^*$ and we move to profile $\mathbf{v}'=(v_1, v'_2)$. 
\begin{equation*}
\begin{IEEEeqnarraybox}[
\IEEEeqnarraystrutmode
\IEEEeqnarraystrutsizeadd{2pt}
{0.5pt}
][c]{c/v/c/v/c/V/c/v/c}
 &&   E_1\mysetminus S^* && S^* && T^* && E_2\mysetminus T^*\\\hline
 v'_2 &&\text{---}\ 1\ \text{---} &&\text{---}\ m\ \text{---}  && \text{---}\ 1\ \text{---}&&\text{---}\ m \text{---}
\end{IEEEeqnarraybox}
\end{equation*}
By truthfulness, the allocation remains the same, i.e., $X_1(\mathbf{v}')=(E_1\mysetminus S^*) \cup T^*$ and  $X_2(\mathbf{v}')=(E_2\mysetminus T^*) \cup S^*$. 

However, now notice that $(S^* \mysetminus S', T^* \mysetminus T')$ is a favorable exchange with respect to $\mathbf{v}'$. Moreover, for every minimally exchangeable set $S_i\subseteq S^* \mysetminus S'$ it holds that $(S_i, T_i)$ is favorable.  By using lemma \ref{lem:values} we have that the gain from the  exchange in $\mathbf{v}'$ for player 1 must be at least $v_1(T^* \mysetminus T')-v_1(S^* \mysetminus S')>v_1(T^*)-v_1(S^*)$ so we arrive at a contradiction.
\end{proof}
\medskip


\begin{proof}[\textbf{Proof of Lemma \ref{lem:final}}]
Let $\mathbf{v}=(v_1, v_2)$ be a profile in $\mathcal{V}_{\! m}$.   
By Lemmata \ref{strict} and \ref{lem:change}, we know that  $X_1^{E_1\cup E_2}(\mathbf{v})$  is the result of some exchanges of $D$ taking place, i.e., $X_1^{E_1\cup E_2}(\mathbf{v})= \left( E_1\mysetminus \bigcup_{i\in I} S_i\right)\cup \bigcup_{i\in I} T_i$, where $I\subseteq [k]$. There are two things that can go wrong: either there exists some $x\in I$ such that $(S_x, T_x)$ is unfavorable, or  there exists some $z\in [k]\mysetminus I$ such that $(S_z, T_z)$ is favorable. We first examine the former case.

Without loss of generality, we may assume that $v_1(T_x) < v_1(S_x)$. Consider the profile $\mathbf{v}'=(v_1, v_2^{\textsc{i}})$ where 
\begin{equation*}
v_{2j}^{\textsc{i}} =
\begin{cases}
m + 2^{j-m-1}  & \text{if } j\in X_2(\mathbf{v})\\
1 + 2^{j-m-1} & \text{otherwise } 
\end{cases}
\end{equation*}
By truthfulness, $X_2(\mathbf{v}')=X_2(\mathbf{v})$. Note also that $v_2^{\textsc{i}}$ induces for player 2 a strict preference over all subsets (see also Remark \ref{rem:notequal}). Moreover, with respect to $v_2^{\textsc{i}}$ the set of ``good'' minimal exchanges is exactly $\{(S_i, T_i) \ | \ i\in I\}$.

We now claim that player 1 can deviate and strictly improve his utility, thus contradicting truthfulness. In particular, consider the profile $\mathbf{v}''=(v_1^{\textsc{i}}, v_2^{\textsc{i}})$ where 
\begin{equation*}
v_{1j}^{\textsc{i}} =
\begin{cases}
m + 2^{j-m-1}  & \text{if } j\in (X_1(\mathbf{v}')\cup S_x)\mysetminus T_x\\
1 + 2^{j-m-1} & \text{otherwise } 
\end{cases}
\end{equation*}
Again, $v_1^{\textsc{i}}$ induces for player 1 a strict preference over all subsets, and thus $\mathbf{v}''\in  \mathcal{V}^{\neq}_{\! m}$. 
As a result, $\argmax_{S \in \mathcal{O}_i} v_i^{\textsc{i}}(S)$ only contains $X_i^{N_i}(\mathbf{v})$, for $i\in\{1, 2\}$, and by Lemma \ref{best} we have $X_1^{N_1\cup N_2}(\mathbf{v}'')=X_1^{N_1\cup N_2}(\mathbf{v}')$.
Additionally, notice that with respect to $\mathbf{v}''$ the set of favorable minimal exchanges is $\{(S_i, T_i) \ | \ i\in I\mysetminus\{x\}\}$. So, by Lemma \ref{lem:favor} we have $X_1^{E_1\cup E_2}(\mathbf{v}'')= \left( E_1\mysetminus \bigcup_{i\in I\mysetminus\{x\}} S_i\right)\cup \bigcup_{i\in I\mysetminus\{x\}} T_i = (X_1^{E_1\cup E_2}(\mathbf{v}')\cup S_x) \mysetminus T_x$. 

So, by deviating from $\mathbf{v}'$ to $\mathbf{v}''$, player 1 improves his utility by $v_1(S_x) - v_1(T_x) > 0$, which contradicts truthfulness. We conclude that there is no  $x\in I$ such that $(S_x, T_x)$ is unfavorable with respect to $\mathbf{v}$.

Next, we move on to the second case, i.e., there exists some $z\in [k]\mysetminus I$ such that $(S_z, T_z)$ is favorable with respect to $\mathbf{v}$. 
Like in the first case, the valuation functions that we define induce strict preferences over all subsets.
Consider the profile $Q=(v_1^{\textsc{ii}}, v_2)$ where 
\begin{equation*}
v_{1j}^{\textsc{ii}} =
\begin{cases}
m^2 + 2^{j-m-1}  & \text{if } j\in X_1(\mathbf{v})\mysetminus S_z\\
m + 2^{j-m-1}  & \text{if } j\in T_z\\
1 + 2^{j-m-1} & \text{otherwise } 
\end{cases}
\end{equation*}
We know, by Lemmata \ref{strict} and \ref{lem:change}, that $X_1^{E_1\cup E_2}(Q)= \left( E_1\mysetminus \bigcup_{i\in J} S_i\right)\cup \bigcup_{i\in J} T_i$ for some $J\subseteq [k]$.
By truthfulness, $X_1^{N_1\cup N_2}(Q)\supseteq X_1^{N_1\cup N_2}(\mathbf{v})$. In fact, by Lemma \ref{best}, it must be the case where $X_1^{N_1\cup N_2}(Q)= X_1^{N_1\cup N_2}(\mathbf{v})$. 
Again by truthfulness, $X_1^{E_1}(Q)\supseteq X_1^{E_1}(\mathbf{v}) \mysetminus S_z = E_1 \mysetminus \bigcup_{i\in I\cup \{z\}} S_i$ 
and $X_1^{E_2}(Q)\supseteq X_1^{E_2}(\mathbf{v})  = \bigcup_{i\in I\cup \{z\}} T_i$. 
This implies that $I\subseteq J\subseteq I\cup \{z\}$. 
If $J= I\cup \{z\}$, then player 1, by deviating from $\mathbf{v}$ to $Q$, improves his utility by $v_1(T_z) - v_1(S_z) > 0$, which contradicts truthfulness. So, it must be the case where $J = I$. 


Now, consider the profile $Q'=(v_1^{\textsc{ii}}, v_2^{\textsc{ii}})\in  \mathcal{V}^{\neq}_{\! m}$ where 
\begin{equation*}
v_{2j}^{\textsc{ii}} =
\begin{cases}
m + 2^{j-m-1}  & \text{if } j\in X_2^{N_1\cup N_2}(Q)\cup \bigcup_{i\in I\cup \{z\}} S_i \cup \bigcup_{i\notin I\cup \{z\}} T_i\\
1 + 2^{j-m-1} & \text{otherwise } 
\end{cases}
\end{equation*}
Since, for $i\in\{1, 2\}$, $\argmax_{S \in \mathcal{O}_i} v_i^{\textsc{ii}}(S)$ only contains $X_i^{N_i}(\mathbf{v})$, by Lemma \ref{best} we have $X_2^{N_1\cup N_2}(Q')=X_2^{N_1\cup N_2}(Q)$.
Further, the set of favorable minimal exchanges  with respect to $Q'$  is $\{(S_i, T_i) \ | \ i\in I\cup\{z\}\}$. So, by Lemma \ref{lem:favor} we have $X_2^{E_1\cup E_2}(Q')= \left( E_2\mysetminus \bigcup_{i\in I\cup\{z\}} T_i\right)\cup \bigcup_{i\in I\cup\{z\}} S_i$. 

So, by deviating from $Q$ to $Q'$, player 2 improves his utility by $v_2(S_z) - v_2(T_z) > 0$, which contradicts truthfulness. Therefore, there is no  $z\in [k]\mysetminus I$ such that $(S_z, T_z)$ is favorable with respect to $\mathbf{v}$, and this concludes the proof.
\end{proof}


\section{{Missing Material from Section \ref{SEC:FAIRNESS}}} \label{app:fairness}

\begin{proof}[\textbf{Proof of Lemma \ref{lem:fairness_mechs}}]
From theorem \ref{thm:truthful-->p/e_mech} we know that every truthful mechanism can be implemented as a picking-exchange mechanism. So consider such a mechanism and let us examine the structure of sets $N_i, E_i$ and $\mathcal{O}_i$. 
Notice that by the definition of picking-exchange mechanisms, each player $i$ controls $N_i \cup E_i$.
If both $N_i, E_i$ are nonempty, or $|E_i| > 1$, or $\mathcal{O}_i$ contains a non-singleton set, then  the respective player has control over some pair of items. Thus we can conclude that every possible mechanism can be implemented as a singleton picking-exchange mechanism. 
\end{proof}
\medskip


\begin{remark}\label{rem:singleton}
Regarding the remaining proofs, it suffices to focus only on singleton picking-exchange  mechanisms. Indeed, by Theorem \ref{thm:truthful-->p/e_mech} we know that every truthful mechanism can be implemented as a picking-exchange mechanism, and by Lemmata \ref{lem:control-2sets} and \ref{lem:fairness_mechs} only the singleton picking-exchange  mechanisms among them may achieve some fairness guarantee.
\end{remark}
\medskip


\begin{proof}[\textbf{Proof of Application \ref{cor:envy}}]
Initially it is easy to see that when $m=1$ or $m=2$, the statement holds in a trivial way for every singleton picking-exchange mechanisms. Indeed, in every instance each player gets at most one item and thus the value a player derives in the worst case is greater or equal to the value of the empty set (bundle of the other player minus an item). 

In the  case of $m=3$, in any instance one player gets one item and the other player two items. The singleton picking-exchange mechanism guarantees that the player who gets one item is allocated with at least his second best in terms of value, so the value he derives is always greater or equal to the value of his least desirable item (bundle of the other player minus an item). On the other hand, the player who is allocated with two items always derives value greater or equal to the value of the empty set.

Finally, in the case of  $m = 4$ with  $|N_1|=|N_2|=2$,  every player gets two items at every instance. The singleton picking-exchange mechanism guarantees that each player will receive at least his second best item, the value of which is greater or equal to the value of his third or fourth best item.

On the other hand, in case of $m\geq 5$ consider  profile $v_1=[1+\epsilon,1,...,1] \cup [1,\delta,...,\delta]$, \ \ $v_2=[1,\delta,...,\delta] \cup [1+\epsilon,1,...,1]$ where $1\gg \epsilon \gg \delta > 0$. The first vector of values is for $N_1$ (or $E_1$) and the second is for $N_2$ (or $E_1$); notice that it is possible for one of them to be empty. We only examine singleton picking-exchange mechanisms (see Remark \ref{rem:singleton}). It is easy to see that in such a case, by  the pigeonhole principle,  no  singleton picking-exchange mechanism can achieve envy-freeness up to one item for both players. 
\end{proof}
\medskip


\begin{proof}[\textbf{Proof of Application \ref{cor:mms}}]
We only need to prove that among all the singleton picking-exchange mechanisms (see Remark \ref{rem:singleton}) there is no better approximation ratio than ${\left\lfloor m/2 \right\rfloor}^{-1}$ for $m\ge 3$. Consider profile $v_1=[1+\epsilon,1,...,1] \cup [|N_1\cup E_1|,\delta,...,\delta]$, \ \  $v_2=[|N_2 \cup E_2|,\delta,...,\delta] \cup [1+\epsilon, 1, ...,1]$, where $1\gg \epsilon \gg \delta > 0$. The first vector of values is for $N_1$ (or $E_1$) and the second is for $N_2$ (or $E_1$); notice that it is possible for one of them to be empty. 

It is easy to see that when both $N_1\cup E_1$, $N_2\cup E_2$ are nonempty, then $\mms_i=|N_i\cup E_i|$ while they both receive value that is slightly greater than 1. Therefore, no  singleton picking-exchange mechanism can achieve a better approximation ratio than ${\left\lceil m/2 \right\rceil}^{-1}$  for both players. 

On the other hand, if $N_1\cup E_1=\emptyset$ (the other case is symmetric) then this is the mechanism in \cite{ABM16} that achieves exactly ${\left\lfloor m/2 \right\rfloor}^{-1}$. 
\end{proof}
\medskip



%

\end{document}